\documentclass[global,twocolumn]{svjour}
\journalname{Computational Mechanics}
\usepackage{mathtools}
\usepackage{amsmath} 
\usepackage{amsfonts}
\usepackage{latexsym}
\usepackage{siunitx}
\usepackage{isomath}
\usepackage{float}
\usepackage{bm}
\usepackage[dvipsnames]{xcolor}
\usepackage{booktabs}
\usepackage{breakcites}
\usepackage{float}
\usepackage[T1]{fontenc}
\usepackage[utf8]{inputenc}
\usepackage{afterpage}
\usepackage{graphicx}
\graphicspath{{figures/}}
\usepackage[backend=biber, maxnames=2, maxbibnames=99, style=numeric,
            giveninits=true, url=false, isbn=false, doi=false, date=year,
            sorting=anyvt]{biblatex}
\DeclareFieldFormat*{title}{\mkbibemph{#1}}
\DeclareFieldFormat*{citetitle}{\mkbibemph{#1}}
\DeclareFieldFormat{journaltitle}{#1}

\renewbibmacro*{in:}{%
  \ifentrytype{article}
    {}
    {\printtext{\bibstring{in}\intitlepunct}}}

\newbibmacro*{pubinstorg+location+date}[1]{%
  \printlist{#1}%
  \newunit%
  \printlist{location}%
  \newunit%
  \usebibmacro{date}%
  \newunit}

\renewbibmacro*{publisher+location+date}{\usebibmacro{pubinstorg+location+date}{publisher}}
\renewbibmacro*{institution+location+date}{\usebibmacro{pubinstorg+location+date}{institution}}
\renewbibmacro*{organization+location+date}{\usebibmacro{pubinstorg+location+date}{organization}}

\newcommand{\reviewerOne}[1]{{#1}}
\newcommand{\reviewerTwo}[1]{{#1}}
\newcommand{\generalChange}[1]{{#1}}

\usepackage{hyperref}

\newcommand{\tensor}[1]{\boldsymbol{#1}}             
\renewcommand{\vec}[1]{{\MakeLowercase{\mathbf{#1}}}}  
\newcommand{\Rvec}[1]{\underline{#1}}                
\newcommand{\Xvec}[1]{\boldsymbol{\mathbf{#1}}} 
\newcommand{\grad}{\operatorname{grad}}      
\newcommand{\Grad}{\operatorname{Grad}}      
\renewcommand{\div}{\operatorname{div}}        

\newcommand{\dd}{\,\mathrm{d}}

\newcommand{\Lm}[2]{{\mathrm{L}^{#1} (#2)}}

\newcommand{\dx}{\,\mathrm{d}\vec{x}}
\newcommand{\dX}{\,\mathrm{d}\Xvec{X}}

\newcommand{\dsX}{\,\mathrm{d}s_{\Xvec{X}}}

\newcommand{\Hm}[2]{{\mathrm{H}^{#1} (#2)}}

\newcommand{\normHm}[3]{{\left\|{#1}\right\|_{\Hm{#2}{#3}}}}

\newcommand{\norm}[2]{{\left\lVert{#1}\right\rVert}_{#2}}
\newcommand{\seminorm}[2]{{\lvert{#1}\rvert}_{#2}}
\newcommand{\refnormal}{\vec{n}_{0}}
\newcommand{\symotimes}{\overset{\mathrm{S}}{\otimes}}
\begin{document}
%
\title{\reviewerTwo{Versatile stabilized finite element formulations for nearly
and fully incompressible solid mechanics}}
\author{Elias Karabelas\inst{1,}\inst{2} \and
        Gundolf Haase\inst{1,}\inst{3} \and
        Gernot Plank\inst{2,}\inst{3} and
        Christoph M. Augustin\inst{2,4,}%
        \thanks{\emph{Present address: }%
        christoph.augustin@medunigraz.at
        Gottfried Schatz Research Center: Division of Biophysics,
        Medical University of Graz,
        Neue Stiftingtalstraße 6 (MC1.D.)/IV, 8010 Graz, Austria}%
}                     
\institute{
  Institute for Mathematics and Scientific Computing,
  NAWI Graz,
  University of Graz, Graz, Austria
  \and
  Gottfried Schatz Research Center: Division of Biophysics,
  Medical University of Graz, Graz, Austria
  \and
  BioTechMed-Graz, Graz, Austria
  \and
  Department of Mechanical Engineering, University of California, Berkeley, Berkeley, CA, United States
}
\date{Received: date / Revised version: date}
%
\maketitle
\begin{abstract}
  Computational formulations for large strain, polyconvex, nearly
  incompressible elasticity have been extensively studied,
  but research on enhancing solution schemes that offer better tradeoffs
  between accuracy, robustness, and computational efficiency
  remains to be highly relevant.

  In this paper, we present two methods to overcome locking phenomena,
  one based on a displacement-pressure formulation
  using a stable finite element pairing with bubble functions,
  and another one using a simple pressure-projection stabilized
  \(\mathbb P_1 - \mathbb P_1\)
  finite element pair.
  A key advantage is the versatility of the proposed methods: with minor
  adjustments they are applicable to all kinds of finite elements and generalize
  easily to transient dynamics.
  The proposed methods are compared to and verified with standard benchmarks
  previously reported in the literature.
  Benchmark results demonstrate that both approaches provide a robust and computationally efficient way of simulating nearly and fully incompressible materials.
\end{abstract}
%
\section{Introduction}
%
Locking phenomena, caused by ill-conditioned global stiffness matrices in
finite element analyses, are an often observed and extensively studied issue
when modeling nearly incompressible, hyperelastic materials~\cite{babuska1992locking,
braess2007finite, hughes1987finite, wriggers2008nonlinear, zienkiewicz2000finite}.
Typically, methods based on Lagrange multipliers are applied to enforce
incompressibility. A common approach is the split of the deformation gradient
into a volumetric and an isochoric part~\cite{flory1961thermodynamic}.
Here, locking commonly arises when unstable standard displacement formulations
are used that rely on linear shape functions to approximate the displacement
field \(\vec{u}\) and piecewise-constant finite elements combined with static
condensation of the hydrostatic pressure \(p\),
e.g., \(\mathbb P_1 - \mathbb P_0\) elements.
It is well known that in such cases solution algorithms may exhibit very low
convergence rates and that variables of interest such as stresses can be
inaccurate~\cite{gultekin2018quasi}.

From mathematical theory it is well known that approximation spaces for the
primal variable \(\vec u\) and \(p\) have to be well chosen to fulfill
the Ladyzhenskaya-Babu\^ska--Brezzi (LBB) or \emph{inf-sup}
condition~\cite{babuska1973finite,brezzi1974existence, chapelle1993infsup}
to guarantee stability.
A classical stable approximation pair is the
Taylor--Hood element~\cite{taylor1973numerical},
however, this requires quadratic ansatz functions for the displacement part.
\reviewerOne{
  For certain types of problems higher order interpolations
  can improve efficiency as higher accuracy is already reached with coarser
  discretizations \cite{chamberland2010comparison, land2015verification}.
  In many applications though, where geometries are fitted to, e.g.,
  capture fine structural features, this is not beneficial due to a
  possible increase in degrees of freedom and consequently a higher
  computational burden.
  Also for coupled problems such as electromechanical or
  fluid-structure-interaction models high-resolution grids for mechanical
  problems are sometimes required when interpolations between grids are not
  desired \cite{augustin2016anatomically, karabelas2018towards}.
  As a remedy for these kind of applications quasi
  Taylor--Hood elements with an order of $\tfrac{3}{2}$ have been considered,
  see~\cite{quaglino2017quasi}, as well as equal order linear pairs of ansatz
  functions which has been a field of intensive research in the last decades,
  see \cite{auricchio2017mixed, hughes2017multiscale} and references therein.
Unfortunately, equal order pairings do not fulfill the LBB
conditions and hence a stabilization of the element is of crucial importance.
}
\generalChange{There is a significant body of literature devoted to }
stabilized finite elements for the Stokes and Navier--Stokes equations.
Many of those methods were extended to incompressible elasticity, amongst other approaches by Hughes, Franca,
Balestra, and collaborators~\cite{franca1988new, hughes1986new}.
Masud and co-authors followed an idea by means of variational multiscale (VMS)
methods~\cite{masud2013framework,masud2005stabilized,nakshatrala2007finite,xia2009stabilized},
a technique that was recently extended to dynamic problems
(D-VMS)~\cite{scovazzi2016simple,rossi2016implicit}.
Further stabilizations of equal order finite elements include orthogonal sub-scale
methods~\cite{chiumenti2015mixed,codina2000stabilization,lafontaine2015explicit, cervera2003mixed}
and methods based on pressure projections~\cite{dohrmann2004stabilized,zienkiewicz1998triangles}.
Different classes of methods to avoid locking for nearly incompressible elasticity
were conceived by introducing nonconforming finite elements such as the
Crouzeix--Raviart element~\cite{dipietro2014extension,falk1991nonconforming} 
and Discontinuous Galerkin methods~\cite{kabaria2015hybridizable,teneyck2006discontinuous}.
Enhanced strain formulations~\cite{Reese1998,taylor2000mixed} have been considered as well as
formulations based on multi-field variational
principles~\cite{bonet2015computational,schroeder2016novel,schroeder2011new}.

In this study we introduce a novel variant of the MINI
element for accurately solving nearly and fully
incompressible elasticity problems.
The MINI element was originally established for computational fluid dynamics
problems~\cite{arnold1984stable} and pure tetrahedral meshes
\reviewerOne{and previously used in the large strain regime, e.g. in
\cite{chamberland2010comparison, lamichhane2009mixed}.}
We extend the MINI element definition for hexahedral meshes by
introducing two bubble functions in the element and provide \reviewerOne{a novel
proof of stability and well-posedness in the case of linear elasticity}.
The support of the bubble functions is restricted to the element and can thus
be eliminated from the system using static condensation.
\reviewerTwo{This also allows for a straightforward inclusion in
  combination with existing finite element codes since all
  required implementations are purely on the element level.}
\reviewerTwo{Additionally, we introduce a pressure-projection stabilization
method originally published for the
Stokes equations~\cite{bochev2006,dohrmann2004stabilized}
and previously used for large strain nearly incompressible elasticity in the
field of particle finite element methods and plasticity \cite{rodriguez2016,cante2014}.}
Due to its \emph{simplicity}, this type of stabilization is especially attractive
from an implementation point of view.

Robustness and performance of both the MINI element and the
pressure-projection approach are verified and compared to standard benchmarks
reported previously in literature.
A key advantage of the proposed methods is their \emph{high versatility}:
first, they are readily applicable to nearly and fully incompressible solid
mechanics; second, with little adjustments the stabilization techniques can be
applied to all kinds of finite elements, in this study we investigate the
performance for hexahedral and tetrahedral meshes;
and third, the methods generalize easily to transient dynamics.

Real world applications often require highly-resolved meshes and thus
efficient and massively parallel solution algorithms for the
linearized system of equations become an important factor to deal with the
resulting computational load. We solve the arising saddle-point systems
by using a GMRES method with a block preconditioner based on an
algebraic multigrid (AMG) approach. Extending our previous
implementations~\cite{augustin2016anatomically} we
performed the numerical simulations with the software
\textit{Cardiac Arrhythmia Research Package}
(CARP)~\cite{vigmond2008solvers} which relies on the MPI based library
\emph{PETSc}~\cite{petsc-user-ref} and the incorporated solver suite
\emph{hypre/BoomerAMG}~\cite{henson2002boomeramg}. The combination of these
advanced solving algorithms with the proposed stable elements which only rely on
linear shape functions proves to be \emph{very efficient} and renders feasible
simulations on grids with high structural detail.

The paper is outlined as follows:
Section~\ref{sec:methods} summarizes in brief the background on the methods.
In Section~\ref{sec:discr}, we introduce the finite element discretization
and discuss stability.
Subsequently, Section~\ref{sec:results} documents benchmark problems
where our proposed elements are applied and compared to results published in the
literature.
Finally, Section~\ref{sec:conclusion} concludes the paper with a discussion of
the results and a brief summary.
%
\section{Continuum mechanics}\label{sec:methods}
\subsection{Nearly incompressible nonlinear elasticity}
Let \(\Omega_0 \subset \mathbb R^3\) denote the reference configuration and
let \(\Omega_t \subset \mathbb R^3\) denote the current configuration
of the domain of interest. Assume that the boundary of \(\Omega_0\) is decomposed into
\(\partial \Omega_0 = \Gamma_{\mathrm{D},0} \cup \Gamma_{\mathrm{N},0}\)
with \(|\Gamma_{D,0}| > 0\).
\reviewerTwo{Here, $\Gamma_{\mathrm{D},0}$ describes the Dirichlet part of the boundary
and $\Gamma_{\mathrm{N},0}$ describes the Neumann part of the boundary, respectively.}
Further, let \(\refnormal\) be the unit outward normal on \(\partial \Omega_0\).
The nonlinear mapping \(\boldsymbol\phi\colon \Xvec X \in \Omega_0 \rightarrow \vec x \in \Omega_t\),
defined by \(\boldsymbol \phi:= \Xvec X + \vec u(\Xvec X,t)\), with displacement \(\vec u\),
maps points in the reference configuration to points in the current configuration.
Following standard notation we introduce the \emph{deformation gradient} \(\tensor F\)
and the Jacobian $J$ as
\[
  \tensor F := \Grad \boldsymbol \phi = \tensor I + \Grad \vec u,\quad
  J := \det(\tensor F),
\]
and the \emph{left Cauchy-Green tensor} as \(\tensor C := \tensor F^\top \tensor F\).
Here, \(\Grad(\bullet)\) denotes the gradient with respect to the reference
coordinates \(\Rvec X \in \Omega_0\). The displacement field \(\vec u\)
is sought as infimizer of the functional
\reviewerTwo{
\begin{align}
 \nonumber\Pi^\mathrm{tot}(\vec u) &:= \Pi(\vec u) + \Pi^\mathrm{ext}(\vec u),\\
 \nonumber\Pi(\vec u) &:= \int\limits_{\Omega_0} \Psi(\tensor F(\vec u))\dX,\\
 \label{eq:potential_ext}
 \Pi^{\mathrm{ext}}(\vec u)
   &:= - \rho_0\int\limits_{\Omega_0} \vec f(\vec X)  \cdot \vec u\dX
   - \int\limits_{\Gamma_\mathrm{N,0}} \vec h(\vec X) \cdot \vec u\dsX,
\end{align}
over all admissible fields $\vec u$ with $\vec u = \vec g_\mathrm{D}$ on
$\Gamma_\mathrm{D,0}$,} where, \(\Psi\) denotes the strain energy function;
\(\rho_0\) denotes the material density in reference configuration;
\(\vec f\) denotes a volumetric body force;
\reviewerTwo{\(\vec g_\mathrm{D}\) denotes a given boundary displacement;}
and \(\vec h\) denotes a given surface traction.
For ease of presentation it is assumed that \(\rho_0\) is constant and
\(\vec f\), \reviewerTwo{\(\vec g_\mathrm{D}\)}, and \(\vec h\) do not depend on \(\vec u\).
Existence of infimizers is, under suitable assumptions,
guaranteed by the pioneering works of Ball, see~\cite{Ball1976}.


In this study we consider nearly incompressible materials,
meaning that \(J \approx 1\). A possibility to model this behavior was
originally proposed by~\textcite{flory1961thermodynamic}
using a split of the deformation gradient \(\tensor F\) such that
\begin{equation}
  \tensor F = \tensor F_\mathrm{vol} \overline{\tensor F}.
\end{equation}
Here, \(\tensor F_\mathrm{vol}\) describes the volumetric change
while \(\overline{\tensor F}\) describes the isochoric change.
By setting \(\tensor F_\mathrm{vol} := J^{\frac{1}{3}} \tensor I\) and
\(\overline{\tensor F} := J^{-\frac{1}{3}} \tensor F\)
we get \(\det(\overline{\tensor F}) = 1\) and \(\det(\tensor F_\mathrm{vol}) = J\).
Analogously, by setting \(\tensor C_\mathrm{vol} := J^{\frac{2}{3}} \tensor I\) and
\(\overline{\tensor C} := J^{-\frac{2}{3}} \tensor C\),
we have \(\tensor C = \tensor C_\mathrm{vol} \overline{\tensor C}\).
Assuming a hyperelastic material, the Flory split also postulates an additive
decomposition of the strain energy function
\begin{equation} \label{eq:flory_strain_energy}
  \Psi = \Psi(\tensor C) = \overline{\Psi}(\overline{\tensor C})
  + \kappa U(J),
\end{equation}
where \(\kappa\) is the \emph{bulk modulus}. The function \(U(J)\)
acts as a penalization of incompressibility and we require that it is
strictly convex and twice continuously differentiable. Additionally,
a constitutive model for \(U(J)\) should fulfill that (i) it vanishes in the reference
configuration and that (ii) an infinite amount of energy is required to shrink
the body to a point or to expand it indefinitely, i.e.,
\[
  \mbox{(i)}\ U(1) = 0,\
  \mbox{(ii)} \lim_{J\to 0+} U(J) = \infty,\ \lim_{J\to \infty} U(J) = \infty.
\]
For the remainder of this work we will focus on functions \(U(J)\) that can be written as
\begin{align*}
  U(J) := \frac{1}{2}{(\Theta(J))}^2.
\end{align*}
In the literature many different choices for the function \(\Theta(J)\) are
proposed, see, e.g.,~\cite{doll2000on,hartmann2003polyconvexity,rossi2016implicit}
for examples and related discussion.

As we also want to study the case of full incompressibility,
meaning \(\kappa \to \infty\), we need a reformulation of the system.
In this work we will use a perturbed Lagrange-multiplier functional,
see~\cite{sussman1987finite, atluri1989formulation, brink1996} for details,
and we introduce
\begin{align*}
  \Pi^\mathrm{PL}(\vec u, q) := \int\limits_{\Omega_0}\overline{\Psi}(\overline{\tensor C}(\vec u))
                             + q \Theta(J(\vec u)) - \frac{1}{2 \kappa} q^2\dX.
\end{align*}
We will now seek infimizers \((\vec u, p) \in V_{\vec g_\mathrm{D}} \times Q \)
of the modified functional
\begin{align}
  \label{eq:inf_prob}
  \Pi^{\mathrm{tot}}(\vec u, q)
    := \Pi^\mathrm{PL}(\vec u, q) + \Pi^{\mathrm{ext}}(\vec u).
\end{align}
To guarantee that the discretization of~\eqref{eq:inf_prob} is well defined,
we assume that
\begin{equation*}
  V_{\vec g_\mathrm{D}} = \{ \vec v \in {[H^1(\Omega_0)]}^3 :
    \left.\vec v\right|_{\Gamma_{\mathrm{D},0}} = \vec g_\mathrm{D}\},
\end{equation*}
with \(H^1(\Omega_0)\) being the standard Sobolev space consisting of all
square integrable functions with square integrable gradient, and \(Q=L^2(\Omega_0)\).
Existence of infimizers in \(V_{\vec g_\mathrm{D}}\) cannot be guaranteed in general.
However, assuming suitable growth conditions on the strain energy function
\(\Psi\), and assuming that the initial data keeps the material in the hyperelastic range,
one can conclude that the space \(V\) for the infimizer \(\vec u\) contains
\(V_{\vec g_\mathrm{D}}\) as a subset, see~\cite{Ball1976} for details.

To solve for the infimizers of~\eqref{eq:inf_prob} we require to compute
the variations of~\eqref{eq:inf_prob} \reviewerTwo{with respect to} \(\Delta\vec u\) and \(\Delta p\)
\begin{align}
  \nonumber
  D_{\Delta \vec v} \Pi^\mathrm{PL}(\vec u,p) &=
  \int\limits_{\Omega_0} \left(\tensor S_\mathrm{isc}
  + p \tensor S_\mathrm{vol}\right):\tensor \Sigma(\vec u, \Delta \vec v) \dX \\
  \label{eq:non_var_prob:1}
  & -\rho_0\int\limits_{\Omega_0} \vec f \cdot \Delta \vec v\dX
  - \int\limits_{\Gamma_{\mathrm{N},0}} \vec h \cdot \Delta \vec v\dsX,\\
  \label{eq:non_var_prob:2}
  D_{\Delta q} \Pi^\mathrm{PL}(\vec u,p)
  &= \int\limits_{\Omega_0} \left(\Theta(J) - \frac{1}{\kappa} p\right)\Delta q\dX,
\end{align}
with abbreviations as, e.g., in~\cite{holzapfel2000nonlinear}
\begin{align}
  \tensor S_\mathrm{isc} &:= J^{-\frac{2}{3}}
  \mathrm{Dev}(\overline{\tensor S}), \quad \text{where~}\overline{\tensor S}
  := \frac{\partial\overline{\Psi}(\overline{\tensor C})}{\partial \overline{\tensor C}}\\
  \tensor S_\mathrm{vol} &:= \pi(J) \tensor C^{-1}, \quad \text{with~} \pi(J)
  := J\Theta'(J),\\
  \tensor \Sigma(\vec u, \vec v)
  &:= \mathrm{sym}(\tensor F^\top(\vec u) \Grad \vec v).
\end{align}
Next, with notations
\begin{align}
  a_{\mathrm{isc}}(\vec u;\Delta \vec v)
    &:= \int\limits_{\Omega_0} \tensor S_\mathrm{isc}(\vec u)
    : \tensor \Sigma(\vec u,\Delta\vec v)\dX,\\
  a_{\mathrm{vol}}(\vec u, p;\Delta \vec v)
    &:= \int\limits_{\Omega_0} p \tensor S_\mathrm{vol}(\vec u)
    : \tensor \Sigma(\vec u,\Delta\vec v)\dX,\\
  b_{\mathrm{vol}}(\vec u; \Delta q)
    &:= \int\limits_{\Omega_0} \Theta(J(\vec u)) \Delta q\dX,\\
  c(p,\Delta q)
    &:= \frac{1}{\kappa}\int\limits_{\Omega_0}p \Delta q\dX,\\
  l_\mathrm{body}(\Delta \vec v)
    &:= \rho_0\int\limits_{\Omega_0} \vec f \cdot \Delta \vec v\dX,\\
  l_\mathrm{surface}(\Delta \vec v)
  &:= \int\limits_{\Gamma_{\mathrm{N},0}} \vec h \cdot \Delta \vec v \dsX,\\
  R_\mathrm{upper}(\vec u, p;\Delta \vec v)
    &:= a_{\mathrm{isc}}(\vec u;\Delta \vec v)
    + a_{\mathrm{vol}}(\vec u, p;\Delta \vec v) \nonumber\\
    &-l_\mathrm{body}(\Delta \vec v)-l_\mathrm{surface}(\Delta \vec v),\\
  \label{eq:def_r_lower}
  R_\mathrm{lower}(\vec u, p;\Delta q)
    &:= b_{\mathrm{vol}}(\vec u; \Delta q) - c(p,\Delta q),
\end{align}
we formulate the mixed boundary value problem of nearly incompressible
nonlinear elasticity via a nonlinear system of equations.
This yields a nonlinear saddle-point problem,
find \((\vec u, p) \in V_{\vec g_D} \times Q\) such that
\begin{align}
  \label{eq:non_lin_var_prob:1}R_\mathrm{upper}(\vec u, p;\Delta \vec v) &= 0,\\
  \label{eq:non_lin_var_prob:2}R_\mathrm{lower}(\vec u, p;\Delta q) &=0,
\end{align}
for all \((\Delta \vec v, \Delta q) \in V_{\vec 0} \times Q\).
%
\subsection{Consistent linearization}
To solve the nonlinear variational equations~\eqref{eq:non_lin_var_prob:1}--\eqref{eq:non_lin_var_prob:2},
with a finite element approach we first apply a Newton--Raphson scheme,
for details we refer to~\cite{deuflhard2011newton}.
Given a nonlinear and continuously differentiable operator \(F\colon X\to Y\)
a solution to \(F(x)=0\) can be approximated by
\begin{align*}
  x^{k+1} &= x^{k} + \Delta x, \\
  \left.\frac{\partial F}{\partial x}\right|_{x = x^k} \Delta x &= -F(x^k),
\end{align*}
which is looped until convergence.
In our case, we have \(X = V_{\vec g_D} \times Q\), \(Y = \mathbb R^2\),
\(\Delta x = {(\Delta \vec u, \Delta p)}^\top\), \(x^k = {(\vec u^k, p^k)}^\top\),
and \(F = {(R_\mathrm{upper}, R_\mathrm{lower})}^\top\).
We obtain the following linear saddle-point problem for each
\((\vec u^k, p^k) \in V_{\vec g_D} \times Q\),
find \((\Delta \vec u, \Delta p) \in V_{\vec 0} \times Q\) such that
\begin{align}
  \label{eq:saddle_point_nl:1}
  a_k(\Delta \vec u, \Delta \vec v) + b_k(\Delta p, \Delta \vec v)
  &= -R_\mathrm{upper}(\vec u^k, p^k; \Delta \vec v),\\
  \label{eq:saddle_point_nl:2}
  b_k(\Delta q, \Delta \vec u) - c(\Delta p, \Delta q)
  &= -R_\mathrm{lower}(\vec u^k, p^k;\Delta q),
\end{align}
where
\begin{align*}
  a_k(\Delta \vec u, \Delta \vec v)
    &:= \int\limits_{\Omega_0} \Grad \Delta \vec v \tensor S_{\mathrm{tot},k}
    : \Grad \Delta \vec{u}\dX\\
    &+ \int\limits_{\Omega_0}\tensor \Sigma(\vec u_k, \Delta \vec v)
    : \mathbb C_{\mathrm{tot},k} : \tensor \Sigma(\vec u_k, \Delta \vec u)\dX,\\
  b_k(\Delta p, \Delta \vec v)
    &:= \int\limits_{\Omega_0} \Delta p \pi(J_k) \tensor F^{-\top}_k
    : \Grad \Delta \vec v \dX,
\end{align*}
with abbreviations
\generalChange{
\begin{align}
  \nonumber \tensor F_k &:= \tensor F(\vec u_k),\\
  \nonumber J_k &:= \det(\tensor F_k),\\
  \label{eq:def_s_tot}\tensor S_{\mathrm{tot},k}
  &:= \left.\tensor S_{\mathrm{isc}}\right|_{\vec u =\vec u_k}
    + p_k \left. \tensor S_{\mathrm{vol}}\right|_{\vec u =\vec u_k}, \\
  \label{eq:def_c_tot}\mathbb C_{\mathrm{tot},k}
    &:= \left.\mathbb{C}_{\mathrm{isc}}\right|_{\vec u =\vec u_k}
    + p_k \left.\mathbb{C}_{\mathrm{vol}}\right|_{\vec{u} = \vec u_k},\\
  \nonumber\mathbb{C}_\mathrm{vol}
    &:= k(J) \tensor C^{-1} \otimes \tensor C^{-1}
    - 2 \pi(J) \tensor C^{-1} \odot \tensor C^{-1},\\
  \nonumber k(J) &:= J^2 \Theta''(J) + J \Theta'(J),
\end{align}
where $\mathbb{C}_{\mathrm{isc}}$ is given in \eqref{eq:def_c_isc}.
}
The derivation of the consistent linearization is lengthy but standard,
we refer to~\cite[Chapter 8]{holzapfel2000nonlinear} for details.
The definition of the higher order tensor and other abbreviations are given
in the Appendix.
%
\subsection{Review on solvability of the linearized problem}
Since~\eqref{eq:saddle_point_nl:1}--\eqref{eq:saddle_point_nl:2}
is a linear saddle-point problem for each given \((\vec u^k, p^k)\)
we can rely on the well-established Babu\v{s}ka--Brezzi theory,
see~\cite{boffi2013mixed,ern2013theory,ruter2000analysis,schwab1998}.
The crucial properties to guarantee that the
problem~\eqref{eq:saddle_point_nl:1}--\eqref{eq:saddle_point_nl:2} is well-posed
are continuity of all involved bilinear forms and the following three conditions:
\begin{itemize}
  \item[(i)] The \emph{inf-sup condition}:
    there exists \(c_1 > 0\) such that
    \begin{align}
      \underset{q \in Q}{\inf} \ \underset{\vec v \in V_{\vec 0}}{\sup}
      \frac{b_k(q,\vec v)}{\norm{\vec v}{V_{\vec 0}} \norm{q}{Q}} \geq c_1.
    \end{align}
  \item[(ii)] The \emph{coercivity on the kernel condition}:
    there exists \(c_2 > 0\) such that
    \begin{align}\label{eq:coercivity_on_kernel}
      a_k(\vec v, \vec v) \geq c_2 \norm{\vec v}{V_{\vec 0}}^2 &&
      \text{for all } \vec v \in \ker B,
    \end{align}
    where
    \begin{align*}
      \ker B := \left\{ \vec v \in V_{\vec 0} : b_k(q, \vec v) = 0
      ~\text{for all } q \in Q \right\}.
    \end{align*}
  \item[(iii)] \emph{Positivity of \(c\)}: it holds
      \begin{align}
        c(q,q) \geq 0 && \text{for all }q\in Q.
      \end{align}
\end{itemize}
Upon observing that \(\tensor F^{-\top} : \Grad \vec v = \div{\vec v}\),
see~\cite{holzapfel2000nonlinear}, we rewrite the bilinear form \(b_k(q, \vec v)\) as
\begin{align}
  b_k(q, \vec v) &= \int\limits_{\Omega_0} q \pi(J_k) \tensor F^{-\top}_k
  : \Grad \vec v \dX \nonumber\\
  &= \int\limits_{\Omega_0} q \pi(J_k) \div{\vec v}\dX\nonumber\\
  &= \int\limits_{\Omega_t} q \Theta'(J_k) \div{\vec v}\dx. \label{eq:non_lin_inf_sup}
\end{align}
Assuming that \(\Theta'(J) \geq 1\), we can conclude the
\emph{inf-sup} condition from standard arguments, see~\cite[Section 5.2]{weise2014}.
The positivity of the bilinear form $c$ is always fulfilled.
However, it is not possible to show the coercivity condition~\eqref{eq:coercivity_on_kernel}
for a general hyperelastic material or load configuration.
Nevertheless, for some special cases it is possible to establish a result.
We refer to~\cite{weise2014,auricchio2005stability,auricchio2010importance} for a more detailed discussion.
Henceforth, we will assume that our given input data is such that we stay
in the range of stability of the problem.
Examples for cases in which bilinear form \(a_k\) lacks coercivity can be found
in~\cite[Chapter 9]{weise2014} and~\cite[Section 4]{auricchio2010importance}.
%
\section{Finite element approximation and stabilization}\label{sec:discr}
Let \(\mathcal T_h\) be a finite element partitioning of \(\overline{\Omega}\)
into subdomains, in our case either tetrahedral or convex hexahedral elements.
The partitioning is assumed to fulfill standard regularity conditions, see~\cite{ciarlet2002finite}.
Let \(\hat{K}\) be the reference element,
and for \(K\in\mathcal T_h\) denote by \(F_\mathrm{K}\) the affine,
or trilinear mapping from \(\hat{K}\) onto \(K\).
We assume that \(F_\mathrm{K}\) is a bijection. For a tetrahedral element \(K\) this can
be assured whenever \(K\) is non-degenerate, however, for hexahedral elements
this may not necessary be the case, see~\cite{knabner2003} for details.
Further, let \(\hat{\mathbb V}\) and \(\hat{\mathbb Y}\) denote two polynomial
spaces defined over \(\hat{K}\). We denote by
\begin{align}
  \label{eq:def_vh}
  V_{h,0} &:= \left\{\vec v \in H^1_0(\Omega_0):\vec v
  = \hat{\vec v} \circ F_\mathrm{K}^{-1},\hat{\vec v}
  \in {[\hat{\mathbb V}]}^3,\forall K \in \mathcal T_h\right\},\\
  \label{eq:def_qh}
  Q_h &:= \left\{q \in L^2(\Omega_0): p = \hat{p}\circ F_\mathrm{K}^{-1},
  \hat{p}\in\hat{\mathbb Y},\forall K\in\mathcal T_h\right\},\\
  V_{h,\vec g_\mathrm{D}} &:= H^1_{\vec g_\mathrm{D}}(\Omega_0) \cap V_{h,0},
\end{align}
the spaces needed for further analysis in the following sections.
%
\subsection{\generalChange{Nearly} incompressible linear elasticity}

\cite{boffi2017remark,shariff1997,shariff2000}

As a model problem we study the well-known equations for
nearly incompressible linear elasticity.
\generalChange{
In this case it is assumed that $\Omega:=\Omega_0\approx\Omega_t$.}
Then, the linear elasticity problem reads: find \generalChange{$(\vec u,p) \in V_{\vec g_D} \times Q$} such that
\begin{align}
  \label{eq:aie_linear:1}
  2\mu\int\limits_{\Omega} \tensor{\varepsilon}(\vec u)
    : \tensor{\varepsilon}(\vec v) \dx
    + \int\limits_{\Omega}p \div \vec v \dd \Xvec{x}
    &= \int\limits_{\Omega} \vec f \cdot \vec v \dd \Xvec{x}\\
  \label{eq:aie_linear:2}
    \int\limits_{\Omega} \div \vec u q \dd \Xvec{x}
    - \frac{1}{\lambda} \int\limits_{\Omega} p q \dd \Xvec{x} &= 0
\end{align}
for all \generalChange{$(\vec v, q) \in V_{\vec 0} \times Q$}.
Here, $\mu>0$ and $\lambda$ denote the Lam\'{e} parameters, \generalChange{and}
$\tensor \varepsilon(\vec v) := \mathrm{sym}(\grad \vec v)$\generalChange{.}
%
%

The regularity of~\eqref{eq:aie_linear:1}--\eqref{eq:aie_linear:2}
is a classical result~\cite{steinbach2008numerical}
and follows with the same arguments as for the Stokes equations.
The discretized analogue of~\eqref{eq:aie_linear:1}--\eqref{eq:aie_linear:2} is:
find \generalChange{$(\vec u_h, p_h) \in V_{h,\vec g_D} \times Q_h$} such that
\begin{align}
  \label{eq:aie_linear_discr:1}
  2\mu\int\limits_{\Omega} \tensor{\varepsilon}(\vec u_h)
    : \tensor{\varepsilon}(\vec v_h) \dx
    + \int\limits_{\Omega}p_h \div \vec v_h \dd \Xvec{x}
    &= \int_{\Omega} \vec f \cdot \vec v_h \dd \Xvec{x} \\
  \label{eq:aie_linear_discr:2}
  \int\limits_{\Omega} \div \vec u_h q_h \dd \Xvec{x}
    - \frac{1}{\lambda} \int\limits_{\Omega} p_h q_h \dd \Xvec{x} &= 0
\end{align}
for all \generalChange{$(\vec v_h, q_h) \in V_{h,\vec 0} \times Q_h$}.
\reviewerOne{Coercivity on the kernel condition \eqref{eq:coercivity_on_kernel}
is a standard result for the case of nearly incompressible linear
elasticity posed in the form \eqref{eq:aie_linear:1}-\eqref{eq:aie_linear:2} and \eqref{eq:aie_linear_discr:1}-\eqref{eq:aie_linear_discr:2}.
In the nonlinear case this is not true in general and will be adressed
in Section \ref{sec:changes_lin_nolin}.
}
The crucial point for checking well-posedness of the discrete
equations~\eqref{eq:aie_linear_discr:1}--\eqref{eq:aie_linear_discr:2}
is the fulfillment of the \emph{discrete inf-sup condition}, reading
\begin{align}\label{eq:discr_inf_sup}
  \underset{q_h \in Q_h}{\inf}\underset{\vec v_h \in V_{h,\vec 0}}{\sup}
    \frac{\int\limits_{\Omega} q_h \div \vec v_h\dx}{\norm{\vec v_h}{V_{\vec 0}} \norm{q_h}{Q}} > 0.
\end{align}
The discrete \emph{inf-sup} condition puts constraints on the choice of the spaces
$V_{h,0}$ and $Q_h$. A finite element pairing fulfilling~\eqref{eq:discr_inf_sup}
is called a \emph{stable pair}.
A classic example for tetrahedral meshes would be the Taylor--Hood element.
In this paper, we will focus on two different finite element pairings,
the MINI element and a stabilized equal order element.
\reviewerTwo{
  The stabilized equal order pairing has been used in this context for pure tetrahedral meshes, see \cite{cante2014,rodriguez2016}. To the best of the authors knowledge those elements have not been used in the present context for general tesselations.
}
\subsection{The pressure-projection stabilized equal order pair}%
\label{sec:pressure_projection}
In the following, we present a stabilized lowest equal order finite
element pairing, adapted to nonlinear elasticity from the pairing originally
introduced by \textcite{dohrmann2004stabilized,bochev2006} for the Stokes
equations.

We choose $\hat{\mathbb V}$ and $\hat{\mathbb Y}$ in~\eqref{eq:def_vh}--\eqref{eq:def_qh}
as the space of linear (or trilinear) functions over $\hat{K}$.
This choice of spaces is a textbook example of an unstable element, however,
following~\cite{dohrmann2004stabilized}, we can introduce a stabilized formulation
of~\eqref{eq:aie_linear_discr:1}--\eqref{eq:aie_linear_discr:2} by:
find $(\vec u_h,p_h) \in V_{h,\vec g_D} \times Q_h$ such that
\begin{align}\label{eq:lin_elast_stab_discr}
  \mu\int\limits_{\Omega} \tensor{\varepsilon}(\vec u_h)
    : \tensor{\varepsilon}(\vec v_h) \dd \Xvec{x}
    + \int\limits_{\Omega}p_h \div \vec v_h \dd \Xvec{x}
    &= \int_{\Omega} \vec f \cdot \vec v_h \dd \Xvec{x}, \\
  \int\limits_{\Omega} \div \vec u_h q_h \dd \Xvec{x}
   - \frac{1}{\lambda} \int\limits_{\Omega} p_h q_h \dd \Xvec{x}
    \quad\qquad& \nonumber \\
    - \frac{1}{\mu^\ast} s_h(p_h,q_h) &= 0,
\end{align}
for all $(\vec v_h, q_h) \in V_{h,0} \times Q_h$, where
\begin{align}\label{eq:def_db_stabil}
  s_h(p_h,q_h) :=  \int\limits_{\Omega}(p_h - \Pi_h p_h) (q_h - \Pi_h q_h)\dd \Xvec{x}
\end{align}
and $\mu^\ast>0$ a suitable parameter. \generalChange{We note that the integral in \eqref{eq:def_db_stabil} has to be understood as sum over integrals of elements of the tessellation.}
The projection operator $\Pi_h$ is defined element-wise for each $K \in \mathcal T_h$
\begin{align*}
  \left.\Pi_h p_h\right|_{K} := \frac{1}{\seminorm{K}{}} \int\limits_{K}p_h\dd\Xvec{x}.
\end{align*}
We can state the following results for this discrete problem:
\begin{theorem}
There exists a unique bounded solution to the discrete
problem~\eqref{eq:lin_elast_stab_discr}.
\end{theorem}
\begin{theorem}
Assume that \(\vec u \in {[H^1_{\vec g_D}(\Omega)]}^3 \cap {[H^2(\Omega)]}^3\) and
\(p \in \Lm{2}{\Omega} \cap \Hm{1}{\Omega}\) solve the
problem~\eqref{eq:aie_linear:1}--\eqref{eq:aie_linear:2}.
Further, assume that \((\vec u_h, p_h)\) are the solutions to the stabilized
problem~\eqref{eq:lin_elast_stab_discr}.
Then there exists a constant \(c_3\) independent of the mesh size \(h\) and it holds:
\begin{align}
  \norm{\vec u - \vec u_h}{V} + \norm{p-p_h}{Q} \leq
    c_3 h(\normHm{\vec u}{2}{\Omega}+\normHm{p}{1}{\Omega})
\end{align}
\end{theorem}
\begin{proof}
Due to the similarity of the linear elasticity and the Stokes problem the proof
follows from~\cite[Theorem 4.1, Theorem 5.1 and Corollary 5.2]{bochev2006}.
\end{proof}
%
\subsection{Discretization with MINI-elements}%
\label{sec:mini_element}
\subsubsection{Tetrahedral elements}
One of the earliest strategies in constructing a stable finite element pairing
for discrete saddle-point problems arising from Stokes Equations is the
MINI-Element, dating back to the works of Brezzi et al, see for
example~\cite{arnold1984stable,brezzi1992relationship}.
In the case of Stokes the velocity ansatz space is enriched by suitable
polynomial bubble functions.
More precisely, if we denote by $\hat{\mathbb P}_1$ the space of polynomials
with degree $\leq 1$ over the reference tetrahedron $\hat{K}$, we will choose
\reviewerTwo{
\begin{align*}
  \hat{\mathbb V} &= \hat{\mathbb P}_1 \oplus \{\hat{\psi}_\mathrm{B}\},\\
  \hat{\mathbb Y} &= \hat{\mathbb P}_1, \\
  \hat{\psi}_\mathrm{B} &:= 256 \xi_0 \xi_1 \xi_2 (1 - \xi_0 - \xi_1 - \xi_2),
\end{align*}
where \((\xi_0, \xi_1, \xi_2) \in \hat{K}\), see also \cite{boffi2013mixed}.
}
Classical results~\cite{boffi2013mixed} guarantee the stability of the
MINI-Element for tetrahedral meshes.
Due to compact support of the bubble functions, static condensation can be
applied to remove the interior degrees of freedom during assembly. A short review on the static condensation process is given in the Appendix. Hence, these degrees of freedom are not needed to be
considered in the full global stiffness matrix assembly which is a key
advantage of the MINI element.

\subsubsection{Hexahedral meshes}
\generalChange{In the literature mostly two dimensional quadrilateral \reviewerTwo{tessellations},
see for example~\cite{boffi2013mixed,bai1997quadmini,lamichhane2017quadrilateral}, were considered for MINI element discretizations.}
In this case, the proof of stability relies on the so-called
\emph{macro-element technique} proposed by \textcite{stenberg1990error}.

To motivate our novel ansatz for hexahedral bubble functions,
we will first give an overview of Stenbergs main results.
A macro-element $M$ is a connected set of elements in $\mathcal T_h$.
Moreover, two macro-elements $M_1$ and $M_2$ are said to be equivalent
\reviewerTwo{if and only if} they can be mapped continuously onto each other.
Additionally, for a macro element $M$ we define the spaces
\begin{align}
\label{eq:macro_elem_kernel}
\begin{split}
  \boldsymbol V_{0,\mathrm{M}}
    &:= \left\{\vec v \in {[H^1_0(M)]}^3 : \vec v
    = \hat{\vec v} \circ F_\mathrm{K}^{-1}, \right. \\
    &\left. \quad\quad\hat{v}
    \in {[\hat{\mathbb V}]}^3,~K \subset M\right\},\\
  P_\mathrm{M}
    &:= \left\{p \in L^2(M):p
    = \hat{p}\circ F_\mathrm{K}^{-1},\hat{p}
    \in \hat{\mathbb Y},~K\subset M\right\},\\
  N_\mathrm{M}
    &:= \left\{p \in P_\mathrm{M} : \int\limits_{M} p \div \vec v\dx = 0,
   \forall \vec v \in \boldsymbol V_{0,\mathrm{M}}\right\}.
\end{split}
\end{align}
Denote by $\Gamma_h$ the set of all edges in $\mathcal T_h$ interior
to $\Omega$. The macro-element partition $\mathcal M_h$ of
$\Omega$ then consists of a (not necessarily disjoint) partitioning into
macro-elements ${\{M_i\}}_{i=1}^M$ with
 $\overline{\Omega} = \bigcup_{i=1}^M \overline{M}_i$.
The macro element technique is then described by the following theorem,
 see~\cite{stenberg1990error}.
\begin{theorem}\label{thm:macroelement}
Suppose that there is a fixed set of equivalence classes
$\mathcal E_j$, $j=1,\ldots,q$, of macro-elements,
a positive integer $L$,
and a macro-element partition $\mathcal M_h$ such that
\begin{enumerate}
  \item[(M1)] for each $M_i \in \mathcal E_j$, $j=1,\ldots,q$,
    the space $N_\mathrm{M}$ is one-dimensional consisting of functions that
    are constant on $M$;
  \item[(M2)] each $M \in \mathcal M_h$ belongs to one of the classes
    $\mathcal E_i$, $i=1,2,\ldots,q$;
  \item[(M3)] each $K \in \mathcal T_h$ is contained in at least one and not
    more than $L$ macro-elements of $\mathcal M_h$;
  \item[(M4)] each $E \in \Gamma_h$ is contained in the interior of at least
    one and not more than $L$ macro-elements of $\mathcal M_h$.
\end{enumerate}
Then the discrete $\inf$-$\sup$-condition~\eqref{eq:discr_inf_sup} holds.
\end{theorem}
Conditions (M2)--(M4) are valid for a quasi-uniform \reviewerTwo{tessellation}
of $\Omega$ into hexahedral elements and, thus, it remains to show (M1).
To this end, we consider a macro-element $M_i \in \mathcal M_h$ consisting of
 eight hexahedrons that share a common vertex $\vec x_i \in \Omega$, see
Figure~\ref{fig:hexa_macroelement_def}.
A macro-element partitioning of this type fulfills conditions (M1)--(M3)
from Theorem~\ref{thm:macroelement}.
We will next show, that Assumption (M1) depends on the choice of the bubble
functions inside every $K \in M_i$.
For ease of presentation and with no loss of generality we will assume that
$M_i$ is a parallelepiped.
This means that the mapping $F_{\mathrm{M}_i}$ from $\hat{K}$ onto $M_i$
is affine, so there exists an invertible matrix
 $\tensor J_i \in \mathbb R^{3\times 3}$ such that
\begin{align*}
  \vec x = F_{\mathrm{M}_i}(\boldsymbol \xi)
    = \tensor J_i \boldsymbol\xi + \vec x_0,
\end{align*}
where $\boldsymbol \xi \in \hat{K} = {[-1,1]}^3$ and $\vec x_0$ is a given node
of $M_i$.
The case of $M_i$ not being the image of an affine mapping of
$\hat{K}$ can be handled analogously, however, there are constraints on the
invertibility of $F_{\mathrm{M}_i}$, see~\cite{knabner2003}.
Let ${\{\psi_j\}}_{j=1}^8$ denote the standard trilinear basis functions on the
unit hexahedron. These functions will serve as a basis for $P_{\mathrm{M}_i}$.
For the space $\boldsymbol{V}_{0,\mathrm{M}_i}$ we will chose one piecewise
continuous trilinear ansatz function defined in $\vec x_i$ and for
each sub-hexahedron we will add two bubble functions as degrees of freedom.
The distribution of the degrees of freedom is depicted in
Figure~\ref{fig:hexa_macroelement}.
On $\hat{K}$ we will define the following two bubble functions
\begin{align}
  \label{eq:def_hexa_bubble:1}
  \hat{\phi}_\mathrm{B}^1 &:= {(1-\xi_0)}^2{(1-\xi_1)}^2{(1-\xi_2)}^2
    \hat{\psi}_{\alpha},\\
  \label{eq:def_hexa_bubble:2}
  \hat{\phi}_\mathrm{B}^2 &:= {(1-\xi_0)}^2{(1-\xi_1)}^2{(1-\xi_2)}^2
    \hat{\psi}_{\beta},
\end{align}
where the indices $\alpha, \beta$ are chosen such that $\hat{\psi}_{\alpha}$
and $\hat{\psi}_{\beta}$ are two ansatz functions belonging to two diagonally
opposite nodes.
Having this, we will form a basis for $\boldsymbol V_{0,\mathrm{M}_i}$
by gluing together the images of the basis functions of each sub-hexahedron.
So we can write a basis for $\boldsymbol V_{0,\mathrm{M}_i}$ as
\begin{align}
  \label{eq:def_v0mi}
  V_{0,\mathrm{M}_i} &:= \operatorname{span}\{\psi_{\vec{x}_i},
    \phi_{\mathrm{B},1}^1, \phi_{\mathrm{B},2}^1,\ldots,\phi_{\mathrm{B},1}^8,
    \phi_{\mathrm{B},2}^8\},\\
  \nonumber
  \boldsymbol{V}_{0,\mathrm{M}_i} &:= {[V_{0,\mathrm{M}_i}]}^3.
\end{align}
Here, $\psi_{\vec x_i}$ corresponds to a piecewise trilinear ansatz function
that has unit value in $\vec x_i$ and zero in all other nodes of $M_i$.
Thus, we can calculate that $\dim(P_{\mathrm{M}_i}) = 27$
and $\dim(\boldsymbol V_{0,\mathrm{M}_i}) = 51$.
For ease of presentation we will rename the elements of~\eqref{eq:def_v0mi}
as ${\{\phi_i\}}_{i=1}^{17}$.
Now, for $q_h \in P_{\mathrm{M}_i}$ and
$\vec v_h \in \boldsymbol{V}_{0,\mathrm{M}_i}$ we can write
\begin{align*}
  \int\limits_{M_i} q_h \div \vec v_h\dx
    = \sum_{k=1}^{17}\sum_{l=1}^{27}\sum_{j=1}^3v_k^j q_l
      \int\limits_{M_i} \grad_{\vec x} \phi_k[j] \psi_l\dx.
\end{align*}
Next, we use the chain rule to get
 $\grad_{\vec x} \phi_k = \tensor J_i^{-\top} \hat{\grad}_{\boldsymbol \xi}$
 and a change of variables to obtain
\begin{align*}
  &\sum_{k=1}^{17}\sum_{l=1}^{27}\sum_{j=1}^3 v_k^j q_l
    \int\limits_{M_i} \grad_{\vec x} \phi_k[j] \psi_l\dx,\\
  &=\sum_{k=1}^{17}\sum_{l=1}^{27}\sum_{j=1}^3v_k^j q_l
    \int\limits_{\hat{K}} \tensor J_i^{-\top}\hat{\grad}_{\boldsymbol \xi}
    \hat{\phi}_k[j] \hat{\psi}_j\lvert \det{\tensor J_i}\rvert\mathrm{d}\boldsymbol\xi.
\end{align*}
This means we can find a matrix $\widetilde{\tensor D} \in \mathbb{R}^{27\times 51}$ such that
\begin{align*}
  \int\limits_{M_i} q_h \div \vec v_h\dx = \vec q^\top  \widetilde{\tensor D} \vec v,
\end{align*}
where $\vec q$ and $\vec v$ encode the nodal values of $q_h$ and $\vec v_h$.
The following ordering will be employed for $\vec v$
\begin{align*}
  \vec v = {\left(v_1^1,v_1^2,v_1^3,\ldots,v_{17}^1,v_{17}^2,v_{17}^3\right)}^\top.
\end{align*}
To proof (M1) we need to show that the rank of the matrix $\widetilde{\tensor D}$ is 26.
Due to the invertibility of $\tensor J_i$ the rank of the matrix $\widetilde{\tensor D}$ will remain unchanged by replacing $M_i$ by $\hat{K}$.
Thus, it suffices to compute the rank of the matrix $\tensor D$
whose $j^{\mathrm{th}}$ row is defined by
\begin{multline*}
  \left(\int\limits_{\hat{K}} \partial_{\xi_1}\phi_1\psi_1\mathrm{d}\boldsymbol\xi,\int\limits_{\hat{K}} \partial_{\xi_2}\phi_1\psi_1\mathrm{d}\boldsymbol\xi,\int\limits_{\hat{K}} \partial_{\xi_3}\phi_1\psi_1\mathrm{d}\boldsymbol\xi,\right.\\
   \left.,\ldots,\right. \\ \left. \int\limits_{\hat{K}} \partial_{\xi_1}\phi_{17}\psi_j\mathrm{d}\boldsymbol\xi,\int\limits_{\hat{K}} \partial_{\xi_2}\phi_{17}\psi_j\mathrm{d}\boldsymbol\xi,\int\limits_{\hat{K}} \partial_{\xi_3}\phi_{17}\psi_j\mathrm{d}\boldsymbol\xi\right).
\end{multline*}
By this formula the matrix $\tensor D$ can be explicitly calculated, e.g., by
using software packages like \emph{Mathematica}$^{\mathrm{TM}}$ and further
analyzed. We can conclude that the rank of $\tensor D$ is 26 and thus~(M1)
holds and we can apply Theorem~\ref{thm:macroelement}.
A \emph{Mathematica}$^{\mathrm{TM}}$ notebook
containing computations discussed in this section is available upon request.
\begin{remark}
  Contrary to the two-dimensional case studied
  in~\cite{bai1997quadmini,lamichhane2017quadrilateral} it is not sufficient to
  enrich the standard isoparametric finite element space for hexahedrons with
  only one bubble function.
  In this case both the spaces $\boldsymbol V_{0,\mathrm{M}_i}$
  and $P_{\mathrm{M}_i}$ have a dimension of 27, however,
  matrix $\tensor D$ has only rank 24.
\end{remark}
\begin{remark}
  Although not mentioned explicity, the stability of the MINI element holds also for mixed discretizations.
\end{remark}
\begin{figure}[ht]
  \centering
  \includegraphics[width=0.8\linewidth]{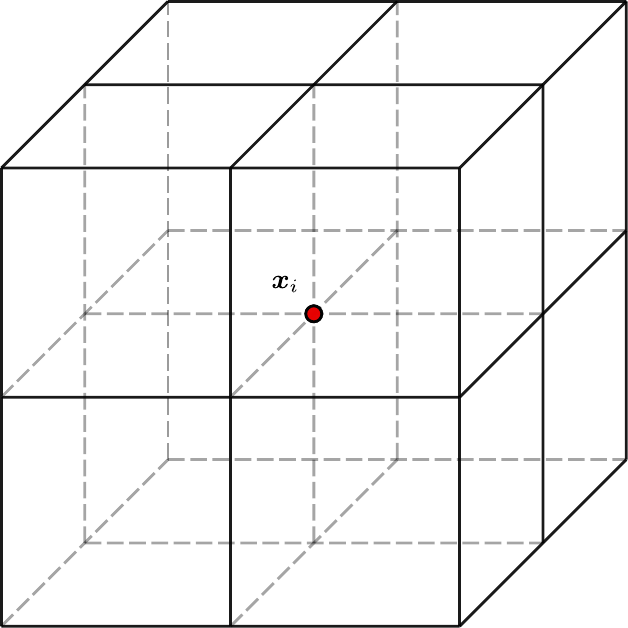}
  \caption{Macro-element definition for a mesh point $\vec x_i$.}%
  \label{fig:hexa_macroelement_def}
\end{figure}
\begin{figure}[ht]
  \centering
  \includegraphics[width=0.8\linewidth]{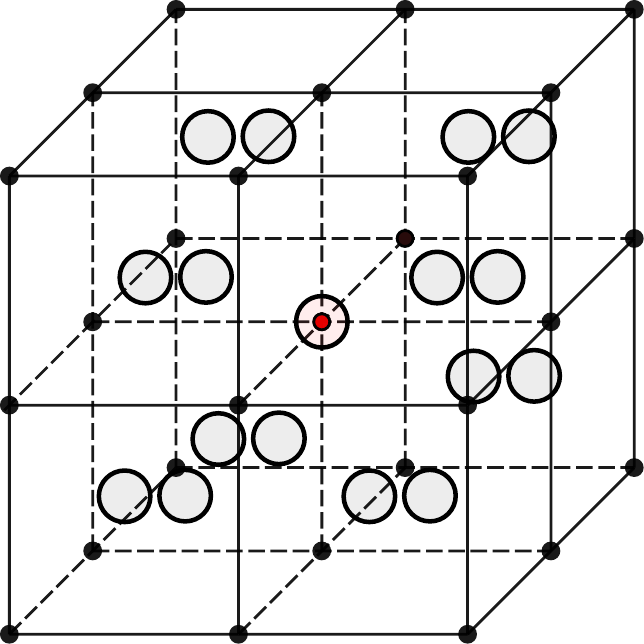}
  \caption{Macro-element distribution of degrees of freedom for
    $\vec v_h \in \boldsymbol{V}_{0,\mathrm{M}}$ and $q_h \in P_\mathrm{M}$.
    Small filled dots correspond to $P_\mathrm{M}$ and bigger opaque circles
    correspond to $\boldsymbol{V}_{0,\mathrm{M}}$.}%
  \label{fig:hexa_macroelement}
\end{figure}
\subsection{Changes and limitations in the nonlinear case}%
\label{sec:changes_lin_nolin}
\reviewerOne{
One of the main differences between the linear and nonlinear case stems
from the definition of the pressure $p$ as remarked in~\cite{boffi2017remark}.
Consider, as an example, the strain energy function for a nearly
incompressible neo-Hookean material where
\begin{align*}
 \overline{\Psi}(\overline{\tensor C}) := \frac{\mu}{2}\left(\mathrm{tr}(\overline{\tensor C})-3\right),
\end{align*}
with $\mu>0$ a material parameter.
Then, $\tensor{S}_{\mathrm{tot}}$ and $\mathbb{C}_{\mathrm{tot}}$,
evaluated at $(\vec u_k, p_k) = (\vec 0, 0)$,
are given by
\begin{align*}
 \tensor{S}_{\mathrm{tot}} = \tensor 0,\quad
 \mathbb{C}_{\mathrm{tot}} = 2\mu \tensor I \odot \tensor I - \frac{2\mu}{3} \tensor I \otimes \tensor I,
\end{align*}
independent of the choice of $\Theta(J)$.
Assuming that $\Omega:=\Omega_0\approx\Omega_t$ we obtain from
Eqs.~\eqref{eq:saddle_point_nl:1}--\eqref{eq:saddle_point_nl:2} the following
linear system
\begin{align}
  \label{eq:aie_hydro:1}
  2\mu\int\limits_{\Omega} \tensor{\varepsilon}_\mathrm{d}(\vec u)
  : \tensor{\varepsilon}_\mathrm{d}(\vec v) \dx
    + \int\limits_{\Omega}p \div \vec v \dd \Xvec{x}
    &= \int\limits_{\Omega} \vec f \cdot \vec v \dd \Xvec{x},\\
  \label{eq:aie_hydro:2}
  \int\limits_{\Omega} \div \vec u q \dd \Xvec{x}
    - \frac{1}{\kappa} \int\limits_{\Omega} p q \dd \Xvec{x} &= 0,
\end{align}
where
$\tensor{\varepsilon}_\mathrm{d}(\vec u) := \tensor{\varepsilon}(\vec u) - \frac{1}{3}\mathrm{div}(\vec u) \tensor I$.
While the pressure in formulation~\eqref{eq:aie_linear:1}--\eqref{eq:aie_linear:2}
is usually denoted as \emph{Herrmann pressure}~\cite{herrmann1965elasticity},
above formulation~\eqref{eq:aie_hydro:1}--\eqref{eq:aie_hydro:2}
uses the so-called \emph{hydrostatic pressure}.\newline
The arguments to prove the \emph{inf-suf} condition for this linear problem
remains the same as for~\eqref{eq:aie_linear:1}--\eqref{eq:aie_linear:2}.
For the extension of the \emph{inf-suf} condition to the nonlinear case we
already stated earlier in Eq.~\eqref{eq:non_lin_inf_sup} that
\begin{align*}
  b_k(q_h,\vec v_h) = \int\limits_{\Omega_{t,h}} q_h \Theta'(J_h) \div \vec v_h\dx.
\end{align*}
Here, $\Omega_{t,h}$ is the approximation of the real current configuration
$\Omega_t$. Our conjecture is that stability of the chosen elements is given
provided sufficient fine discretizations and volumetric functions
$\Theta(J)$ fulfilling $\Theta'(J) \geq 1$.
However, we can not offer a rigorous proof of this, and rely on our numerical
results which showed no sign of numerical instabilities.\newline
Concerning well-posedness of~\eqref{eq:aie_hydro:1}--\eqref{eq:aie_hydro:2},
it was noted in~\cite{boffi2017remark},
that the coercivity on the kernel condition~\eqref{eq:coercivity_on_kernel}
does not hold in general, which makes the formulation with hydrostatic pressure
not well-posed in general.
However, it remains well-posed for strictly divergence-free finite elements or pure
Dirichlet boundary conditions.
This has also been observed by other authors, see~\cite{khan2019,viebahn2018}.
Even if the coercivity on the kernel condition can be shown for the hydrostatic,
nearly incompressible linear elastic case this result may not transfer
to the nonlinear case.
Here, this condition is highly dependent on the chosen nonlinear material law and
for the presented benchmark examples (Section~\ref{sec:results})
we did not observe any numerical instabilities.\newline
For an in-depth discussion we refer the interested reader
to~\cite{auricchio2010importance,auricchio2005stability}.
A detailed discussion on Herrmann-type pressure in
the nonlinear case is presented in~\cite{shariff1997,shariff2000}. \newline
To show well-posedness for the special case of the presented MINI element
discretizations
we rely on results given in \cite[Section 4]{boffi2017remark}.
There it is shown, that discrete coercivity on the kernel holds,
provided that a rigid body mode is the only function that renders
\begin{align*}
  a(\vec{u}_h, \vec{v}_h) := \int\limits_{\Omega_h}
  \tensor{\varepsilon}_d(\vec{u}_h) : \tensor{\varepsilon}_d(\vec{v}_h)\dx
\end{align*}
from~\eqref{eq:aie_hydro:1}--\eqref{eq:aie_hydro:2} zero.
We could obtain this result following the same procedure outlined in
\cite{boffi2017remark} for both hexahedral and tetrahedral MINI elements.
A \emph{Mathematica}$^{\mathrm{TM}}$ notebook
containing the computations discussed is available upon request.\newline
In the case of the pressure-projection stabilization
we will modify Equation~\eqref{eq:def_r_lower} using the stabilization
term~\eqref{eq:def_db_stabil}
\begin{align*}
  R_\mathrm{lower}(\vec u_h, p_h;\Delta q_h) &:= b_{\mathrm{vol}}(\vec u_h; \Delta q_h) - c(p_h,\Delta q_h) \\&- \frac{1}{\mu^*} s_h(p_h,q_h).
\end{align*}
Here, the stabilization parameter $\mu^*>0$ is supposed to be large enough and
will be specifically defined for each nonlinear material considered.
Note, that by modifying the definition of the lower residual,
we introduced a mesh dependent perturbation of the original residual.
An estimate of the consistency error caused by this is not readily
available and will be the topic of future research. However, results and
comparisons to benchmarks in Section~\ref{sec:results} suggest that
this error is negligible for the considered problems as long as $\mu^\ast$ is
well-chosen.
\reviewerTwo{
If not specified otherwise we chose
\begin{itemize}
  \item $\mu^\ast=\mu$ for neo-Hookean materials and
  \item $\mu^\ast=c_1$ for Mooney--Rivlin materials
\end{itemize}
in the results section.}
For the pressure-projection stabilized equal order pair we can not transfer
the results from the linear elastic case to the non-linear case,
as the proof of well-posedness relies on the coercivity of
$a_k(\vec u, \vec v)$ which can not be concluded for this formulation.
However, no convergence issues occured in the numerical examples given
in Section \ref{sec:results}.\\
The considerable advantage of the MINI element is that there are no
modifications needed and that no additional stabilization parameters are
introduced into the system.
}

\subsection{Changes and limitations in the transient case}%
The equations presented in Section~\ref{sec:methods} are not yet suitable
for transient simulations. To include this feature we modify the nonlinear
variational problem~\eqref{eq:non_lin_var_prob:1} in the following way:
\begin{align}
  \label{eq:trans_nonlin:1}
  R_\mathrm{upper}^\mathrm{trans}(\vec u,p; \Delta \vec v)
    &:= \rho_0 \int\limits_{\Omega_0} \ddot{\vec u} \cdot \Delta \vec v\dx
    + R_\mathrm{upper}(\vec u, p; \Delta v),\\
  \label{eq:trans_nonlin:2}
  R_\mathrm{lower}^\mathrm{trans}(\vec u,p; \Delta q)
    &:= R_\mathrm{lower}(\vec u, p; \Delta q).
\end{align}
For time discretization we considered a generalized-$\alpha$ method,
see~\cite{chung1993} and also the Appendix for a short summary.
Due to the selected formulation, the resulting ODE system turns out to be of
degenerate hyperbolic type. Hence, we implemented a variant of the
generalized-$\alpha$ method as proposed in~\cite{kadapa2017} and using that we
did not observe any numerical issues in our simulations.
Note, that other groups have proposed a different treatment of the
incompressibility constraints in the case of transient problems,
see~\cite{rossi2016implicit,scovazzi2016simple} for details.
%
\section{Numerical examples}%
\label{sec:results}
%
While benchmark cases presented in this section are fairly simple,
mechanical applications often require highly resolved meshes. Thus,
efficient and massively parallel solution algorithms for the
linearized system of equations become an important factor to deal with the
resulting computational load.
After discretization, at each Newton--Raphson step a block system of the form
\begin{align*}
\begin{pmatrix}
 \tensor K_h & \tensor B_h^\top \\
 \tensor B_h & \tensor C_h
\end{pmatrix}
\begin{pmatrix}
  \Delta \vec u \\
  \Delta \vec p
\end{pmatrix}
=
\begin{pmatrix}
  -\Xvec R_\mathrm{upper} \\
  -\Xvec R_\mathrm{lower}
\end{pmatrix}
\end{align*}
has to be solved.
In that regard, we used a generalized minimal residual method (GMRES)
and efficient preconditioning based
on the \texttt{PCFIELDSPLIT}%
\footnote{https://www.mcs.anl.gov/petsc/petsc-current/docs/manualpages/PC/PCFIELDSPLIT.html}
package from the library
\emph{PETSc}~\cite{petsc-user-ref} and the incorporated
solver suite \emph{hypre/BoomerAMG}~\cite{henson2002boomeramg}.
By extending our previous work~\cite{augustin2016anatomically}
we implemented the methods in the finite element code
\textit{Cardiac Arrhythmia Research Package} (CARP)~\cite{vigmond2008solvers}.
%
\subsection{Analytic solution}\label{sec:analytic_solution}
To verify our implementation we consider a very simple uniaxial tension test,
see also~\cite[Sec. 10.1]{weise2014}.
The computational domain is described by one eighth part of a cylinder with
length $L=\SI{2}{\mm}$, and radius $R=\SI{1}{\mm}$
\begin{align*}
  \Omega_\mathrm{cyl,0} := \left\{ \vec x \in [0,L] \times {[0,R]}^2
  : y^2+z^2\leq R\right\},
\end{align*}
see Figure~\ref{fig:analytic_solution_cylinder}.
This cylinder is stretched to a length of $L+\Delta L$,
with $\Delta L = \SI{2}{\mm}$.
\begin{figure}[htbp]
  \centering
  \includegraphics[width=1.0\linewidth]{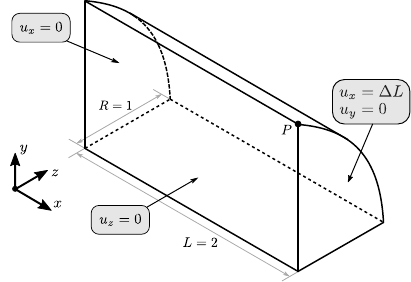}
  \caption{\emph{Analytic solution:} geometry and boundary conditions.}%
  \label{fig:analytic_solution_cylinder}
\end{figure}

\reviewerTwo{We chose a neo-Hookean material
\begin{align*}
  \Psi(\tensor C) = \frac{\mu}{2}\left(\mathrm{tr}(\overline{\tensor C})
    - 3\right) + \frac{\kappa}{2} {\ln(J)}^2,
\end{align*}
with $\mu =\SI{7.14}{\mega\pascal}$ and impose full incompressiblity, i.e.,
$1/\kappa=0$.}
For this special case, an analytic solution can be computed by
\begin{align*}
  \vec u &= (t x, \Delta R(t) y, \Delta R(t) z),\\
  p(t) &= \frac{\mu}{3}\left({\left(1+\frac{t \Delta L}{L}\right)}^2
      -{\left(1+\frac{t \Delta L}{L}\right)}^{-1}\right),\\
      \Delta R(t) &= {\left(1+\frac{t\Delta L}{L}\right)}^{-\frac{1}{2}} - 1,
\end{align*}
where $t \in [0,1]$ corresponds to the load increment.
Two meshes consisting of \num{5420} points and \num{4617} hexahedral
or \num{27702} tetrahedral elements were used.
We performed 20 incremental load steps \reviewerTwo{with respect to} $\Delta L$.
In Figure~\ref{fig:cyl_comp} it is shown that the results of the numerical
simulations render identical results for all the chosen setups and are in
perfect agreement with the exact solution plotted in blue.
\begin{figure}[htbp]
  \textbf{(a)}\\
  \includegraphics[width=0.9\linewidth]{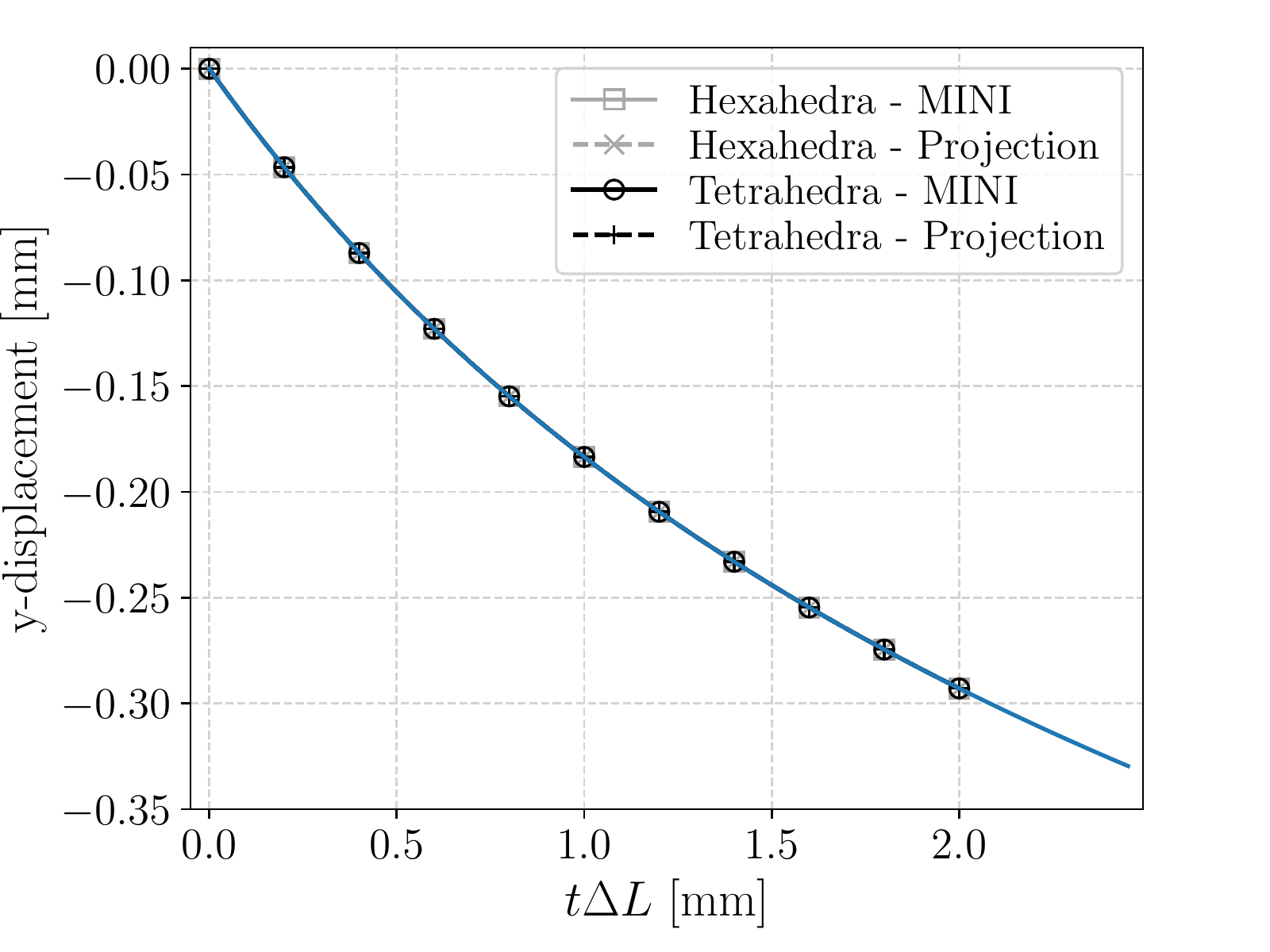}\\
  \textbf{(b)}\\
  \includegraphics[width=0.9\linewidth]{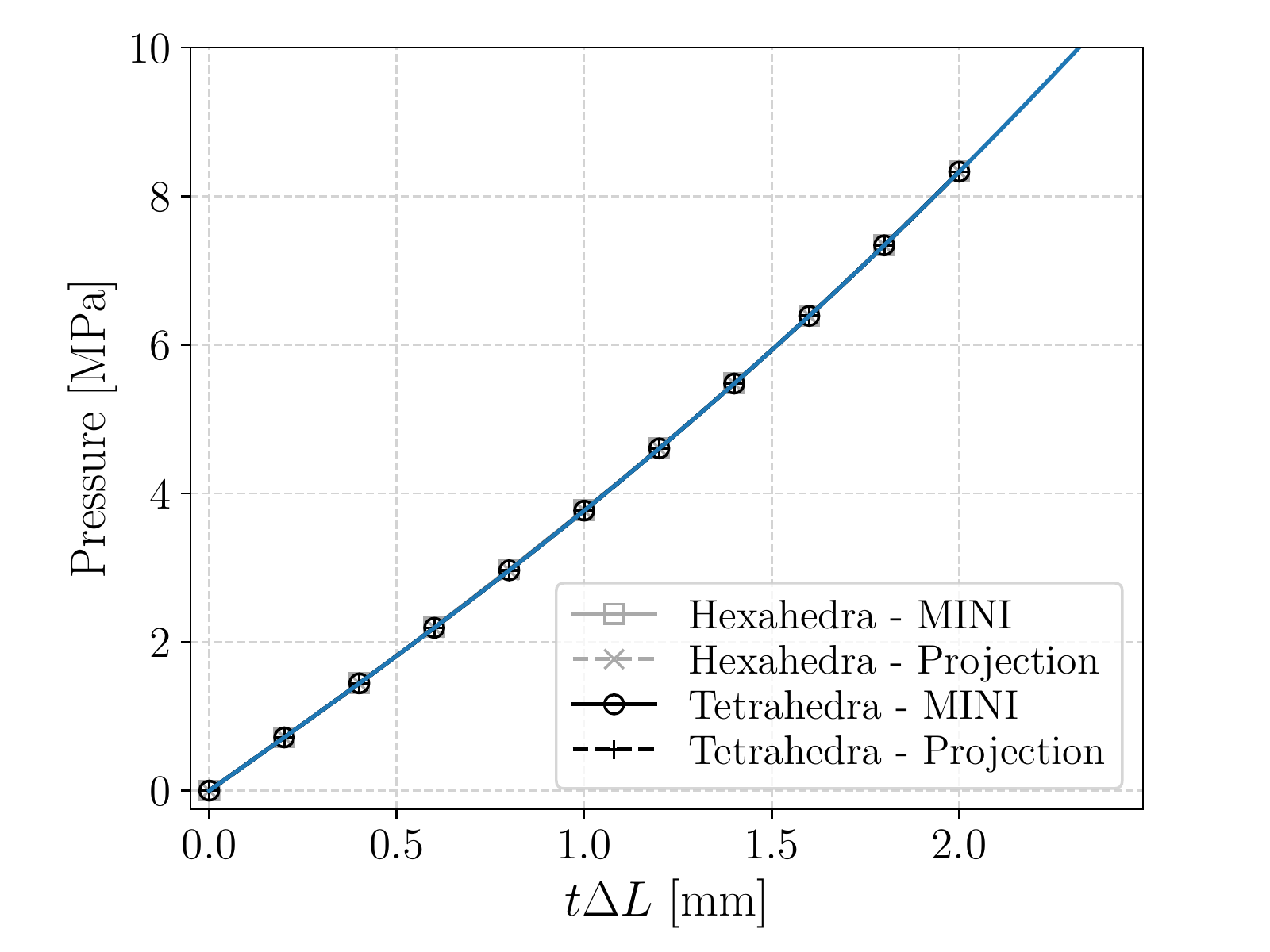}
  \caption{\emph{Analytic solution:} (a) $y$-component of displacement and (b) pressure at point
    $P={(2,0,1)}^\top$. Simulation results of all proposed formulations are in
    perfect alignment with the analytic solution printed in blue.}%
  \label{fig:cyl_comp}
\end{figure}

\subsection{Block under compression}\label{sec:block}
The computational domain, studied by multiple authors,
see, e.g.,~\cite{caylak2012stabilization,masud2013framework,reese2000new},
consists of a cube loaded by an applied pressure in the center of the top face;
see Figure~\ref{fig:block_under_compression}.
A quarter of the cube is modeled, where symmetric Dirichlet boundary conditions
are  applied to the vertical faces and the top face is fixed in the horizontal
plane.
\begin{figure}[htbp]
  \centering
  \includegraphics[width=0.8\linewidth]{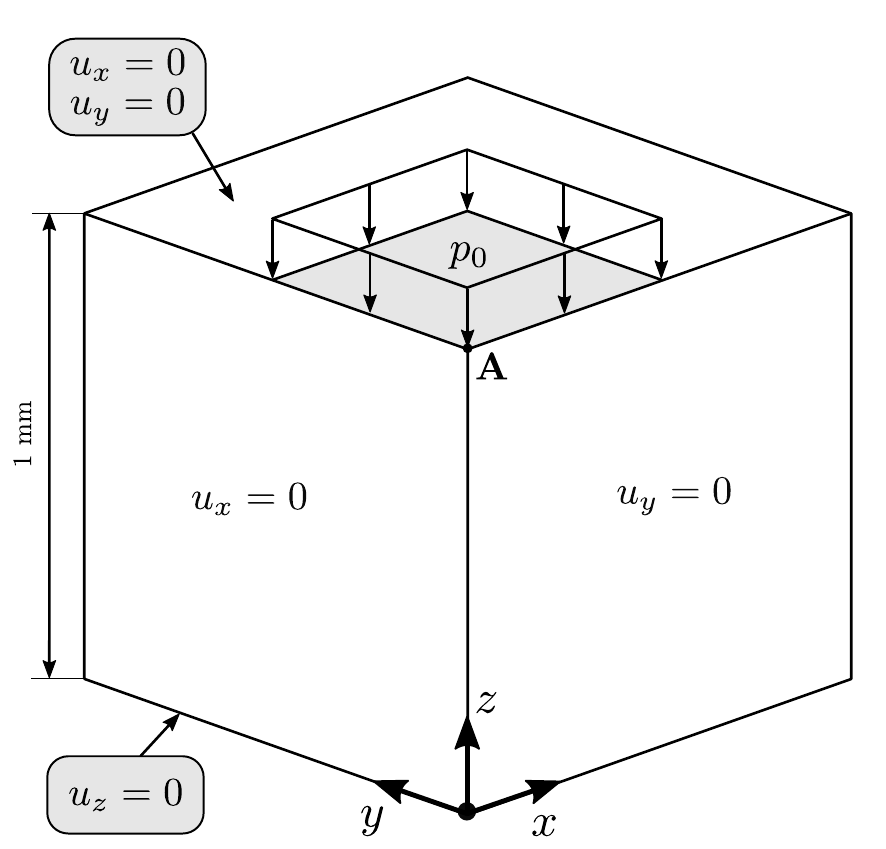}
  \caption{\emph{Block under compression:} geometry and boundary conditions.}%
  \label{fig:block_under_compression}
\end{figure}

The same neo-Hookean material model as in~\cite{masud2013framework} is used:
\begin{align*}
  \Psi(\tensor{C})
    &=\frac{1}{2}\mu\left(\operatorname{tr}(\tensor{C})-3\right)-\mu\ln J
    + \frac{\lambda}{2}{(\ln J)}^2,
\end{align*}
with material parameters
\(\lambda=\SI{400889.806}{\MPa}\), \(\mu=\SI{80.194}{\MPa}\).
To test mesh convergence the simulations were computed on a series of
uniformly refined tetrahedral and hexahedral meshes,
see Table~\ref{tab:block_meshes}.
\reviewerTwo{
  Figure~\ref{fig:block_meshes} shows the deformed meshes for the level $\ell=2$ with loads
  $p_0=\SI{320}{\mega\Pa}$ and $p_0=\SI{640}{\MPa}$, respectively.
}
\generalChange{In all cases discussed in this section we used 10 loading steps to arrive
at the target pressure $p_0$.}
\begin{table}[htpb]
\reviewerOne{
  \caption{Properties of \emph{cube meshes} used
           in Section~\ref{sec:block}.}%
  \label{tab:block_meshes}
  \centering
  \begin{tabular}{rrr}
    \toprule
    \multicolumn{3}{c}{Hexahedral Meshes}\\
    $\ell$ & Elements & Nodes \\
    \midrule
    1 & \num{512}   & \num{729}  \\
    2 & \num{4096}  & \num{4913}  \\
    3 & \num{32768} & \num{35937} \\
    4 & \num{262144} & \num{274625} \\
    5 & \num{2097152} & \num{2146689} \\
    \bottomrule
  \end{tabular}\hspace{1em}
  \begin{tabular}{rrr}
    \toprule
    \multicolumn{3}{c}{Tetrahedral Meshes}\\
    $\ell$ & Elements & Nodes \\
    \midrule
    1 & \num{3072}   & \num{729}  \\
    2 & \num{24576}  & \num{4913}  \\
    3 & \num{196608} & \num{35937} \\
    4 & \num{1572864} & \num{274625} \\
    5 & \num{12582912} & \num{2146689} \\
    \bottomrule
  \end{tabular}
}
\end{table}
As a measure of the compression level the vertical displacement of the node
at the center of the top surface, i.e. the edge point \textbf{A} of
the quarter of the cube, is plotted in Figure~\ref{fig:block_displacementA}.
Small discrepancies can be attributed to differences in the meshes for
tetrahedral and hexahedral grids, however, the different stabilization
techniques yield almost the same results for finer grids.
\reviewerOne{
Note, that the displacements at the edge point \textbf{A} obtained using
the simple \(\mathbb Q_1 - \mathbb P_0\) hexahedral and
\(\mathbb P_1 - \mathbb P_0\) tetrahedral elements seem to be in a similar
range compared to the other approaches. The overall displacement field, however,
was totally inaccurate rendering \(\mathbb Q_1 - \mathbb P_0\) and \(\mathbb P_1 - \mathbb P_0\)
elements an inadequate choice for this benchmark problem.
The solution for Taylor--Hood (\(\mathbb P_2 - \mathbb P_1\)) tetrahedral elements
was obtained using the FEniCS~project~\cite{alnaes2015fenics}.
Here, as a linear solver, we used a GMRES solver with preconditioning similar to the MINI and
projection-based approach, see first paragraph of Section~\ref{sec:results}.
The \texttt{PCFIELDSPLIT} and \emph{hypre/BoomerAMG} settings were slightly
adapted to optimize computational performance for quadratic ansatz functions.
We comparing simulations with about the same number of degrees of freedom,
not accuracy as, e.g., in \cite{chamberland2010comparison}.
For coarser grids computational times were in the same time range for
all approaches; see, e.g., the cases with approximately $10^6$ degrees of
freedom and target pressure of $p_0=\SI{320}{\mmHg}$ in Table~\ref{tab:block:computational_times}(a).
For the simulations with the finest grids
with approximately $10^7$ degrees of freedom, however, we could not find a
setting for the Taylor--Hood elements that was competitive to MINI and
pressure-projection stabilizations. The computational times to
arrive at the target pressure of $p_0=\SI{320}{\mmHg}$ using 192 cores on ARCHER, UK
were about 10 times higher for Taylor--Hood elements using FEniCS, see
Table~\ref{tab:block:computational_times}(b).}
\reviewerOne{
We attribute that to a higher communication load and  higher memory requirements
of the Taylor--Hood elements: memory to store the block stiffness matrices
was approximately
\num{2.5} times higher for Taylor--Hood elements compared to MINI and
projection-stabilization approaches (measured using the \texttt{MatGetInfo}%
\footnote{https://www.mcs.anl.gov/petsc/petsc-current/docs/manualpages/Mat/MatGetInfo.html}
function provided by PETSc).
Note, that although we used the same linear solvers,
the time comparisons are not totally just as results were obtained using two
different finite element solvers, CARP and FEniCS.
Note also, that timings are usually very problem dependent and for this block
under compression benchmark high accuracy was already achieved with
coarse grids for hexahedral and Taylor--Hood discretizations.\newline
For a further analysis regarding computational
costs of the MINI element and the pressure-projection stabilization, see Section~\ref{sec:twisting_column}.
}
\begin{table}[htbp!]
\reviewerOne{
  \caption{\emph{Block under compression:}
    Comparison of computational times for different discretizations.
    Timings were obtained using \textbf{(a)} 48 cores and \textbf{(b)}
    192 cores on ARCHER, UK.
    Coarser grids, see Table~\ref{tab:block_meshes}, are used for
    Taylor--Hood elements
    \(\mathbb P_2- \mathbb P_1\) to compare computational times for a
  similar number of degrees of freedom (DOF).}%
  \label{tab:block:computational_times}
  \textbf{(a)}\\[0.5em]
  \begin{tabular}{rrrrr}
    \toprule
    Discretization & Grid & DOF  &Tet. & Hex. \\
    \midrule
     Projection & $\ell=4$ & 1.098 Mio. & \SI{330}{\s} & \SI{438}{\s}    \\
     MINI       & $\ell=4$ & 1.098 Mio. & \SI{873}{\s} & \SI{655}{\s}   \\
     \(\mathbb P_2- \mathbb P_1\) &$\ell=3$ & 0.860 Mio. &  \SI{1202}{\s}  & -- \\
    \bottomrule
  \end{tabular}\\[1em]
  \textbf{(b)}\\[0.5em]
  \begin{tabular}{rrrrr}
    \toprule
    Discretization & Grid & DOF & Tet. & Hex. \\
    \midrule
    Projection & $\ell=5$ & 8.587 Mio. &   \SI{2488}{\s} & \SI{2192}{\s}    \\
    MINI       & $\ell=5$ & 8.587 Mio. &   \SI{3505}{\s} & \SI{4640}{\s}   \\
    \(\mathbb P_2 - \mathbb P_1\) & $\ell=4$ & 6.715 Mio.  & \SI{27154}{\s}  & -- \\
    \bottomrule
  \end{tabular}
}
\end{table}

\generalChange{
In Figure~\ref{fig:block_comparison} the hydrostatic pressure is plotted for
the MINI element and the projection-based stabilization. These
results are very smooth in all cases and agree well with those published
in~\cite{caylak2012stabilization,elguedj2008projection,masud2013framework,reese2000new}.}

\begin{figure}[htbp]
  \textbf{(a)}\\
  \includegraphics[width=0.9\linewidth]{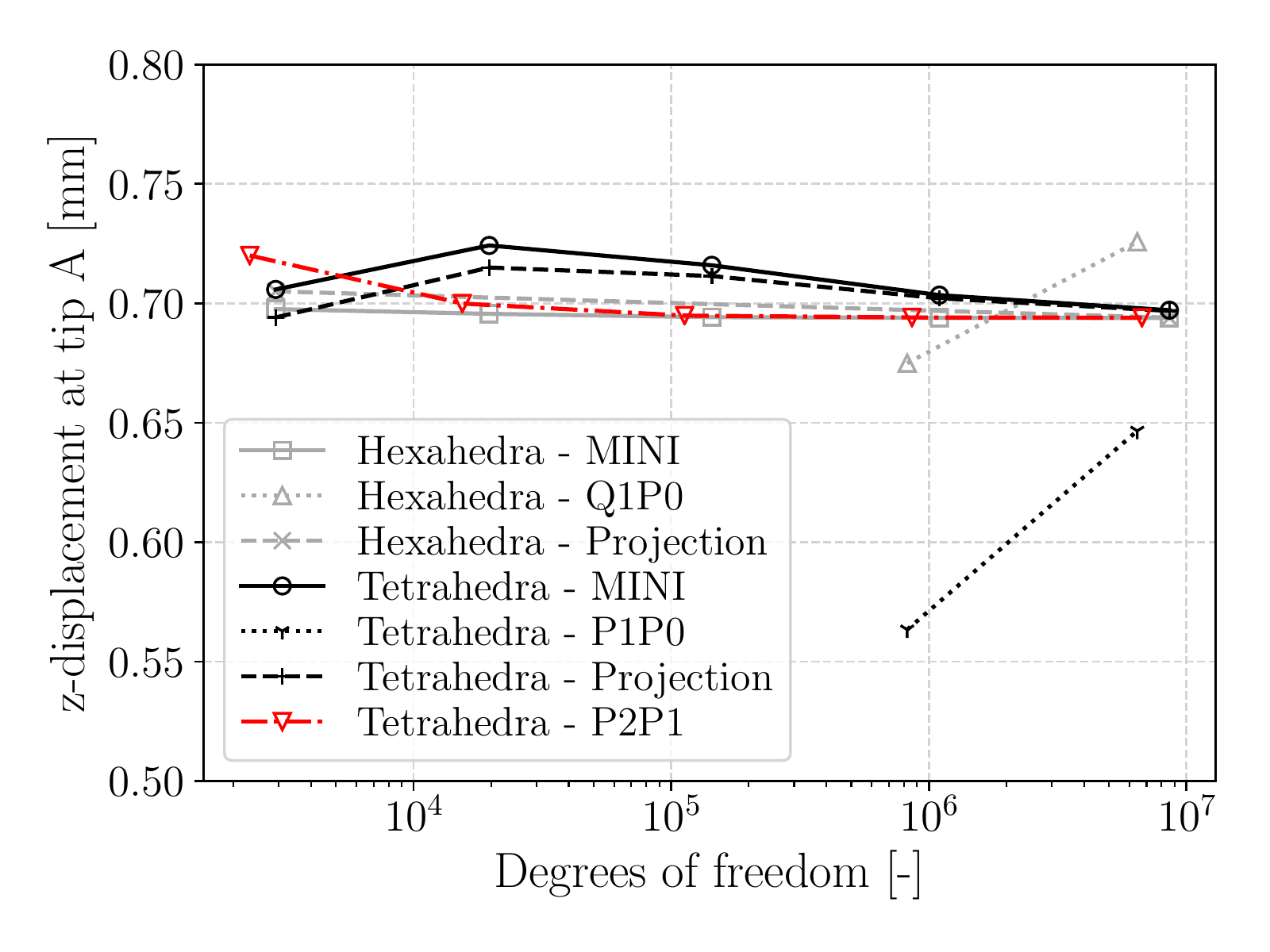}\\
  \textbf{(b)}\\
  \includegraphics[width=0.9\linewidth]{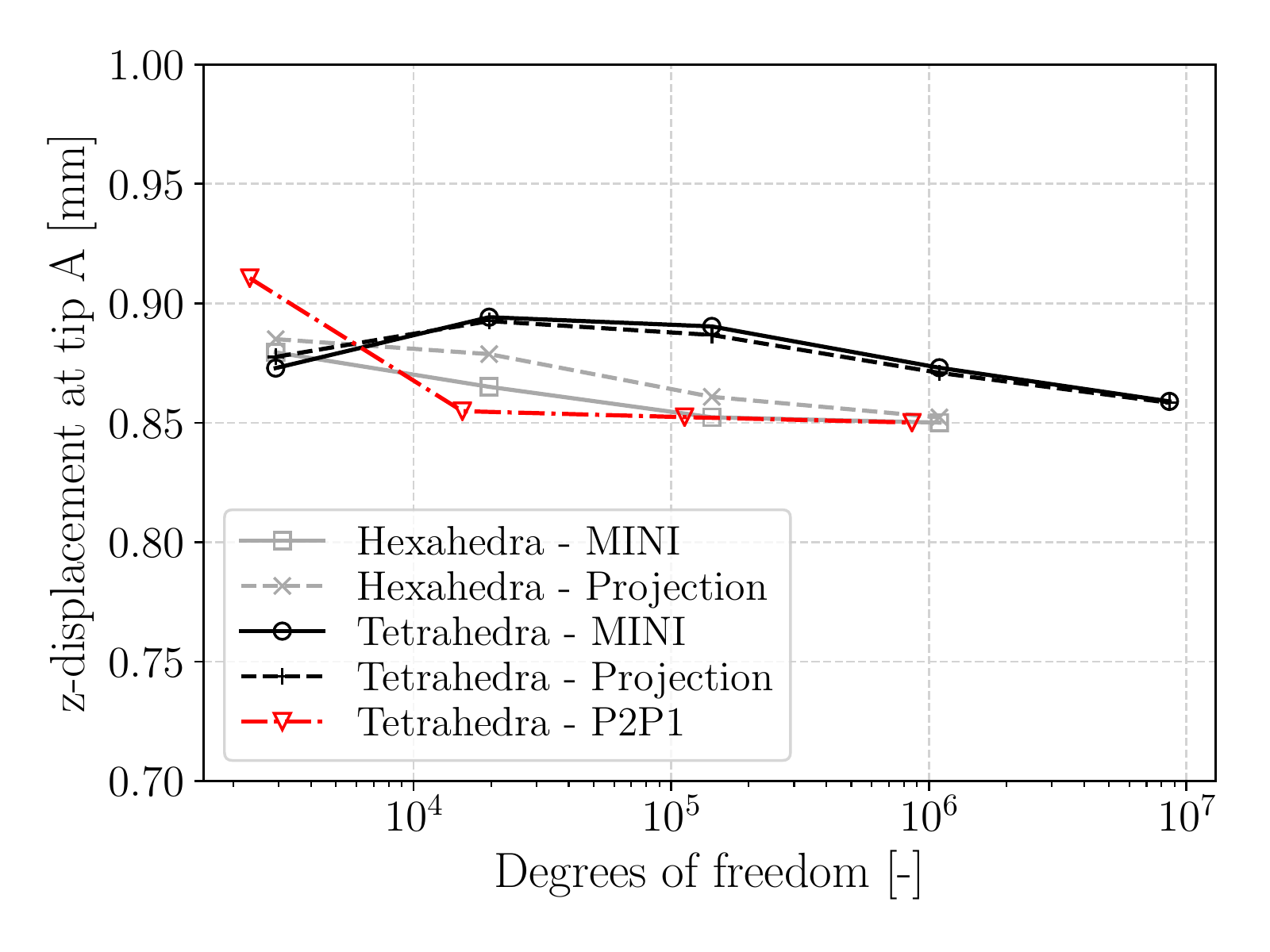}
  \caption{\reviewerOne{\emph{Block under compression:}
    vertical displacement at point \textbf{A} versus number of degrees of freedom
    in a logarithmic scale at load level
    (a) $p_0=\SI{320}{\mega\Pa}$ and (b) $p_0=\SI{640}{\MPa}$. Results for the MINI
    element and the  pressure-projection stabilization are compared to classical
    choices of elements, i.e., \(\mathbb Q_1 - \mathbb P_0\) hexahedral elements,
    \(\mathbb P_1 - \mathbb P_0\) tetrahedral elements,
    and Taylor--Hood (\(\mathbb P_2 - \mathbb P_1\)) tetrahedral elements.
    For case (b) the choice of \(\mathbb Q_1 - \mathbb P_0\) and
    \(\mathbb P_1 - \mathbb P_0\) elements did not give reasonable results and
    were thus omitted.}}%
  \label{fig:block_displacementA}
\end{figure}
\begin{figure*}[htbp]
  \centering
  \includegraphics[width=1.0\linewidth]{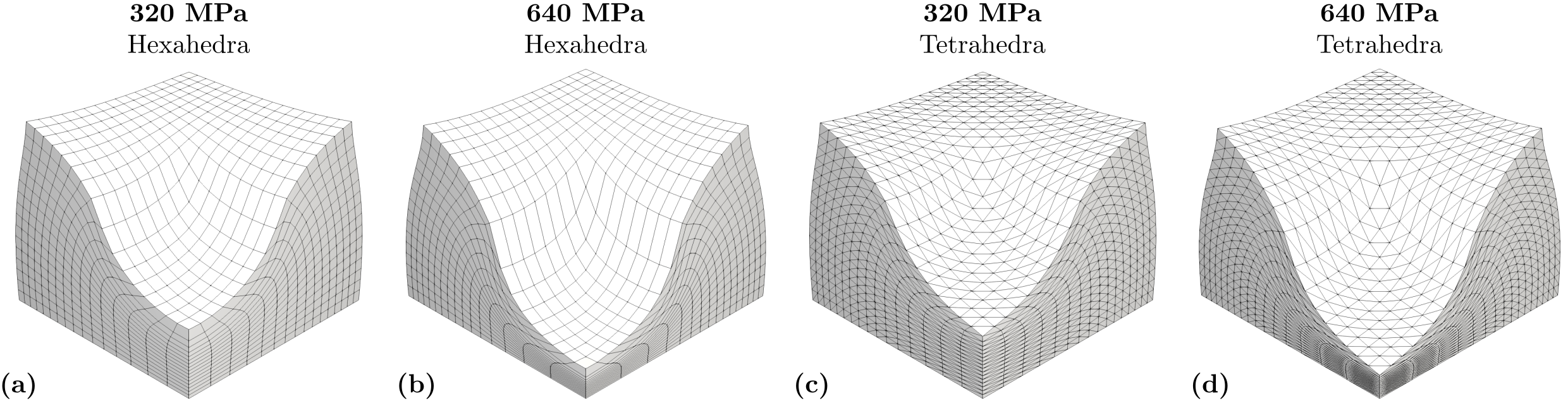}
  \caption{\emph{Block under compression:}
    deformed meshes of hexahedral (a,b) and tetrahedral (c,d) elements
    for the $\ell=2$ mesh in~Table~\ref{tab:block_meshes} at load level
  $p=\SI{320}{\mega\Pa}$ (a,c) and load level $p=\SI{640}{\MPa}$ (b,d).}%
  \label{fig:block_meshes}
\end{figure*}
\begin{figure*}[ht]
  \centering
  \includegraphics[width=1.0\linewidth]{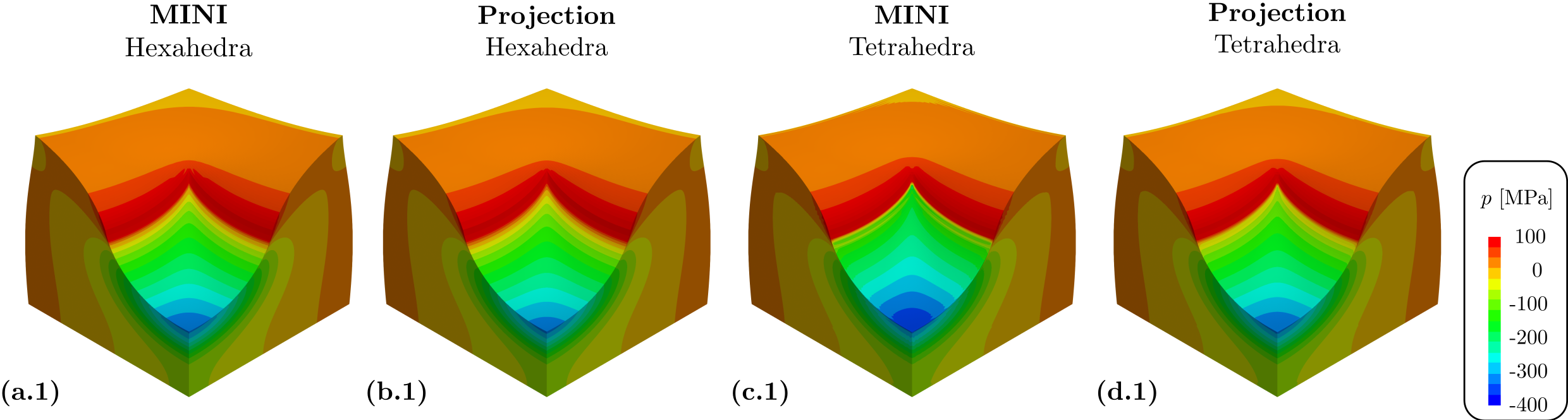}\\[1em]
  \includegraphics[width=1.0\linewidth]{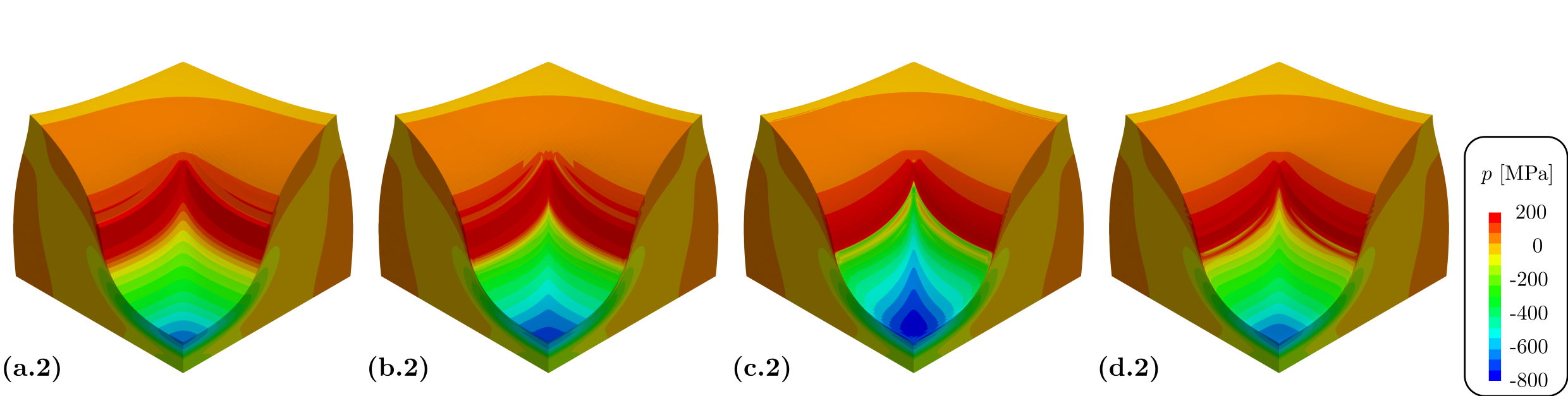}
  \caption{\emph{Block under compression:}
    comparison of hexahedral (a,b) and tetrahedral (c,d) elements with
    bubble-based (a,c) and projection-based (b,d) stabilization. Shown is the
    pressure contour on the deformed mesh at load level $p=\SI{320}{\mega\Pa}$
  in the first row and $p=\SI{640}{\MPa}$ in the second row.}%
  \label{fig:block_comparison}
\end{figure*}
%
\clearpage%
\subsection{Cook-type cantilever problem}\label{sec:cook}
In this section, we analyze the same Cook-type cantilever beam problem presented
in~\cite{bonet2015computational,schroeder2011new},
see also Figure~\ref{fig:cook_cantilever}.
Displacements at the plane $x=\SI{0}{\mm}$ are fixed.
At the plane $x=\SI{48}{\mm}$ a parabolic load,
which takes its maximum at $t_0=\SI{300}{\kPa}$, is applied.
Note, that this in-plane shear force in y-direction is considered as a dead load
in the deformation process.
\begin{figure}[htpb]
  \centering
  \includegraphics[width=0.8\linewidth]{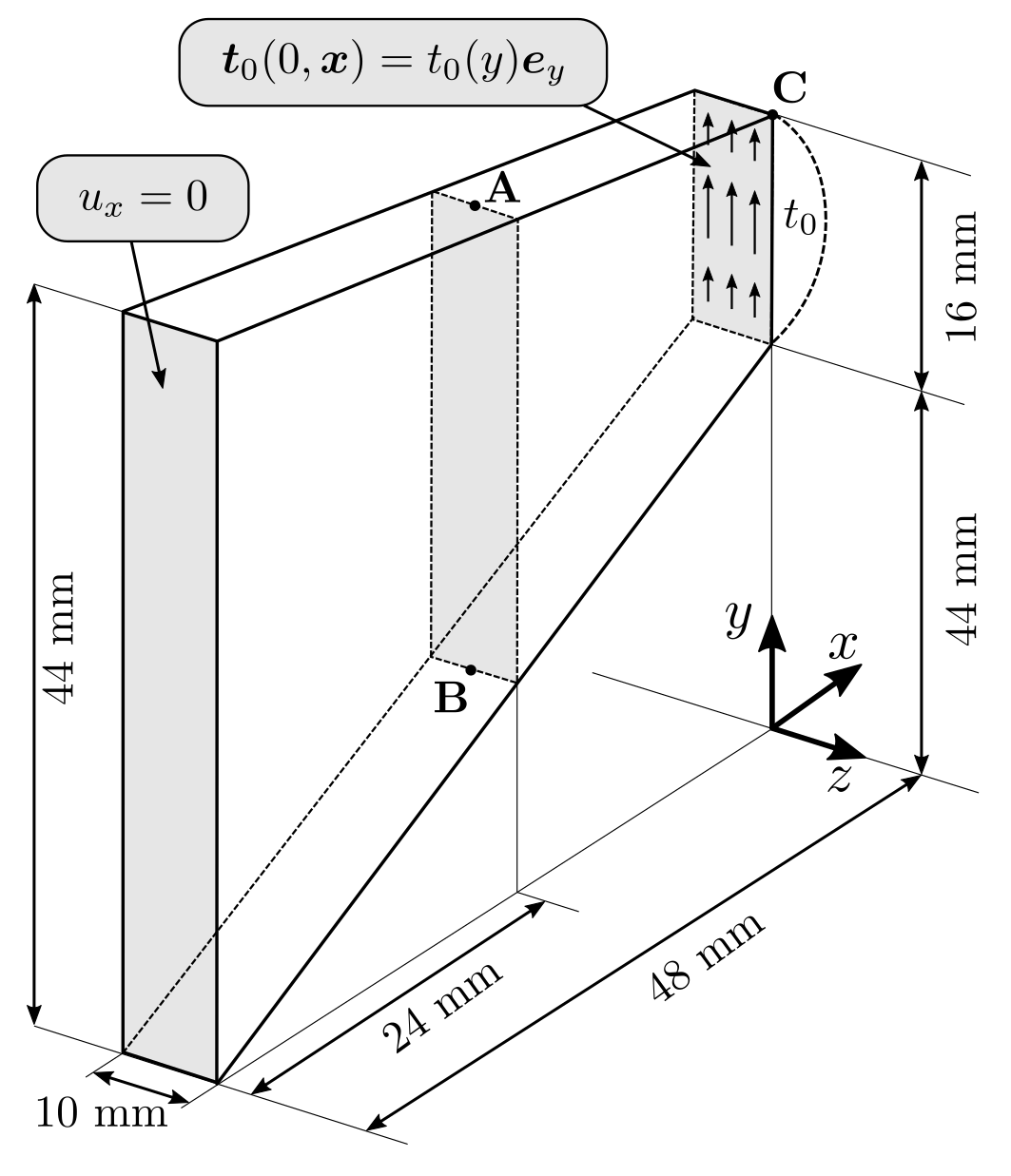}
  \caption{\emph{Cook-type cantilever problem:} geometry and boundary conditions.}%
  \label{fig:cook_cantilever}
\end{figure}
To compare to results in~\cite{schroeder2011new} the same isotropic strain energy
function was chosen
\[
  \Psi^\mathrm{iso}(\tensor{C})=c_{1}{(\operatorname{tr}\tensor{C})}^2
  + c_{2}{( {(\operatorname{tr}\tensor{C})}^2
  - \operatorname{tr}(\tensor{C}^2))}^2 - \gamma\ln(J),
\]
with material properties $c_1=\SI{21}{\kPa}$, $c_2=\SI{42}{\kPa}$, and
$\gamma=12c_1+24c_2$ to satisfy the condition of a stress-free reference geometry.

We chose a fully incompressible material, hence,
\[
  \Psi^\mathrm{vol}(\tensor{C})=\frac{\kappa}{2}{\left(J-1\right)}^2,
\]
with \reviewerTwo{$1/\kappa=0$}.
First, mesh convergence with respect to resulting displacements is analyzed
for the tetrahedral and hexahedral meshes with discretization details given in
Table~\ref{tab:cantilever_meshes}.
\begin{table}[htbp]
  \caption{Properties of \emph{cantilever meshes} used in
    Section~\ref{sec:cook}.}%
  \label{tab:cantilever_meshes}
  \centering
  \begin{tabular}{rrr}
    \toprule
    \multicolumn{3}{c}{Hexahedral Meshes}\\
    $\ell$ & Elements & Nodes \\
    \midrule
    1 & \num{324}    & \num{500}    \\
    2 & \num{2592}   & \num{3249}   \\
    3 & \num{20736}  & \num{23273}  \\
    4 & \num{165888} & \num{175857} \\
    \bottomrule
  \end{tabular}\hspace{2em}
  \begin{tabular}{rrr}
    \toprule
    \multicolumn{3}{c}{Tetrahedral Meshes}\\
    $\ell$ & Elements & Nodes \\
    \midrule
    1 & \num{1944}   & \num{500}    \\
    2 & \num{15552}  & \num{3249}   \\
    3 & \num{124416} & \num{23273}  \\
    4 & \num{995328} & \num{175857} \\
    \bottomrule
  \end{tabular}
\end{table}

Displacements $u_x$, $u_y$, and $u_z$ at point $\mathbf{C}$ are shown in
Figure~\ref{fig:cook_displacements}. The proposed stabilization techniques
give comparable displacements in all three directions and also match results
published in~\cite{bonet2015computational,schroeder2011new}.
Mesh convergence can also be observed for the stresses
$\sigma_{xx}$ at point $\mathbf{A}$ and $\mathbf{B}$
and $\sigma_{yy}$ at point $\mathbf{B}$, see Figure~\ref{fig:cook_stresses}.
Again, results match well those presented
in~\cite{bonet2015computational,schroeder2011new}.
Small discrepancies can be attributed to the fully
incompressible formulation used in our work and differences in grid construction.

\reviewerOne{In Figure~\ref{fig:cook_boxplots} and
Figure~\ref{fig:cook_comparison}(a) the distribution of
$J=\det(\tensor{F})$ is shown to provide an estimate of how accurately the
incompressibility constraint is fulfilled by the proposed stabilization
techniques. For most parts of the computational domains the values
of $J$ are close to $1$, however, hexahedral meshes and here in particular
the MINI element maintain the condition of $J \approx 1$ more accurately on the
element level. Note, that for all discretizations the overall volume of the
cantilever remained unchanged at \SI{14400}{\mm\cubed}, rendering the material
fully incompressible on the domain level.}

\reviewerOne{Figure~\ref{fig:cook_comparison} gives a comparison of several
computed values
in the deformed configuration of Cook's cantilever for the finest grids ($\ell=4$).
Slight pressure oscillations in Figure~\ref{fig:cook_comparison}(b)
on the domain boundary for the MINI element are to be expected,
see~\cite{soulaimani1987}; this also affects the distribution of $J$ in
Figure~\ref{fig:cook_comparison}(a). A similar checkerboard pattern is present
for the projection based stabilization.}

\generalChange{In the third row of Figure~\ref{fig:cook_comparison} we compare
the stresses $\sigma_{xx}$ for the different stabilization techniques.
We can observe slight oscillations for the the projection-based approach,
whereas the MINI element gives a smoother solution.}
Compared to results in~\cite[Figure~10]{schroeder2011new} the $\sigma_{xx}$
stresses have a similar contour but are slightly higher.
As before, we attribute that to the
fully incompressible formulation in our paper compared to the quasi-incompressible
formulation in~\cite{schroeder2011new}.
\begin{figure}[htbp]
  \textbf{(a)}\\
  \includegraphics[width=0.9\linewidth]{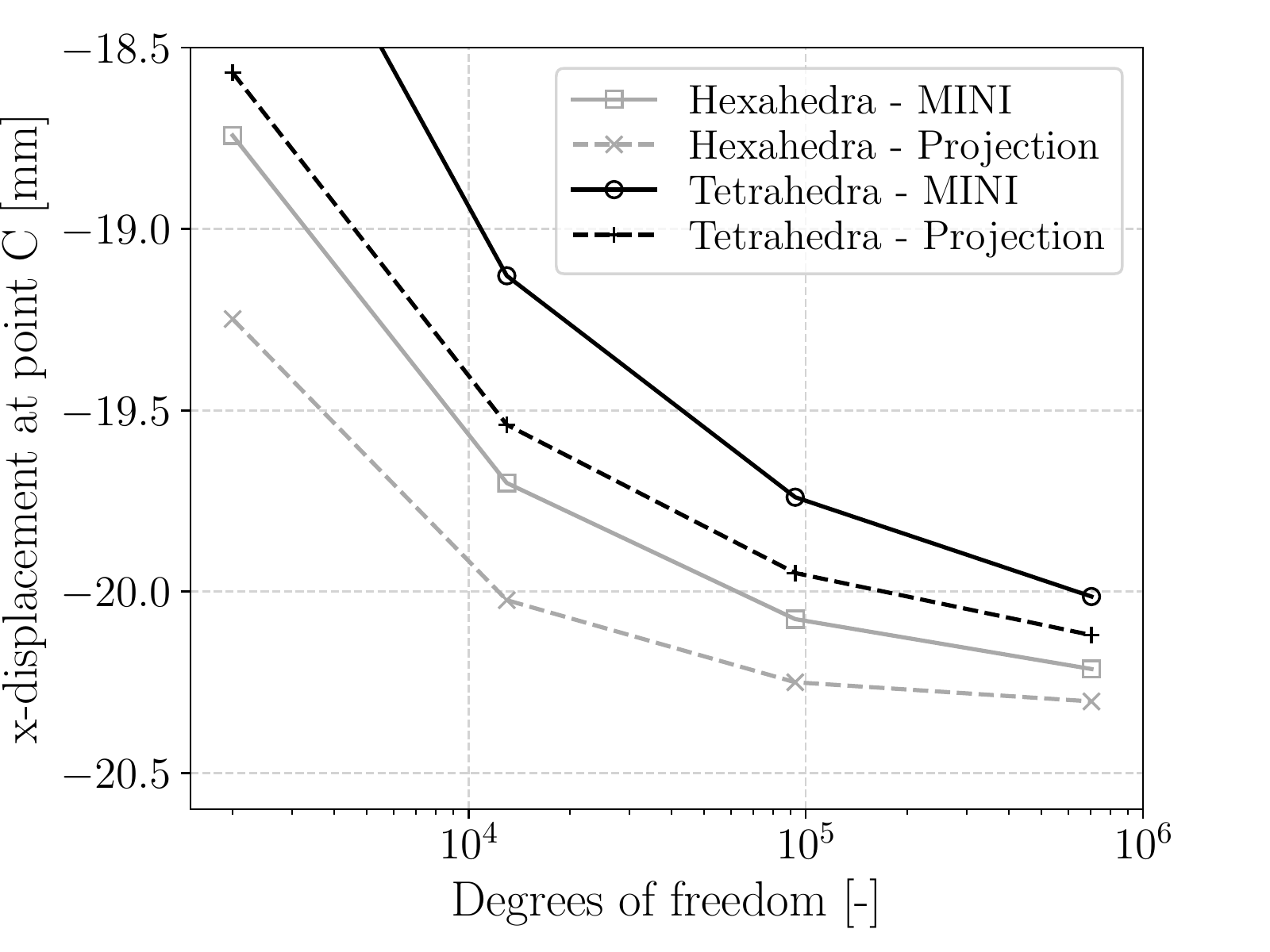}\\
  \textbf{(b)}\\
  \includegraphics[width=0.9\linewidth]{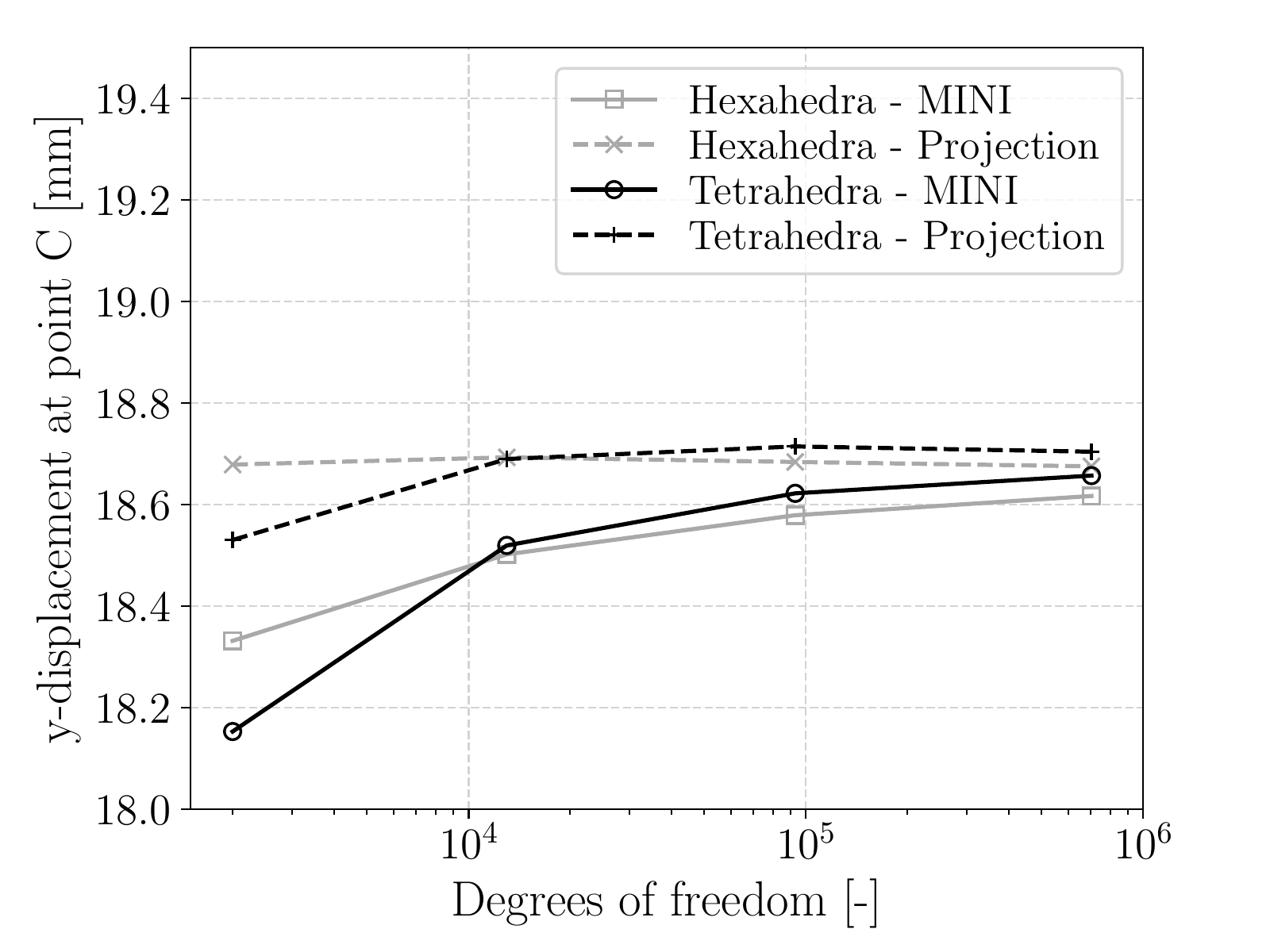}\\
  \textbf{(c)}\\
  \includegraphics[width=0.9\linewidth]{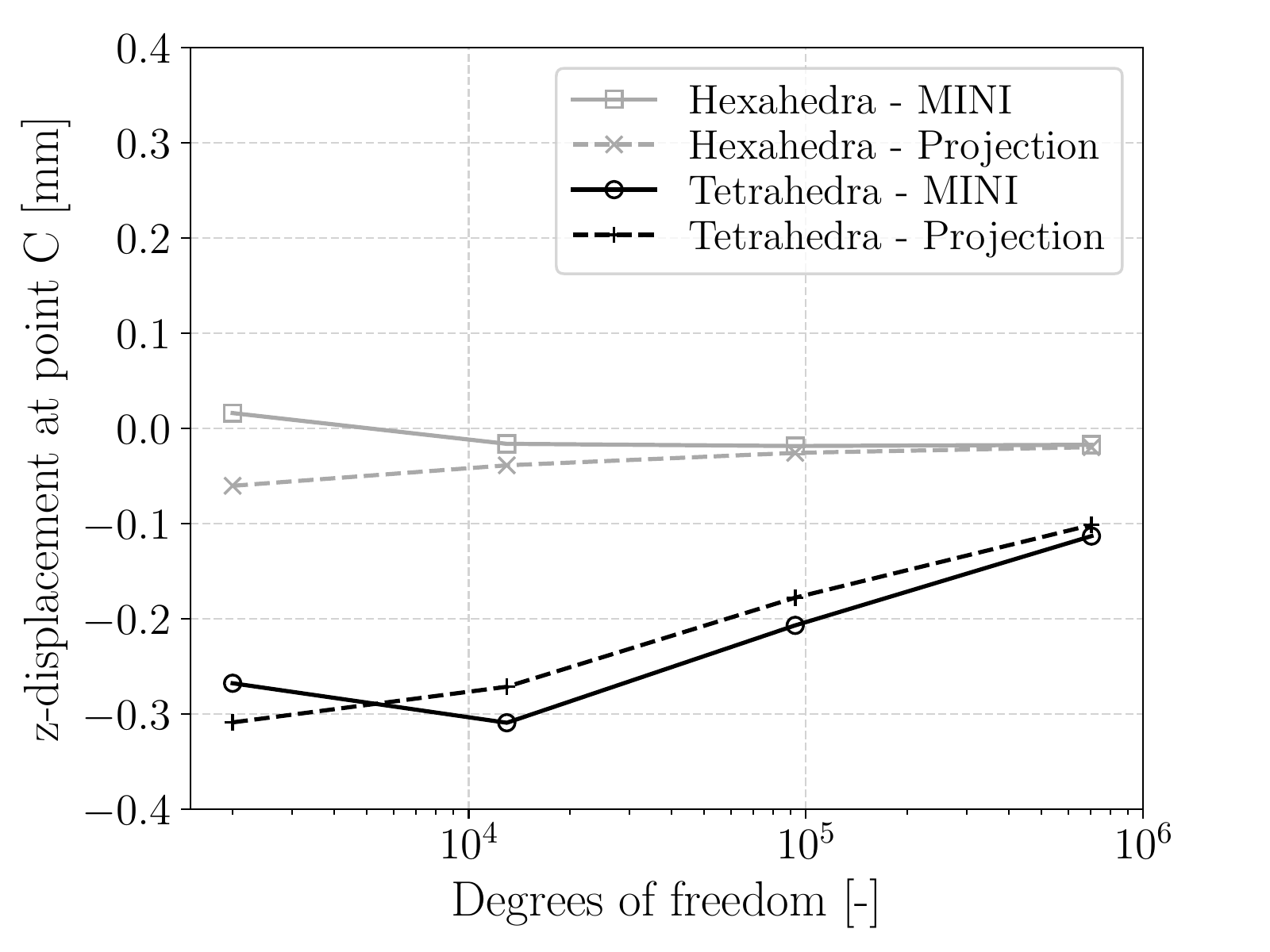}%
  \caption{\emph{Cook-type cantilever problem:}
    displacements $u_x$, $u_y$, and $u_z$ at point \textbf{C} versus
    the number of degrees of freedom in a logarithmic scale using the fully
    incompressible formulation.}%
  \label{fig:cook_displacements}
\end{figure}
\begin{figure}[htpb]
  \textbf{(a)}\\
  \includegraphics[width=0.9\linewidth]{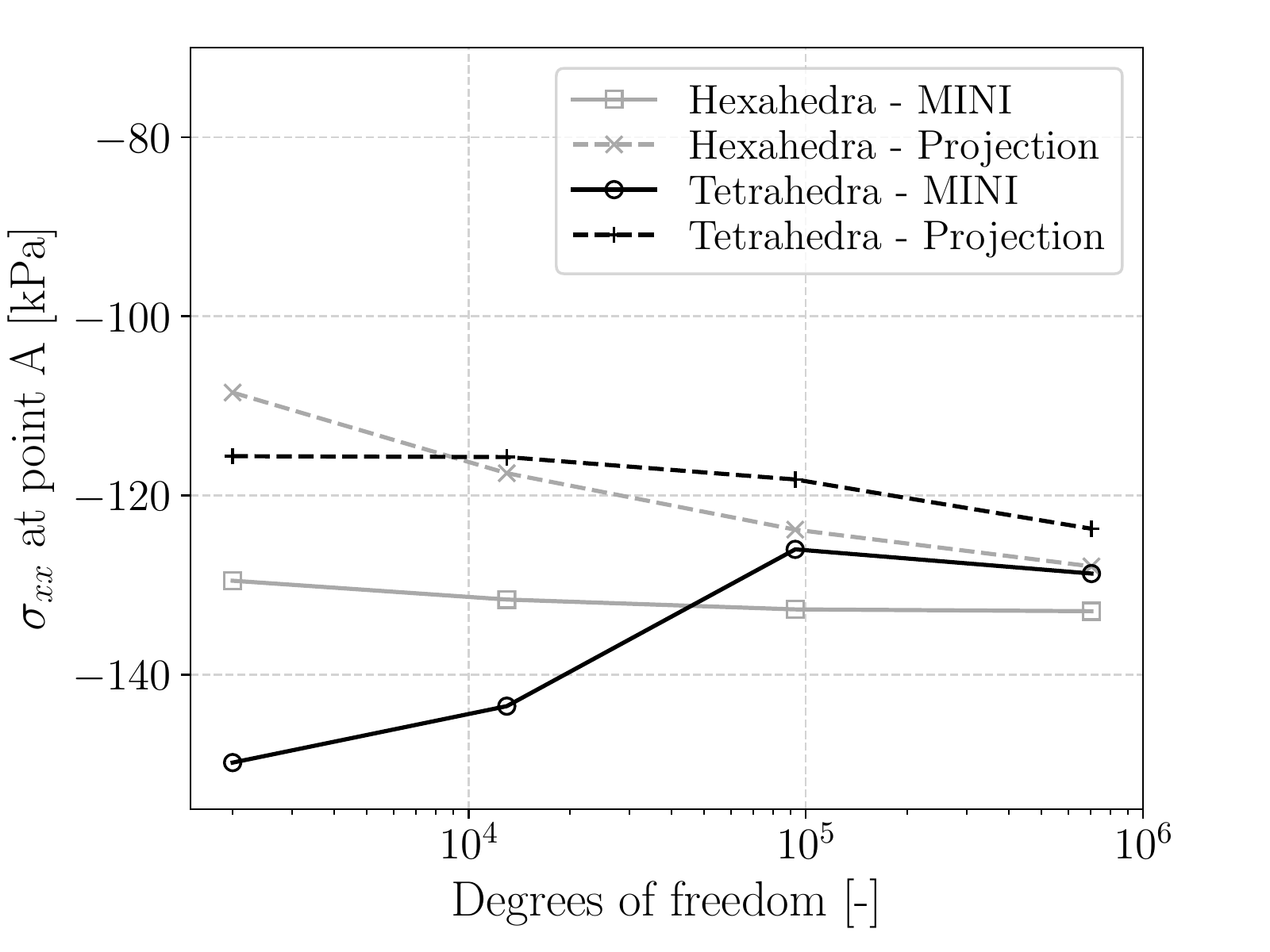}\\
  \textbf{(b)}\\
  \includegraphics[width=0.9\linewidth]{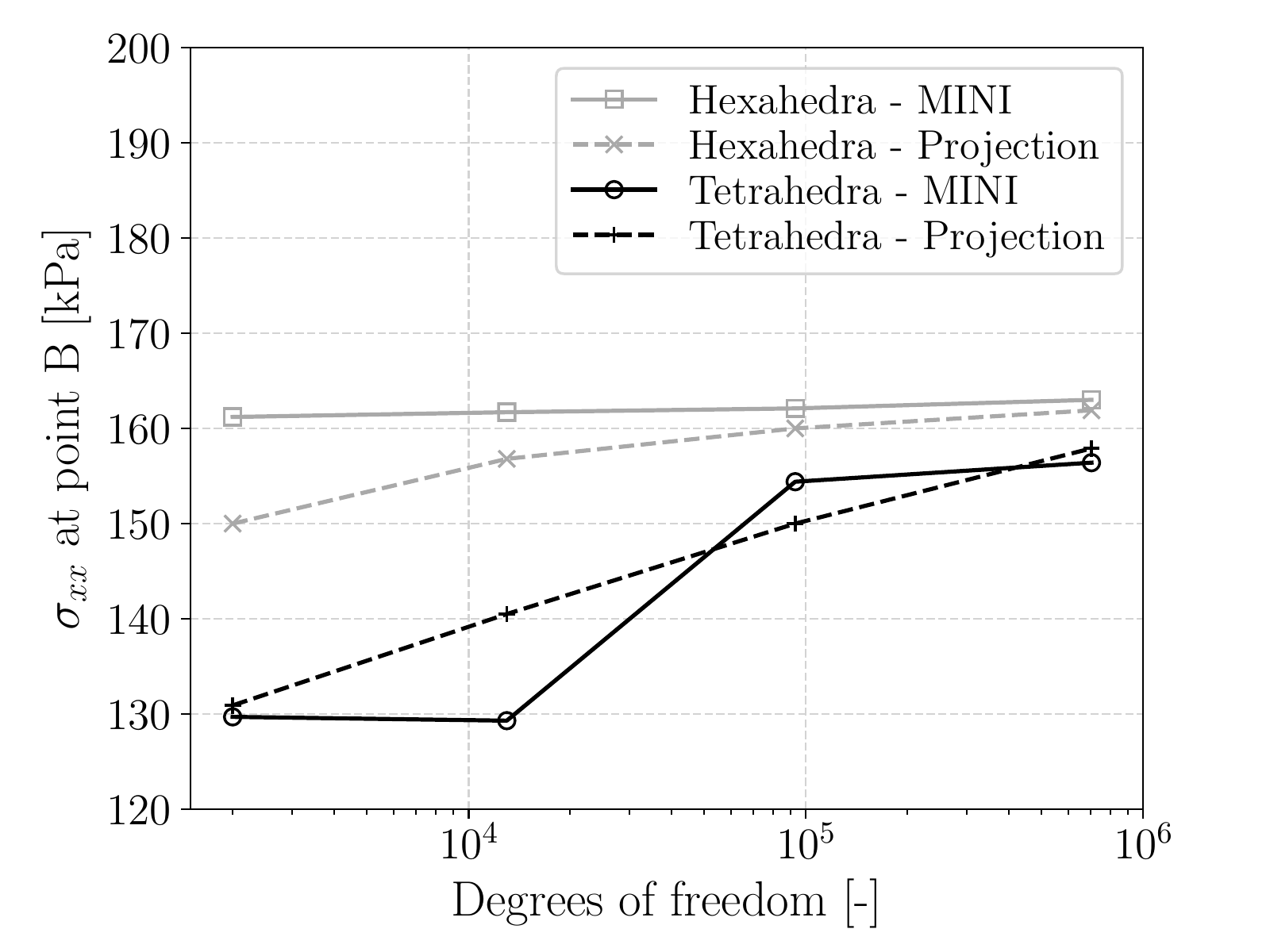}\\
  \textbf{(c)}\\
  \includegraphics[width=0.9\linewidth]{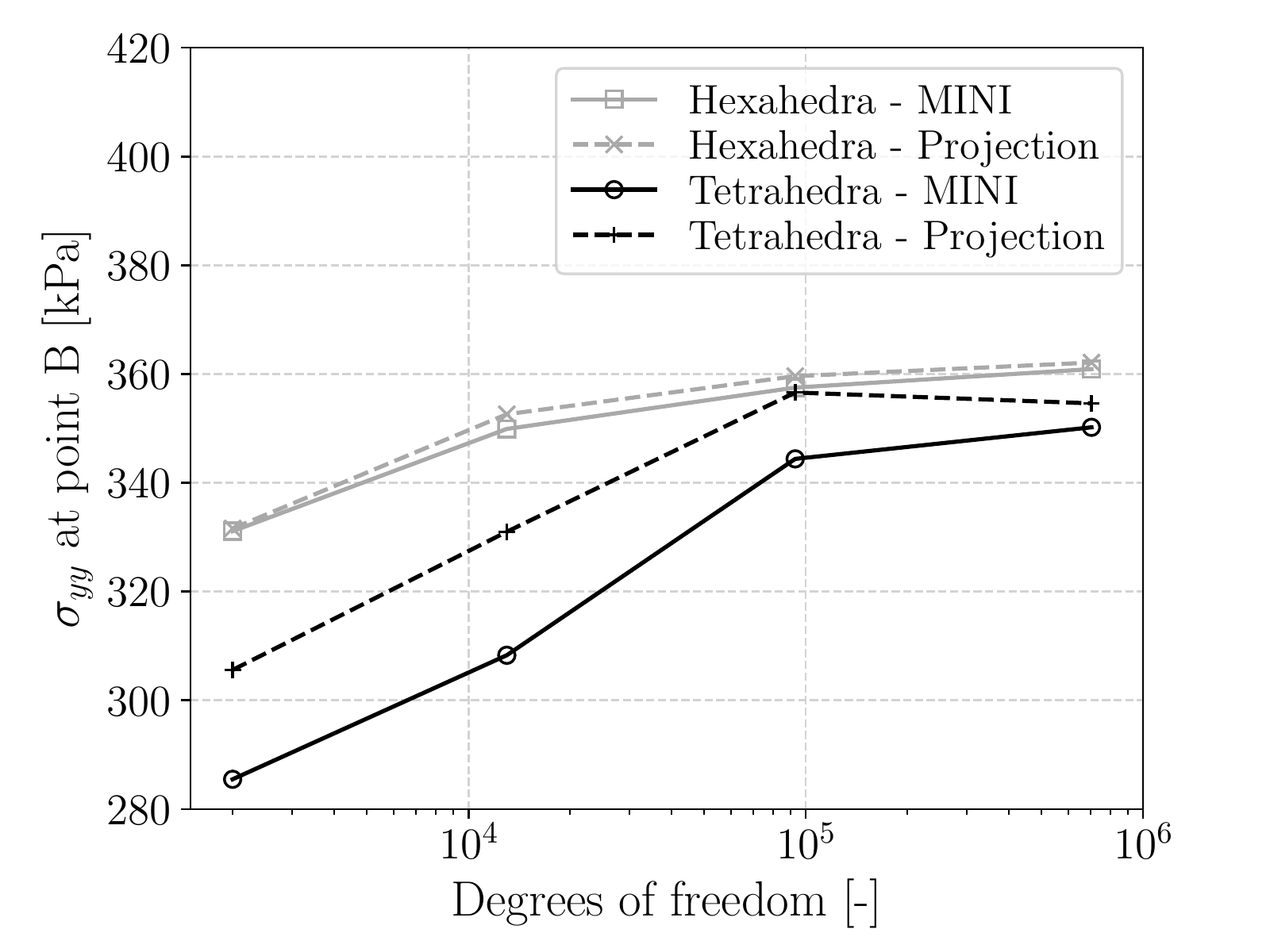}
  \caption{\emph{Cook-type cantilever problem:}
    stresses $\sigma_{xx}$ at (left) point \textbf{A}
    and (middle) point \textbf{B} and $\sigma_{yy}$ at (right) point \textbf{B}
    versus the number of degrees of freedom in a logarithmic scale using the
  fully incompressible formulation.}%
 \label{fig:cook_stresses}
\end{figure}

\begin{figure}[htbp]
  \textbf{(a)}\\
  \includegraphics[width=1.0\linewidth]{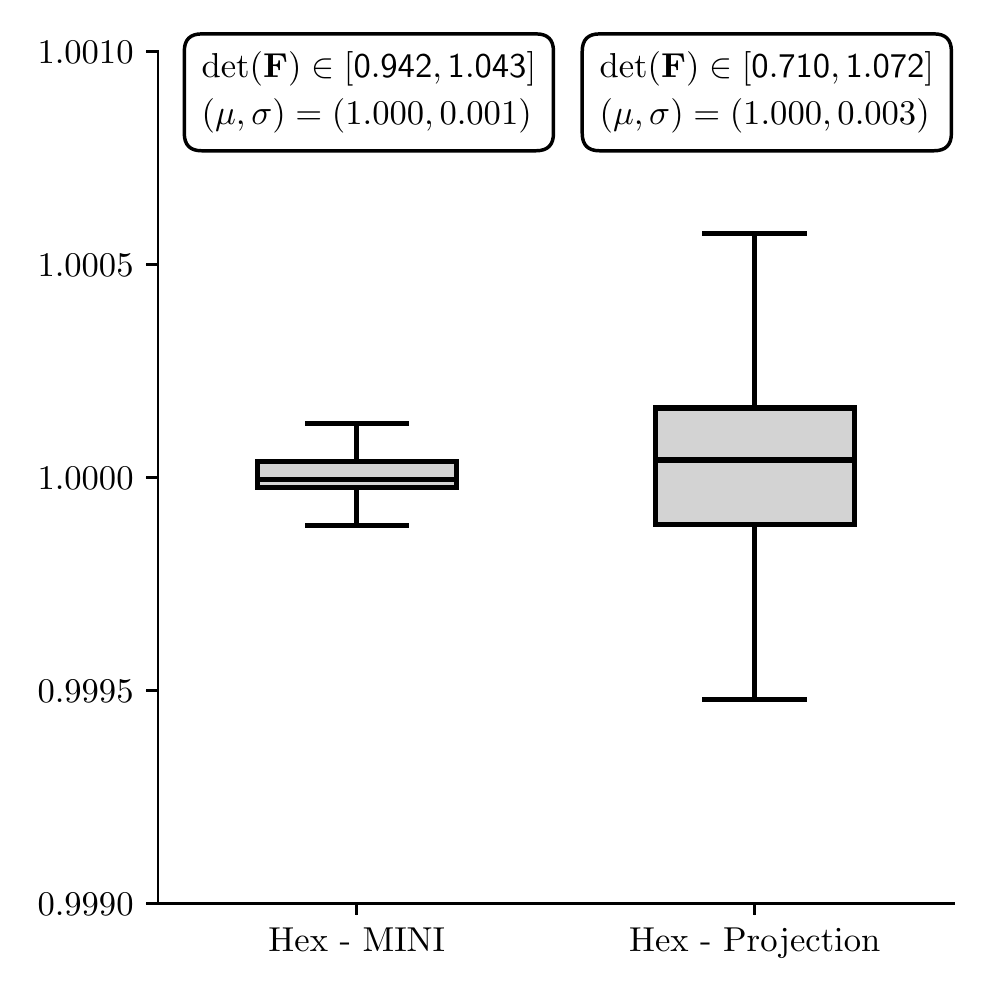}\\
  \textbf{(b)}\\
  \includegraphics[width=1.0\linewidth]{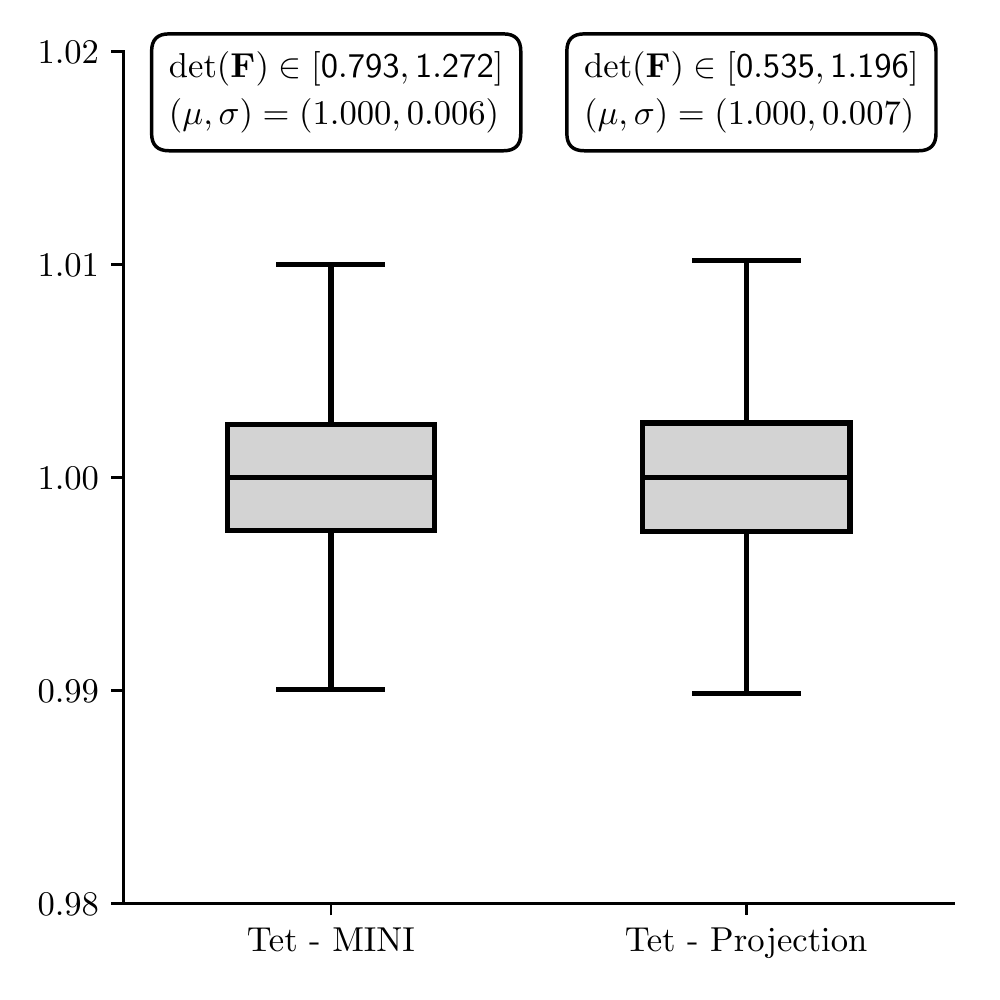}
  \caption{\reviewerOne{\emph{Cook-type cantilever problem:}
    boxplots showing the distribution of $J=\det(\tensor{F})$ for \textbf{(a)}
    hexahedral and \textbf{(b)} tetrahedral elements.
    Additionally, the minimal and maximal value, as well as the mean ($\mu$)
    and the standard deviation ($\sigma$) is given for each setting.}}%
 \label{fig:cook_boxplots}
\end{figure}
\begin{figure*}[htbp]
  \centering
  \includegraphics[width=1.0\linewidth]{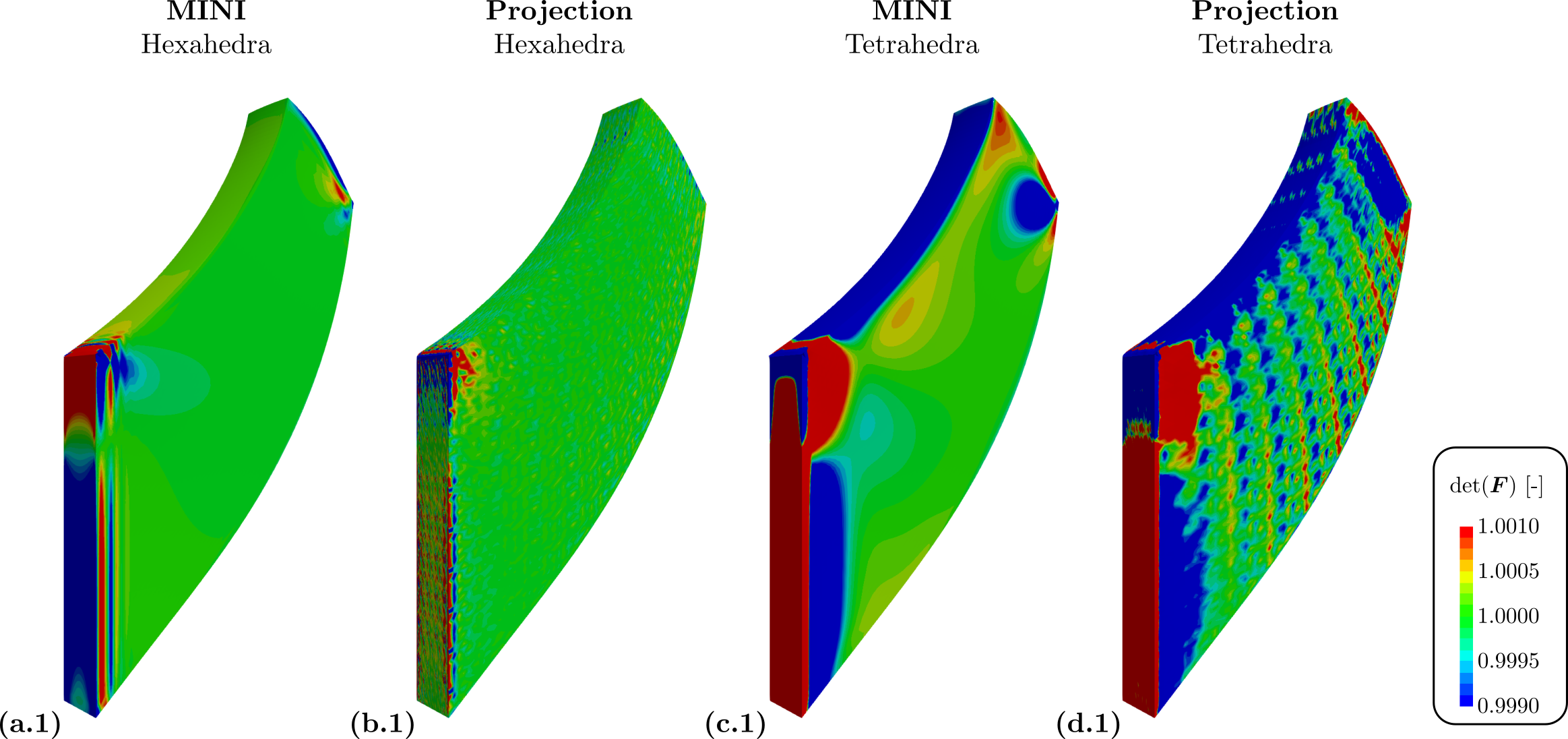}
  \\[1em]
  \includegraphics[width=1.0\linewidth]{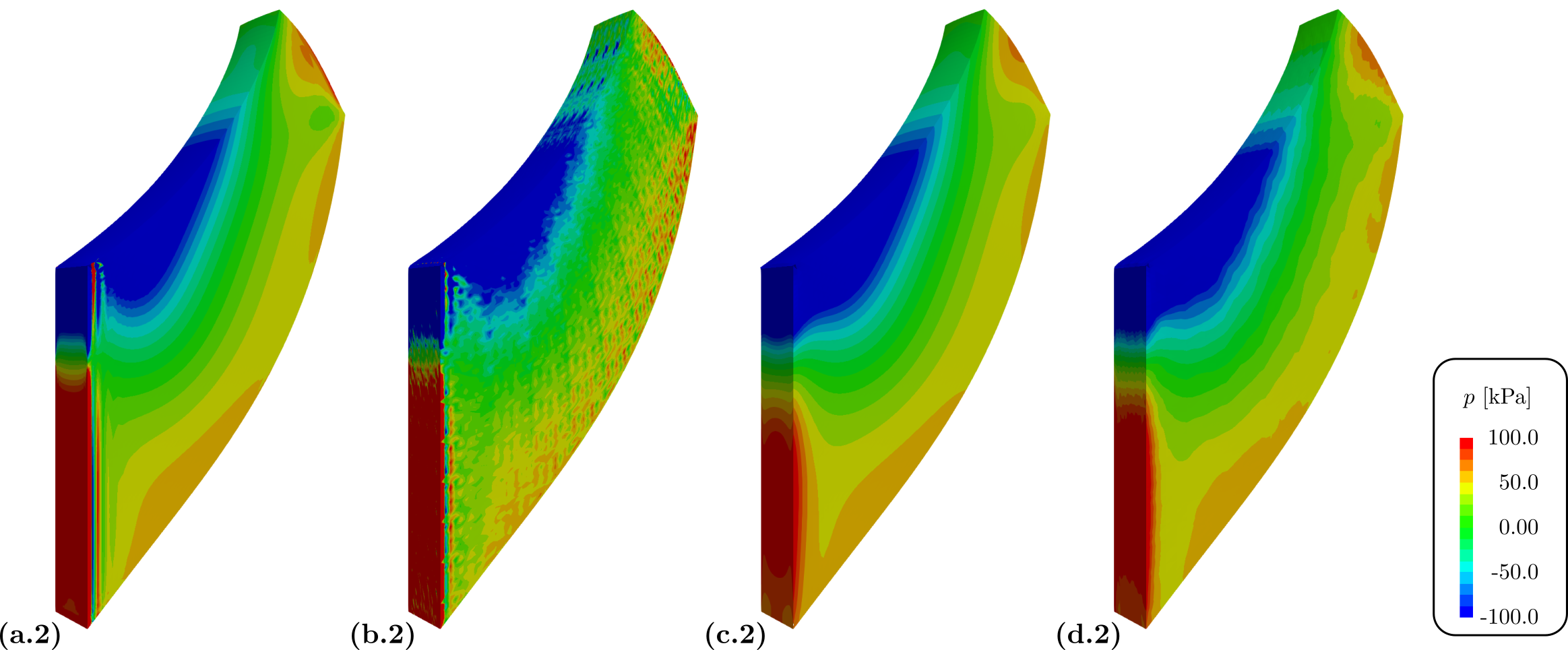}
  \includegraphics[width=1.0\linewidth]{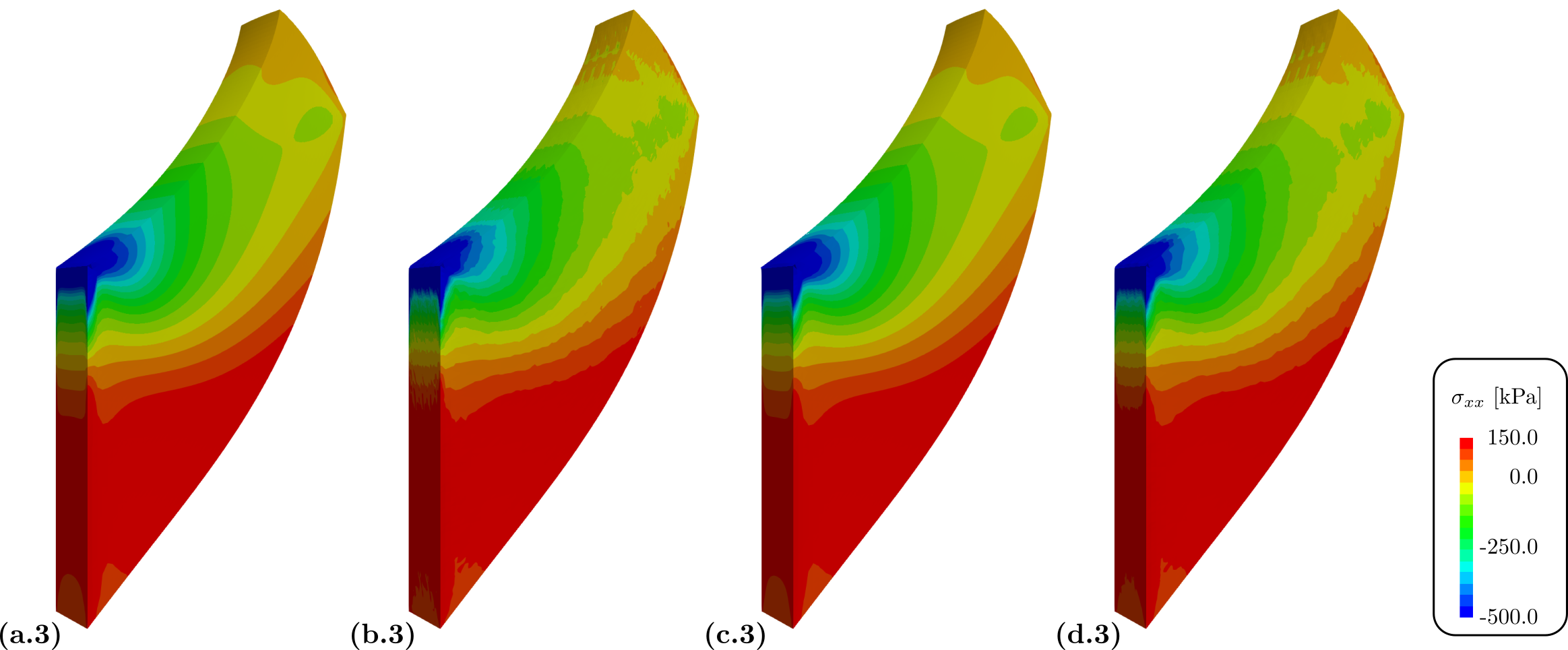}
  \caption{\emph{Cook-type cantilever problem:}
    comparison of hexahedral (a,c) and tetrahedral (b,d) elements with
    bubble-based (a,b) and projection-based (c,d) stabilization. Shown is the
    distribution of $J=\det(\tensor{F})$ (first row);
    distribution of the hydrostatic pressure $p$ (second row) in \si{\kPa};
    and the distribution of the stress $\sigma_{xx}$ (third row) in \si{\kPa}
    for the fully incompressible formulation.
 }%
  \label{fig:cook_comparison}
\end{figure*}

\clearpage%
\subsection{Twisting column test}%
\label{sec:twisting_column}
Finally, we show the applicability of our stabilization techniques for the
transient problem of a twisting column~\cite{aguirre2014vertex, gil2014stabilised, scovazzi2016simple}.
The initial configuration of the geometry is depicted in
Figure~\ref{fig:twisting_column}. There is no load prescribed and the column is
restrained against motion at its base.
A twisting motion is applied to the domain by means of the following initial
condition on the velocity
\[
  \vec{v}(\vec{x}, 0)
    = \vec{v}(x,y,z, 0)
    = 100\sin\left(\frac{\pi y}{12}\right){\left(z, 0, -x\right)}^\top \si{\m/\s},
\]
for \(y\in\left[0,6\right]\si{m}\). To avoid symmetries in the problem the
column is rotated about the $z$-axes by an angle of $\theta=\SI{5.2}{\degree}$.
\begin{figure}[htbp]
  \centering
  \includegraphics[width=0.8\linewidth]{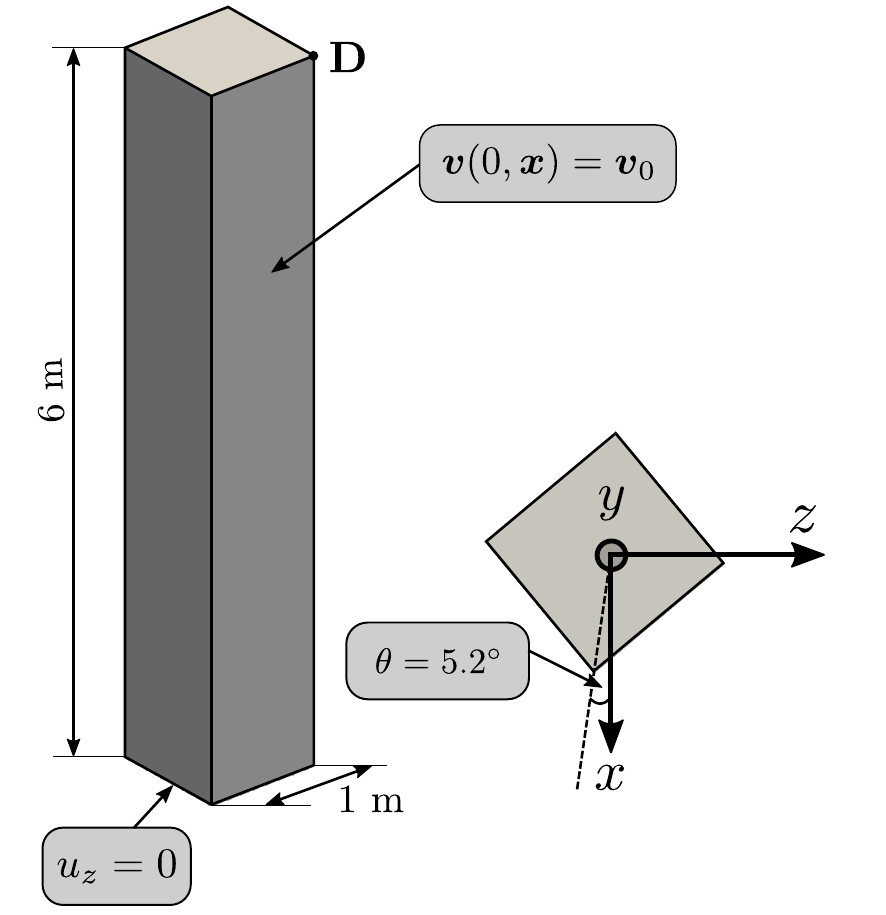}
  \caption{\emph{Twisting column test:} geometry and boundary conditions.}%
  \label{fig:twisting_column}
\end{figure}

We chose the neo-Hookean strain-energy
\[
  \Psi(\tensor C) = \frac{\mu}{2}\left(\mathrm{tr}(\overline{\tensor C}) - 3\right)
 + \frac{\kappa}{2} {(J-1)}^2,
\]
with parameters $\mu=\SI{5704.7}{\kPa}$ and $\kappa=\SI{283333}{kPa}$
for the nearly incompressible and \reviewerTwo{$1/\kappa=0$} for the
truly incompressible case.
For the results presented, we considered hexahedral and tetrahedral meshes with
five levels of refinement, respectively;
for discretization details of the column meshes
see Table~\ref{tab:column_meshes}.
\begin{table}[htbp]
  \caption{Properties of \emph{column meshes} used
           in Section~\ref{sec:twisting_column}.}%
  \label{tab:column_meshes}
  \centering
  \begin{tabular}{rrr}
    \toprule
    \multicolumn{3}{c}{Hexahedral Meshes}\\
    $\ell$ & Elements & Nodes \\
    \midrule
    1 & \num{48}     & \num{117}  \\
    2 & \num{384}    & \num{625}  \\
    3 & \num{3072}   & \num{3969} \\
    4 & \num{24576}  & \num{28033} \\
    5 & \num{196608} & \num{210177} \\
    \bottomrule
  \end{tabular}\hspace{2em}
  \begin{tabular}{rrr}
    \toprule
    \multicolumn{3}{c}{Tetrahedral Meshes}\\
    $\ell$ & Elements & Nodes \\
    \midrule
    1 & \num{240}    & \num{117}  \\
    2 & \num{1920}   & \num{625}  \\
    3 & \num{15360}  & \num{3969} \\
    4 & \num{122880} & \num{28033} \\
    5 & \num{983040} & \num{210177} \\
    \bottomrule
  \end{tabular}
\end{table}
In Figure~\ref{fig:twisting_mesh_convergence},
mesh convergence with respect to tip displacement $(u_x,u_y,u_z)$
at point \textbf{D} is analyzed. While differences at lower levels of
refinement $\ell=1,2$ are severe, the displacements converge for
higher levels of refinement $\ell=3,4,5$. For finer grids the curves for
tetrahedral and hexahedral elements are almost indistinguishable, see also
Figure~\ref{fig:twisting_method_comp}, and the results are in good agreement
with those presented in~\cite{scovazzi2016simple}. While this figure was
produced using MINI elements we also observed a similar behavior of mesh
convergence for the projection-based stabilization.
In fact, for the finest grid, all the proposed stabilization techniques and
elements gave virtually identical results,
see Figure~\ref{fig:twisting_method_comp}.
Further, as already observed by~\textcite{scovazzi2016simple}, the fully and nearly
incompressible formulations gave almost identical deformations,
see Figure~\ref{fig:twisting_formulation_comp}.

In Figure~\ref{fig:twisted_stress} stress $\sigma_{yy}$ and pressure $p$ contours
are plotted on the deformed configuration for the incompressible case
at time instant $t=\SI{0.3}{\s}$.
Minor pressure oscillations can be observed for tetrahedral elements.
Again, results match well those presented in~\cite[Figure 22]{scovazzi2016simple}.

Finally, in Figure~\ref{fig:twisted_velocity}, we compare the magnitude of
velocity and acceleration at time instant $t=\SI{0.3}{\s}$. Results for these
variables are very smooth and hardly distinguishable for all the different
approaches.

The \emph{computational costs} for this nonlinear elasticity problem
were significant due to the required solution of a saddle-point problem in each
Newton step and a large number of time steps.
However, this challenge can be addressed by using a massively parallel
iterative solving method and exploiting potential of modern HPC hardware.
The most expensive simulations were the fully incompressible cases for the
finest grids with a total of \num{840708} degrees of freedom and \num{400} time
steps. These computations were executed at the national HPC computing facility
ARCHER in the United Kingdom using \num{96} cores.
Computational times were as follows:
\SI{239}{\min} for tetrahedral meshes and projection-based stabilization;
\SI{283}{\min} for tetrahedral meshes and MINI elements;
\SI{449}{\min} for hexahedral meshes and projection-based stabilization; and
\SI{752.5}{\min} for hexahedral meshes and MINI elements.
Simulation times for nearly incompressible problems were lower, ranging from
\num{177} to \SI{492}{\min}. This is due to the additional matrix on the
lower-right side of the block stiffness matrix which led to a smaller number of
linear iterations. Simulations with hexahedral meshes were, in general,
computationally more expensive compared to simulations with tetrahedral grids;
the reason beeing mainly a higher number of linear iterations.
Computational burden for MINI elements was larger due to higher
matrix assembly times. However, this assembly time is highly scalable as there
is almost no communication cost involved in this process.
\begin{figure*}[htbp]
  \textbf{(A)} \emph{Mesh convergence for hexahedral elements:}\\
  \includegraphics[width=0.33\linewidth]{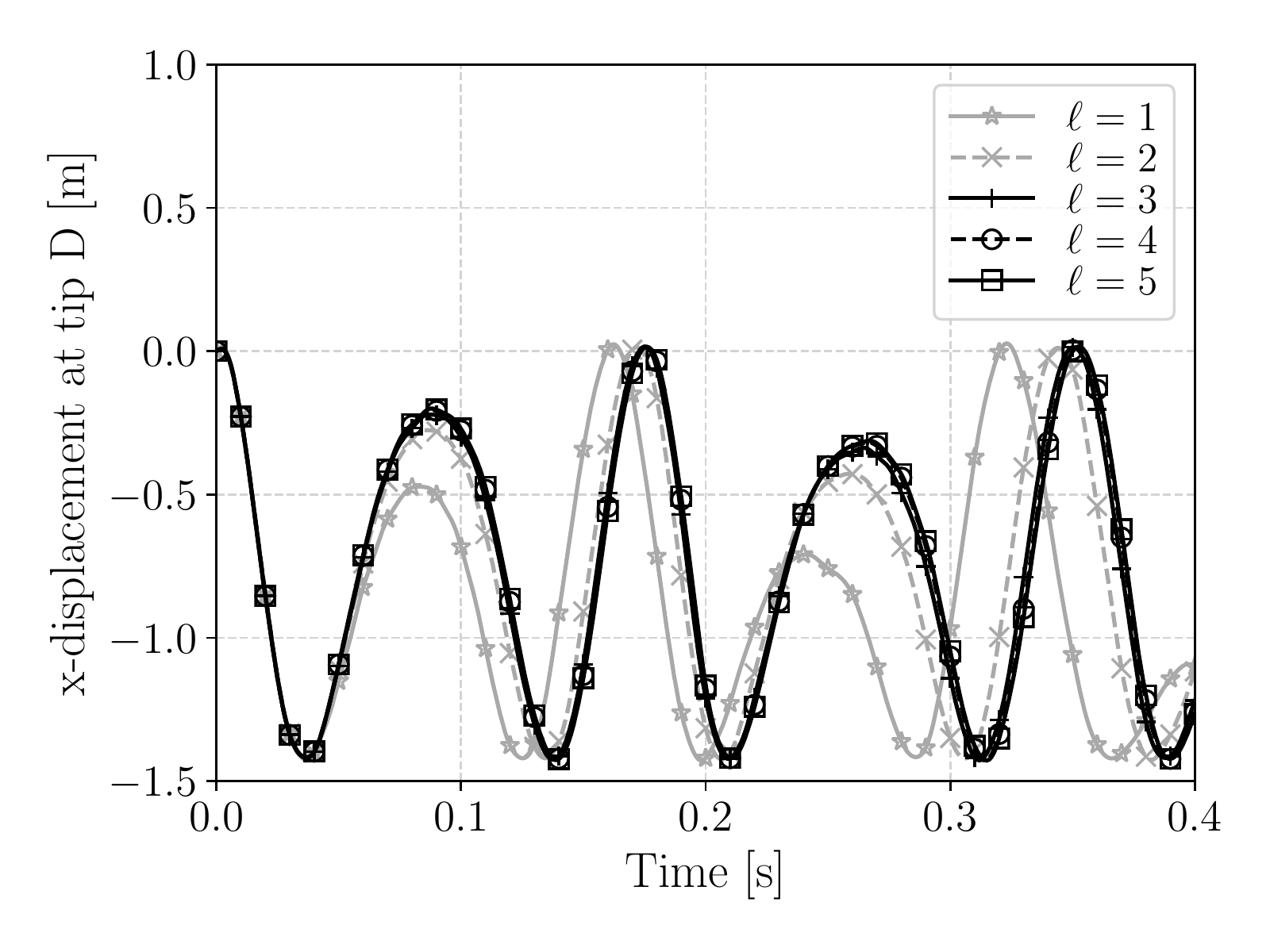}
  \includegraphics[width=0.33\linewidth]{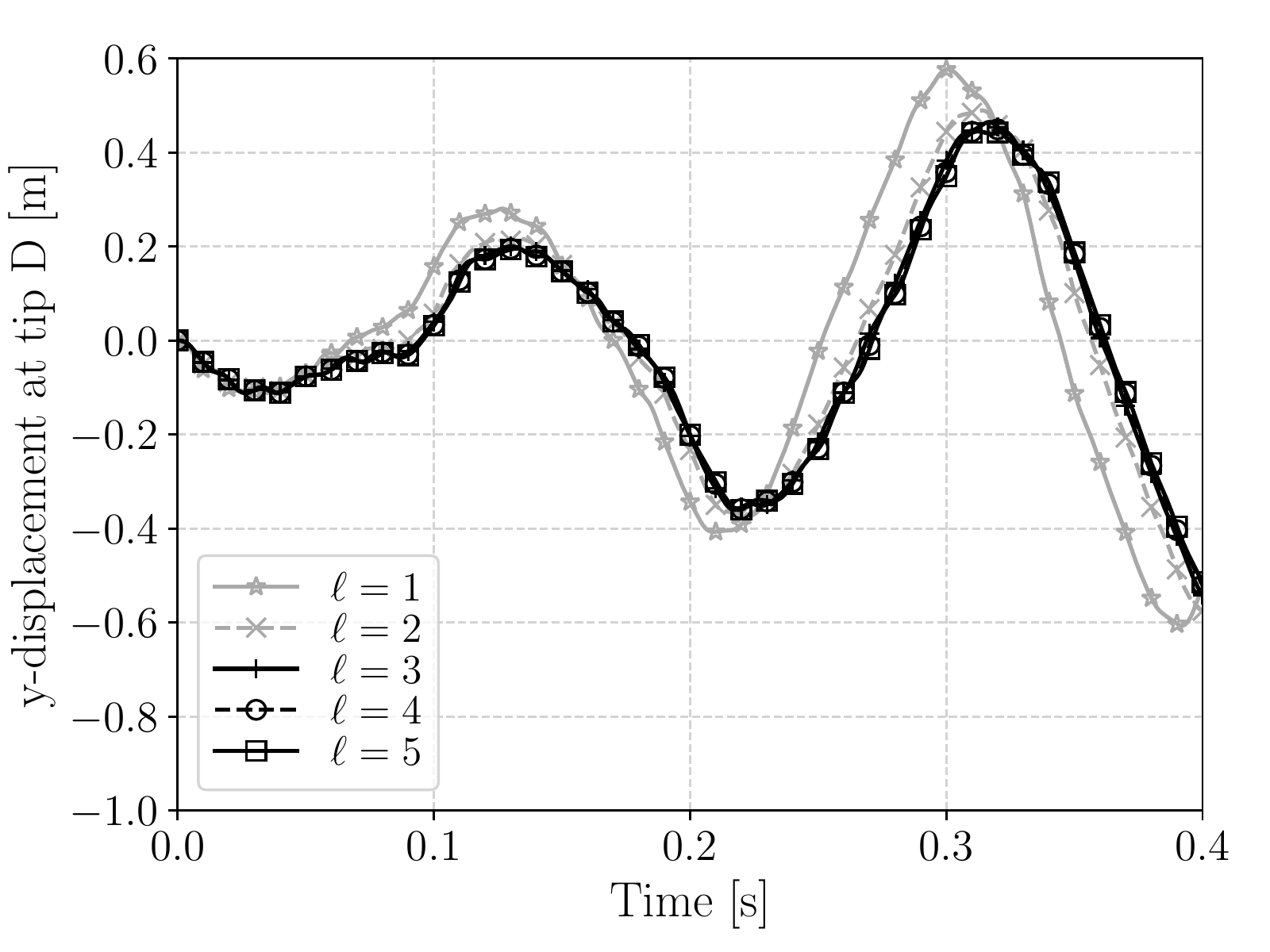}
  \includegraphics[width=0.33\linewidth]{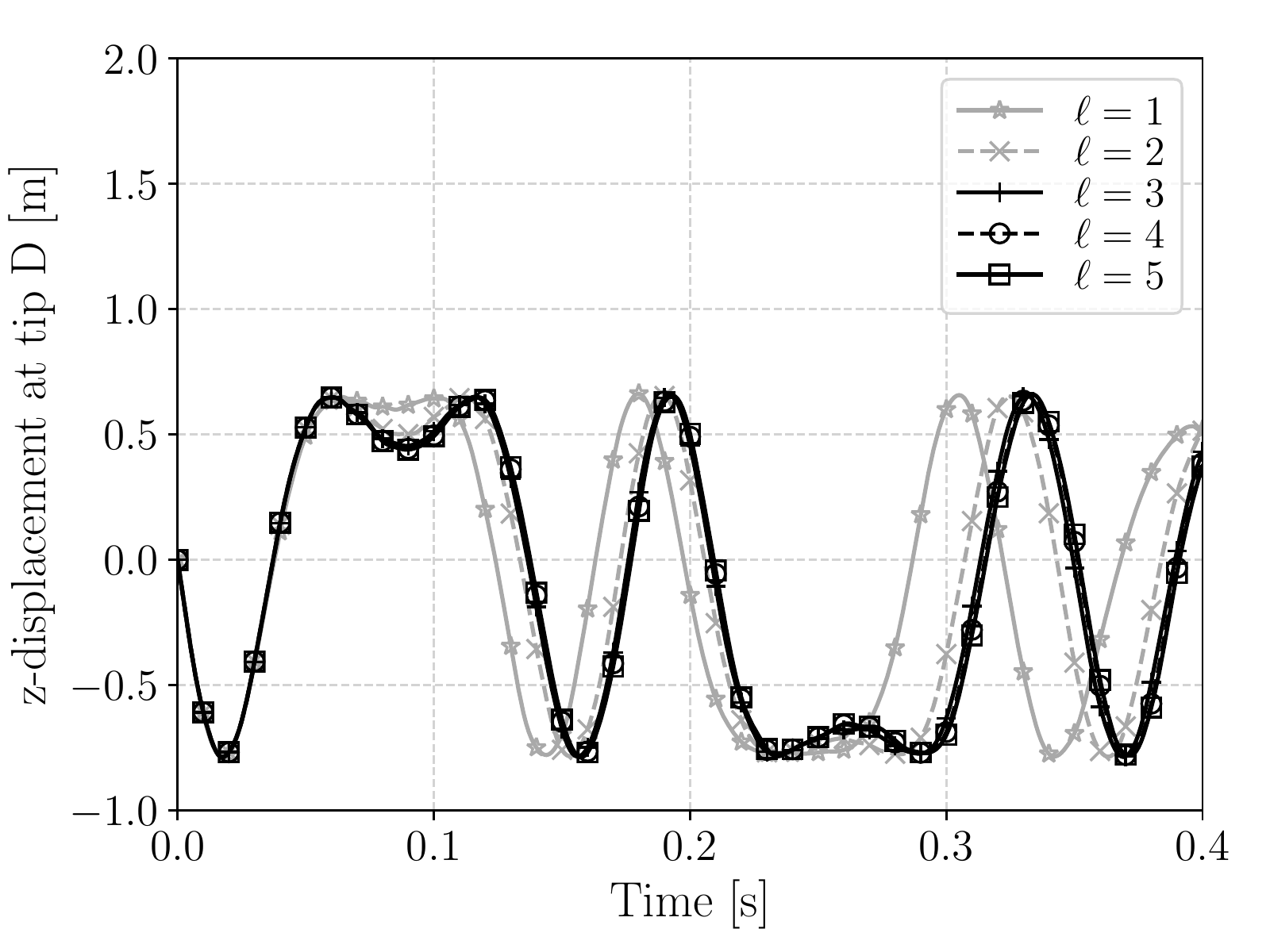}\\
  \textbf{(B)} \emph{Mesh convergence for tetrahedrnts:}\\
  \includegraphics[width=0.33\linewidth]{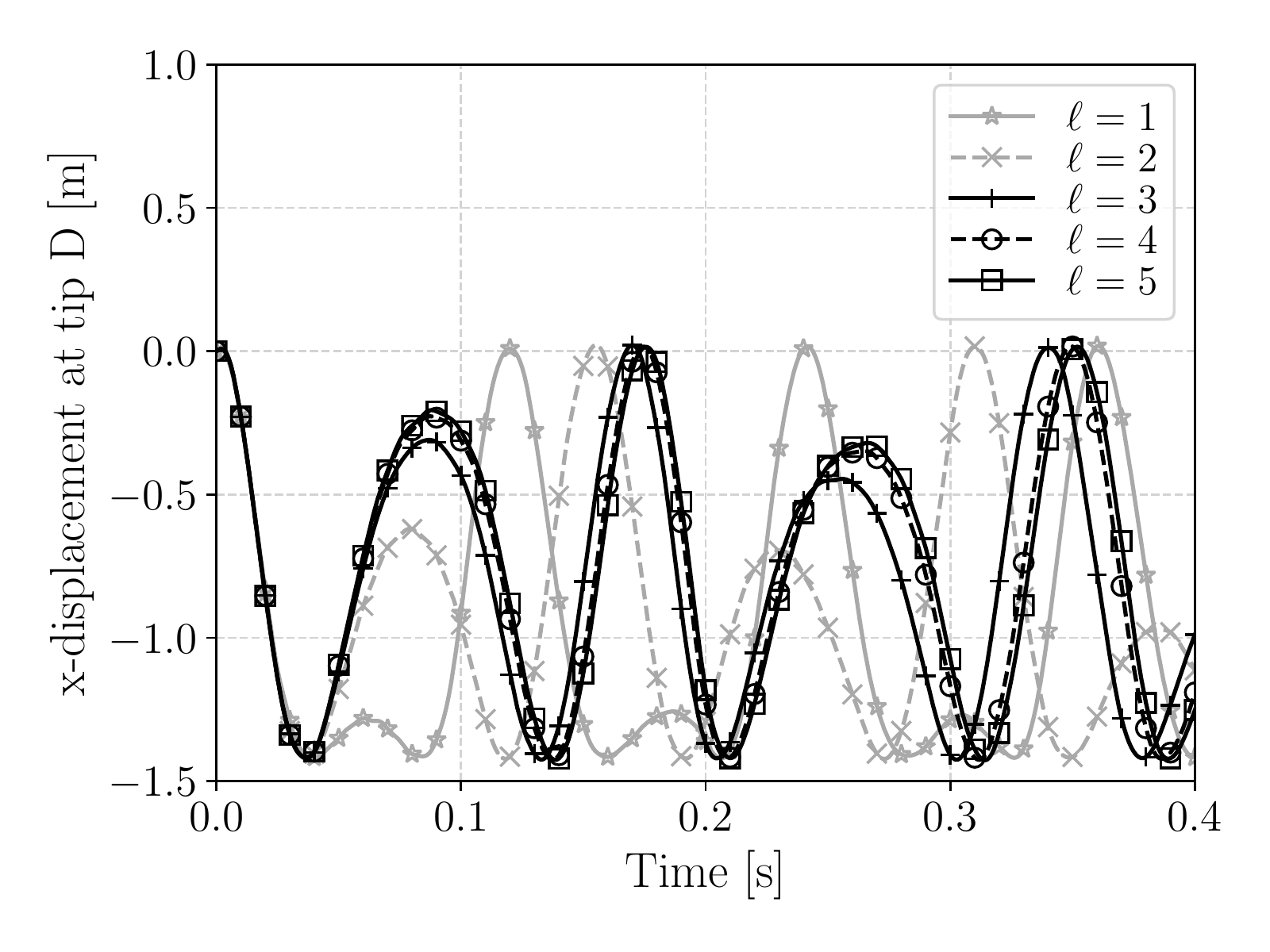}
  \includegraphics[width=0.33\linewidth]{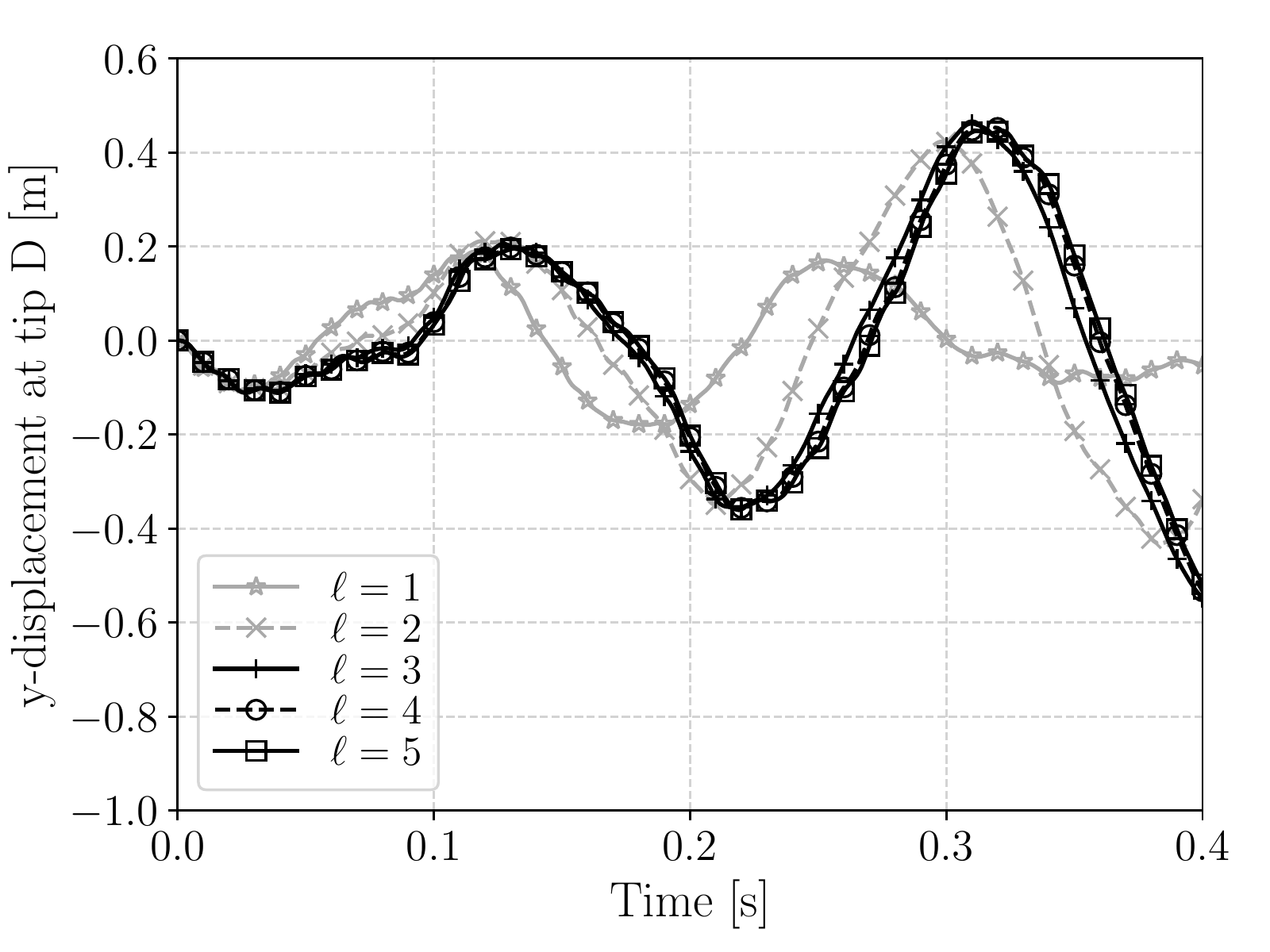}
  \includegraphics[width=0.33\linewidth]{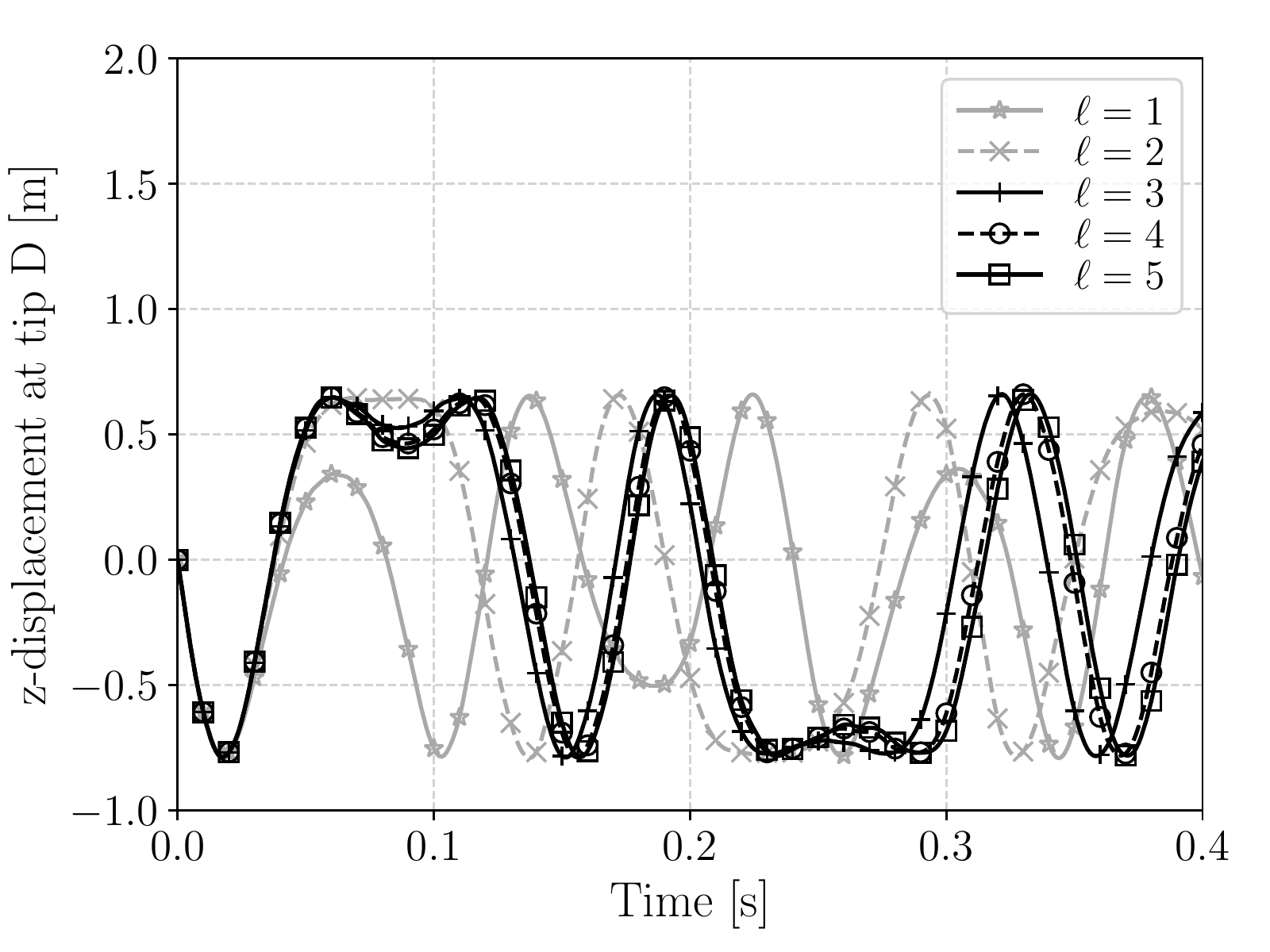}
  \caption{\emph{Twisting column test:} mesh convergence for
    \textbf{(A)} hexahedral and \textbf{(B)} tetrahedral elements.
    Shown are displacements $u_x$, $u_y$, and $u_z$ at tip \textbf{D}
    versus time.
    For experiments depicted the incompressible formulation with MINI elements
    was chosen. At finer levels of refinement $\ell=3,4,5$ (in black)
    results converge to a solution for each displacement direction.}%
  \label{fig:twisting_mesh_convergence}
\end{figure*}
\begin{figure*}[htbp]
  \includegraphics[width=0.33\linewidth]{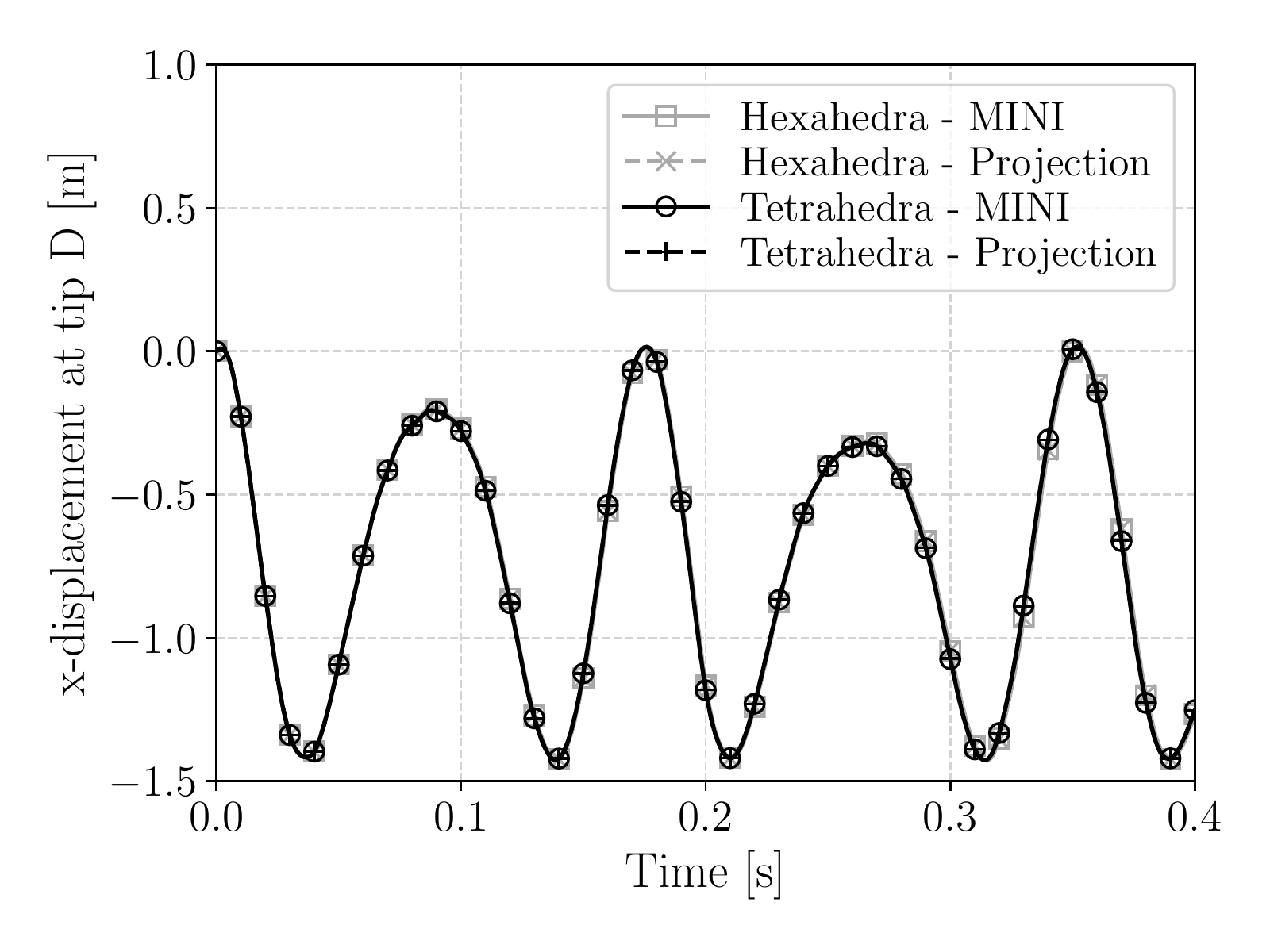}
  \includegraphics[width=0.33\linewidth]{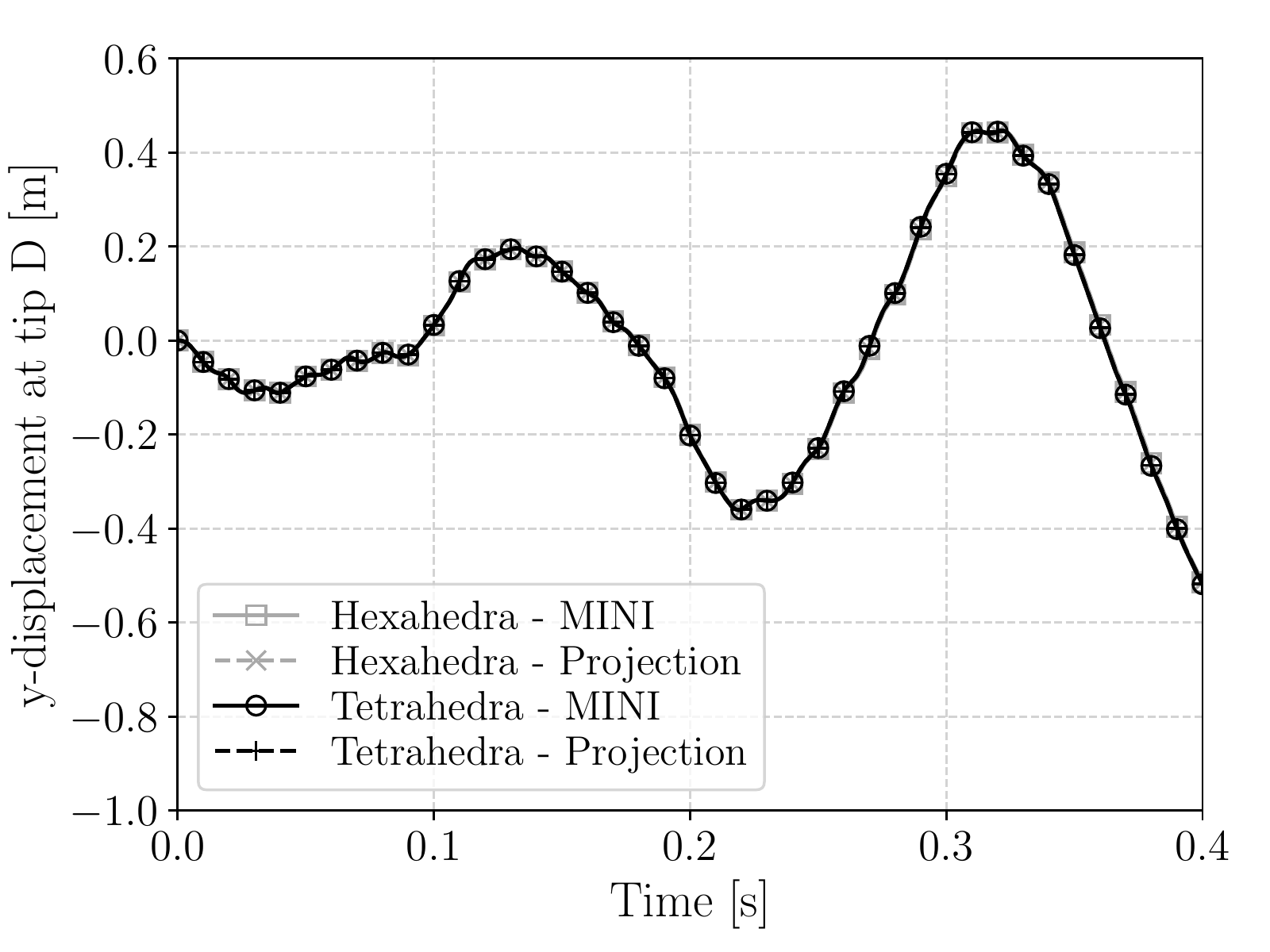}
  \includegraphics[width=0.33\linewidth]{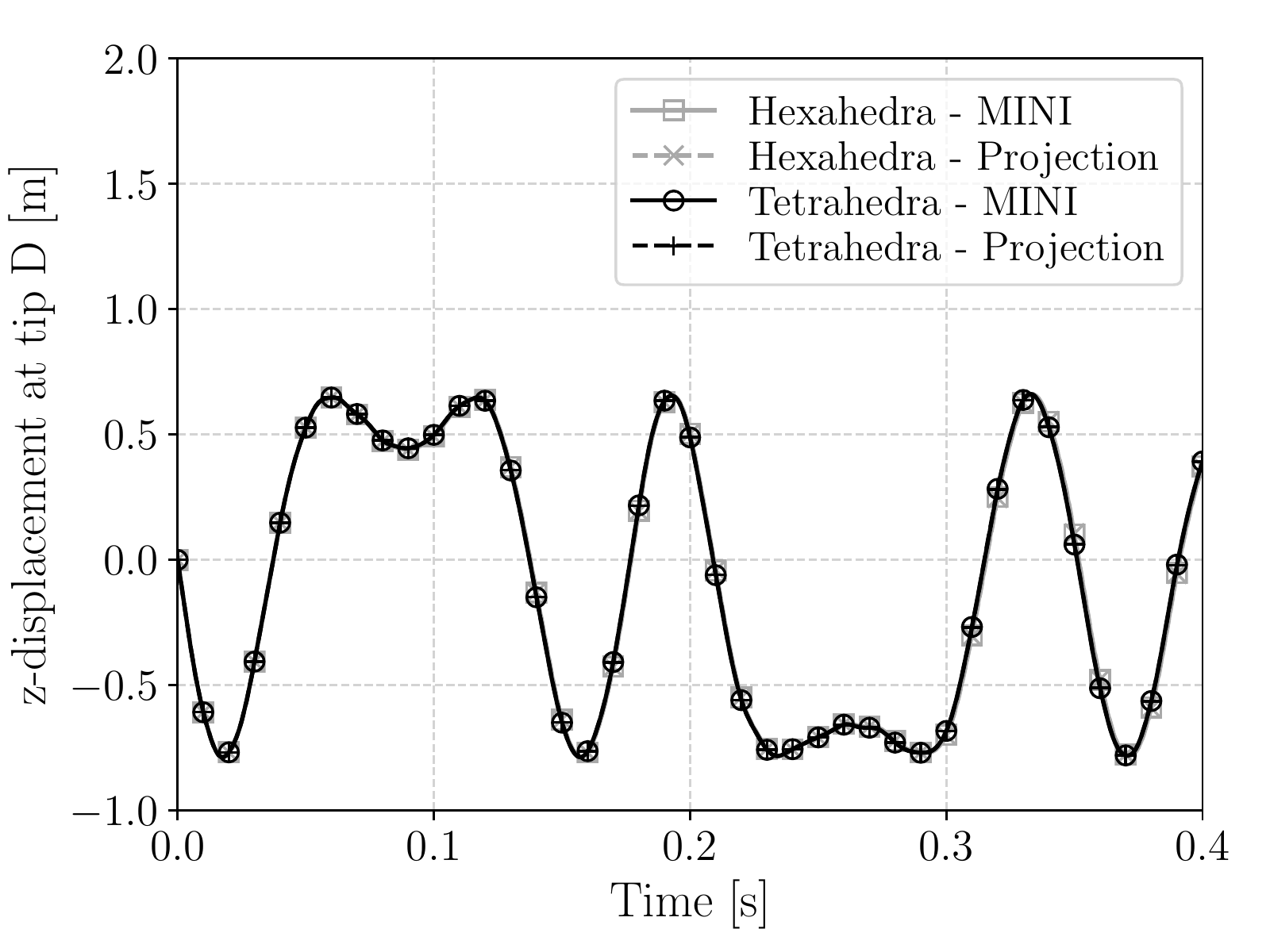}
  \caption{\emph{Twisting column test:}
    comparison of stabilization techniques \reviewerTwo{for the finest grids ($\ell=5$).}
    Shown are displacements
    $u_x$, $u_y$, and $u_z$ at tip \textbf{D} versus time. Both MINI elements
    (dashed line) and projection-based stabilization (dashed lines) render
    almost identical results for hexahedral (in gray) and tetrahedral elements
    (in black).
  }%
  \label{fig:twisting_method_comp}
\end{figure*}
\begin{figure*}[htbp]
  \includegraphics[width=0.33\linewidth]{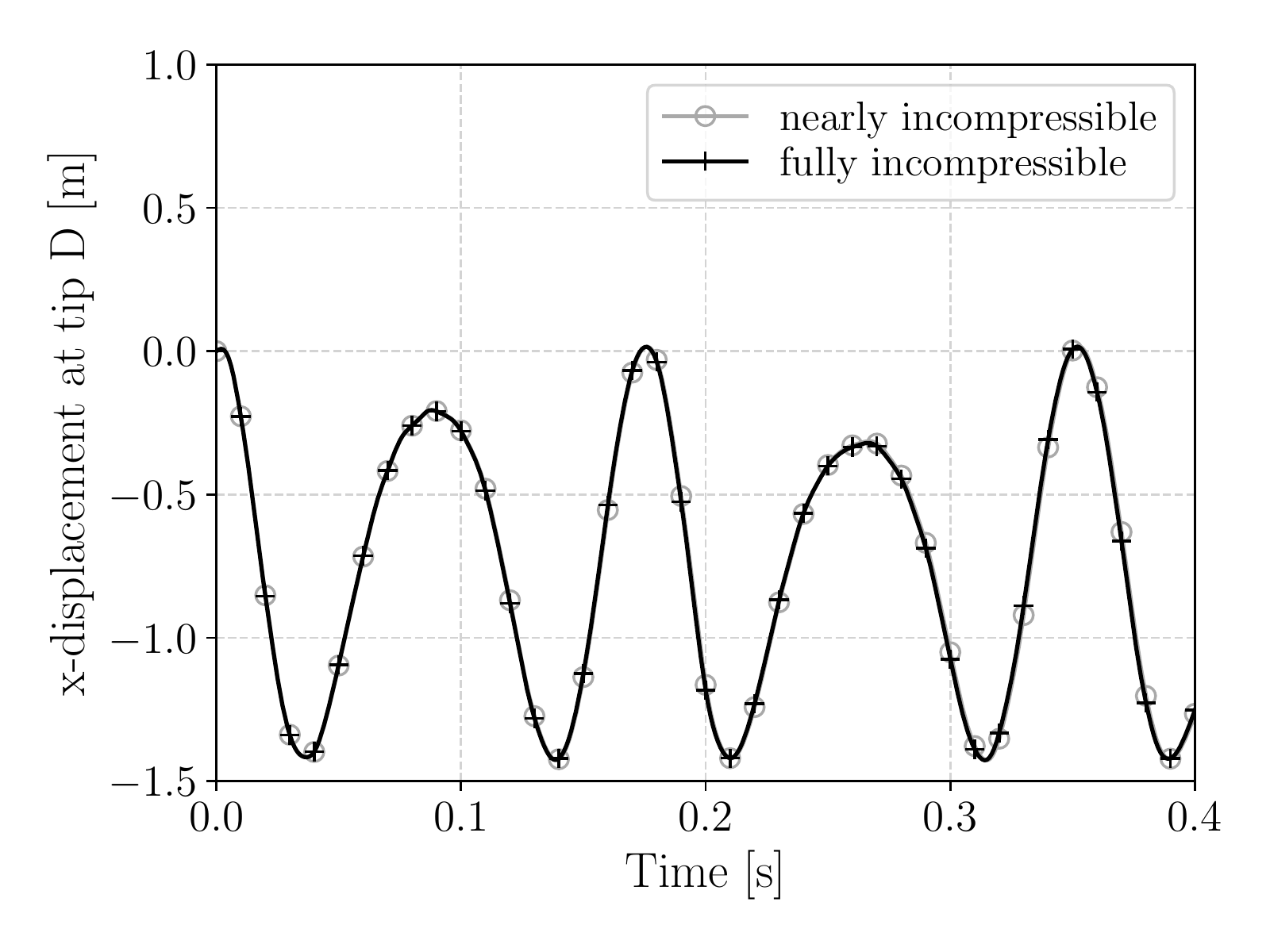}
  \includegraphics[width=0.33\linewidth]{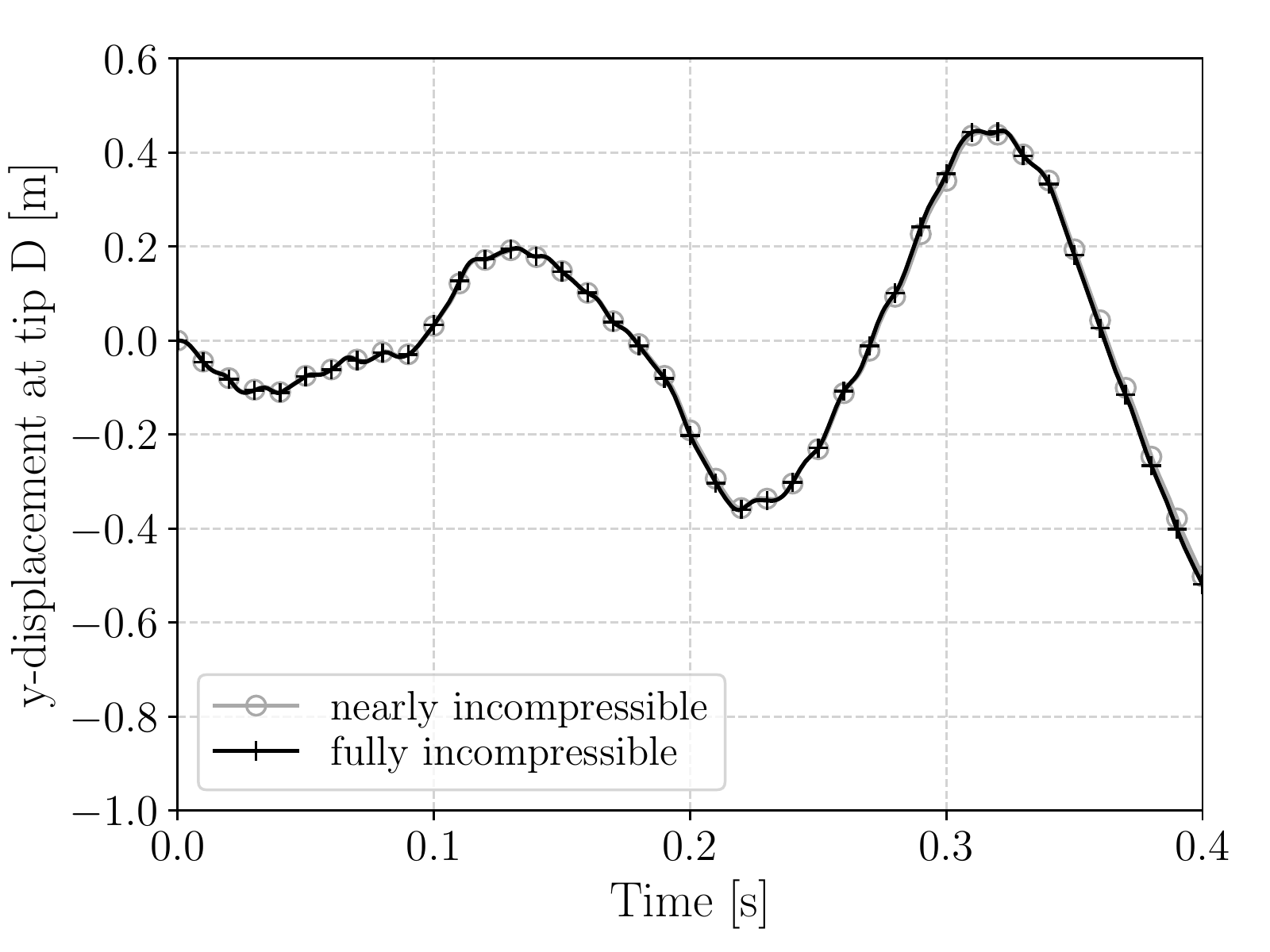}
  \includegraphics[width=0.33\linewidth]{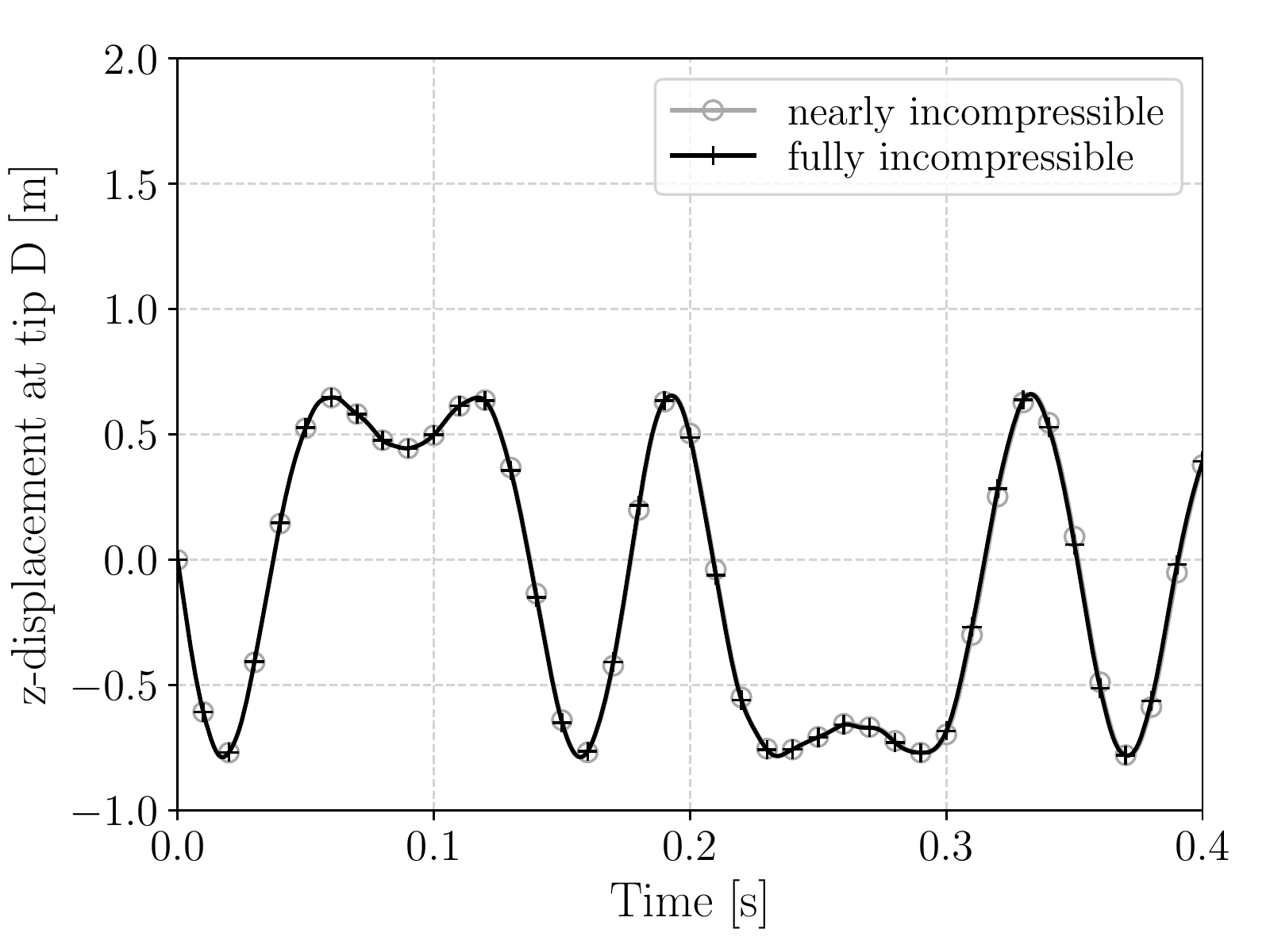}
  \caption{\emph{Twisting column test:}
    comparison of nearly and fully incompressible formulation
    \reviewerTwo{for the finest tetrahedral grids ($\ell=5$) and MINI elements.}
    Displacements $u_x$, $u_y$, and $u_z$ are almost identical for the whole
    simulation duration of \SI{0.4}{\s}.}%
  \label{fig:twisting_formulation_comp}
\end{figure*}
\begin{figure*}[htbp]
  \centering
  \includegraphics[width=\linewidth]{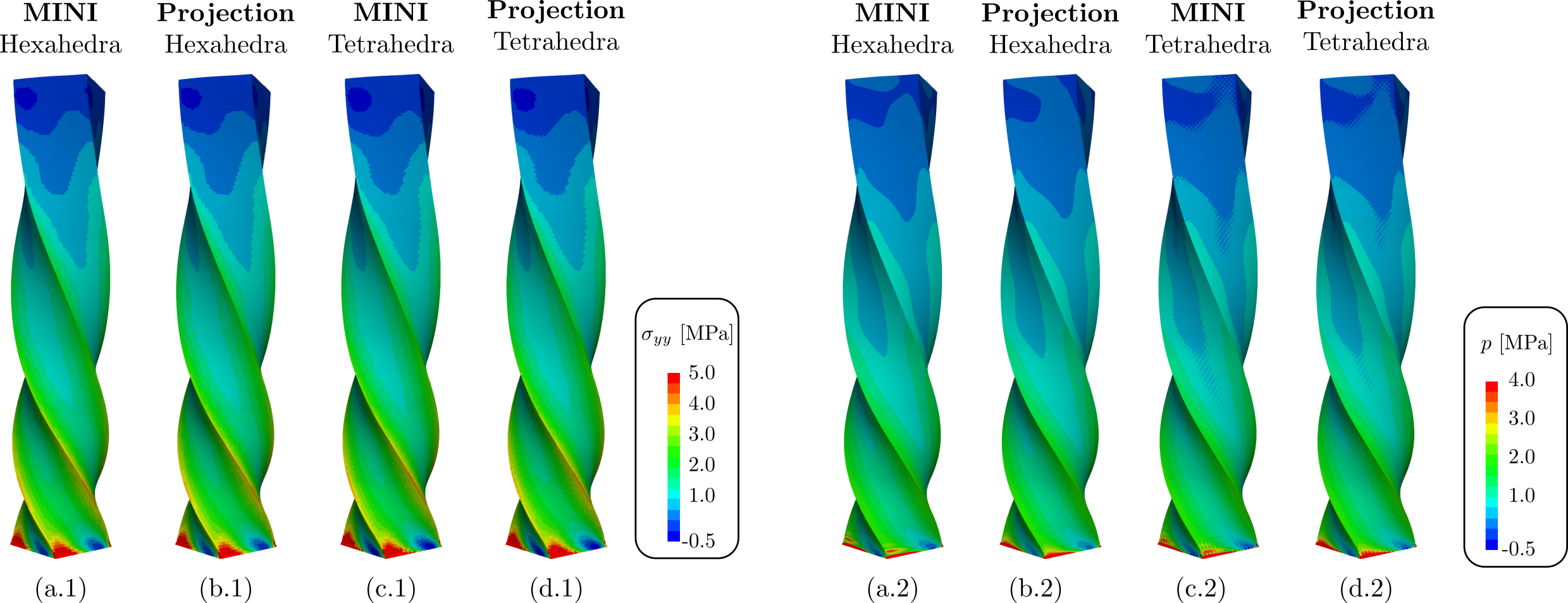}
  \caption{\emph{Twisting column test:}
    (a) stress $\sigma_{yy}$ and (b) hydrostatic pressure $p$ contours
    at time instant $t=\SI{0.3}{\s}$ for the different grids and stabilization
    techniques.}%
  \label{fig:twisted_stress}
\end{figure*}
\begin{figure*}[htbp]
  \centering
  \includegraphics[width=\linewidth]{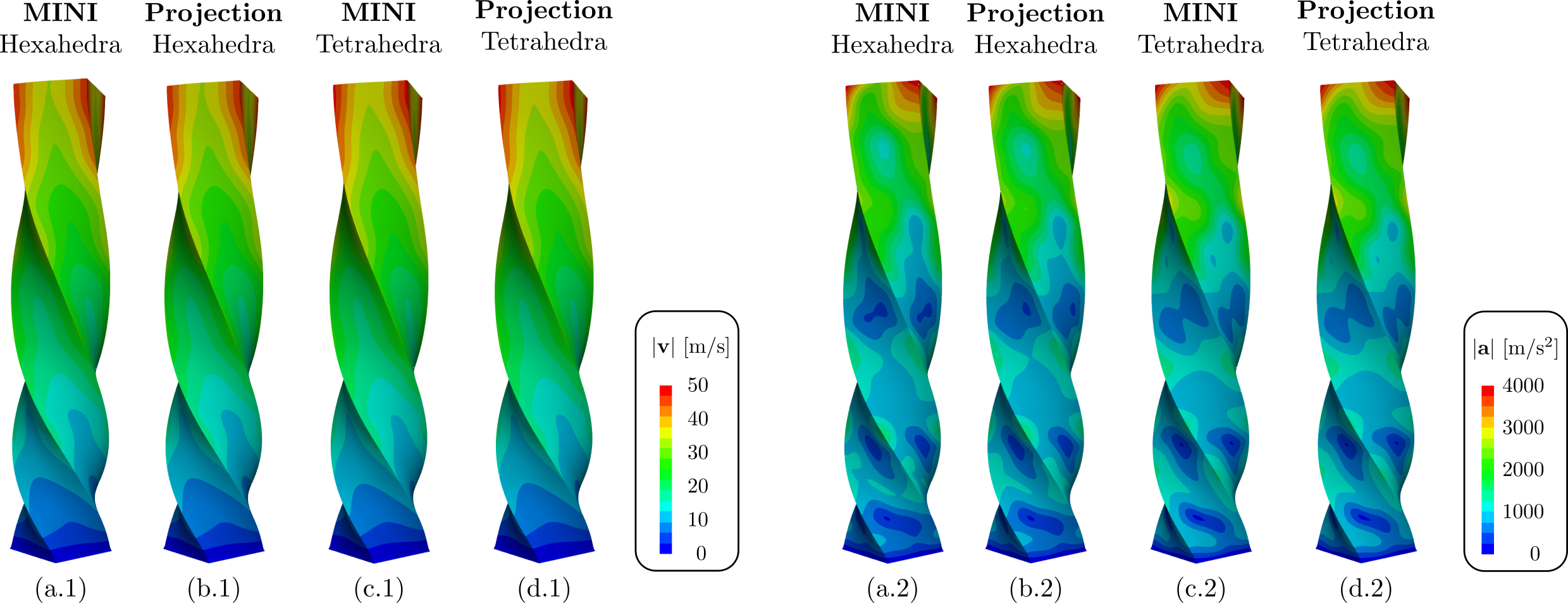}
  \caption{\emph{Twisting column test:}
    magnitude of (a) velocity $\vec{v}$ and (b) acceleration $\vec{a}$
    at time instant $t=\SI{0.3}{\s}$ for the different grids and stabilization
    techniques.}%
  \label{fig:twisted_velocity}
\end{figure*}

\section{Conclusion}%
\label{sec:conclusion}
In this study we described methodology for modeling nearly and fully
incompressible solid mechanics for a large variety of different scenarios.
A stable MINI element was presented which can serve as an
\reviewerOne{excellent choice for applied problems where the use of higher order
  element types is not desired, e.g., due to fitting accuracy of the problem
  domain.}
\reviewerTwo{
We also proposed an easily implementable and computationally cheap technique based on a local
pressure projection.} Both approaches can be applied to stationary as well as
transient problems without modifications and perform excellent with both
hexahedral and tetrahedral grids.
\reviewerTwo{
Both approaches allow a straightforward inclusion in combination with existing
finite element codes since all required implementations are purely on the element
level and are well-suited for simple single-core simulations as well as HPC computing.
Numerical results demonstrate the robustness of the formulations,
exhibiting a great accuracy for selected benchmark problems from the literature.}

While the proposed projection method works well for relatively
stiff materials as considered in this paper,
the setting of the parameter $\mu^\ast$ has to be adjusted for
soft materials such as biological tissues.
A further limitation is that both formulations render the need of solving a
block system, which is computationally more demanding and suitable
preconditioning is not trivial.
However, the MINI element approach can be used without further tweaking of
artificial stabilization coefficients and preliminary results suggested robustness,
even for very soft materials.
Consistent linearization as presented ensures that quadratic convergence of
the Newton--Raphson algorithm was achieved for all the problems considered.
Note that all computations for forming the tangent matrices and also the right
hand side residual vectors are kept local to each element.
This benefits scaling properties of parallel codes and also enables
seamless implementation in standard finite element software.

The excellent performance of the methods along with their high versatility
ensure that this framework serves as a solid platform for simulating
nearly and fully incompressible phenomena in stationary and transient solid
mechanics.
In future studies, we plan to extend the formulation to anisotropic materials
with stiff fibers as they appear for example in the simulation of cardiac
tissue and arterial walls.

\begin{acknowledgement}
  This project has received funding from the European Union's Horizon 2020
  research and innovation programme under the Marie Sk{\l}odowska--Curie
  Action H2020-MSCA-IF-2016 InsiliCardio, GA No. 750835 to CMA.
  Additionally, the research was supported by the grants F3210-N18 and
  I2760-B30 from the Austrian Science Fund (FWF), and a BioTechMed award to GP.
  We acknowledge PRACE for awarding us access to resource ARCHER based in the
  UK at EPCC.
\end{acknowledgement}

\clearpage%
\appendix
\section*{Appendix}
\subsection*{Generalized-$\alpha$ time integration}
After spatial discretization
of~\eqref{eq:trans_nonlin:1}--\eqref{eq:trans_nonlin:2}
we get the following degenerate hyperbolic system
\begin{align*}
  \rho_0\tensor M_h \ddot{\vec u}(t) + \tensor C_h \dot{\vec u}(t)
    + \Xvec R_\mathrm{upper}(\vec u(t), \vec p(t)) &= \vec 0,\\
  \Xvec R_\mathrm{lower}(\vec u(t), \vec p(t)) &= \vec 0,\\
  \vec u(0) &= \vec u_0,\\
  \dot{\vec u}(0) &= \vec u_0,
\end{align*}
where $\tensor M_h$ denotes the mass matrix;
$\tensor C_h$ denotes an optional damping matrix;
$\ddot{\vec u}(t)$ denote the unknown nodal accelerations;
$\dot{\vec u}(t)$ denote the unknown nodal velocities;
$\vec u(t)$ denote the unknown nodal displacements; and
$\vec p(t)$ denote the unknown nodal pressure values.
We will use the modified generalized-$\alpha$ method proposed in~\cite{kadapa2017}.
To this end we introduce the auxiliary velocity $\vec v = \dot{\vec u}$.
Then, applying the standard generalized-$\alpha$ integrator from~\cite{chung1993}
we obtain
\begin{align}
  \label{eq:ga:1}
  \tensor M_h \dot{\vec u}_{n+\alpha_\mathrm{m}}
    - \tensor M_h \vec v_{n+\alpha_\mathrm{f}} &= \vec 0, \\
  \label{eq:ga:2}\rho_0 \tensor{M}_h \dot{\vec v}_{n+\alpha_\mathrm{m}}
    + \tensor C_h \vec v_{n+\alpha_\mathrm{f}}
    + \Xvec R_\mathrm{upper}^{n+\alpha_\mathrm{f}} &= \vec 0,\\
  \label{eq:ga:3}\Xvec R_\mathrm{lower}^{n+\alpha_\mathrm{f}} &= \vec 0,
\end{align}
where
\begin{align*}
 \Xvec R_\mathrm{upper}^{n+\alpha_\mathrm{f}} &:= \alpha_\mathrm{f} \Xvec{R}_\mathrm{upper}(\vec u_{n+1}, \vec p_{n+1}),\\
 &+(1-\alpha_\mathrm{f})\Xvec{R}_\mathrm{upper}(\vec u_{n}, \vec p_{n}),\\
\Xvec R_\mathrm{lower}^{n+\alpha_\mathrm{f}} &:= \alpha_\mathrm{f} \Xvec{R}_\mathrm{lower}(\vec u_{n+1}, \vec p_{n+1}),\\
 &+(1-\alpha_\mathrm{f})\Xvec{R}_\mathrm{lower}(\vec d_{n}, \vec p_{n}),
\end{align*}
and
\begin{align}
 \label{eq:ga:4}\dot{\vec u}_{n+\alpha_\mathrm{m}}
   &:= \alpha_\mathrm{m}\dot{\vec u}_{n+1} + (1-\alpha_\mathrm{m}) \dot{\vec u}_n,\\
 \label{eq:ga:5}\dot{\vec v}_{n+\alpha_\mathrm{m}}
   &:= \alpha_\mathrm{m}\dot{\vec v}_{n+1} + (1-\alpha_\mathrm{m}) \dot{\vec v}_n,\\
 \label{eq:ga:6}\vec v_{n+\alpha_\mathrm{f}}
   &:= \alpha_\mathrm{f}\ \vec v_{n+1} + (1-\alpha_\mathrm{f}) \ \vec v_n.
\end{align}
Moreover, we employ Newmark's approximations,~\cite{newmark85method},
\begin{align}
  \label{eq:ga:7}\dot{\vec u}_{n+1}
    &= \frac{1}{\gamma \Delta t}\left(\vec u_{n+1} - \vec u_{n}\right)
    + \frac{\gamma-1}{\gamma}\dot{\vec u}_n,\\
  \label{eq:ga:8}\vec v_{n+1}
    &= \frac{1}{\gamma \Delta t}\left(\vec v_{n+1} - \vec v_n\right)
    + \frac{\gamma-1}{\gamma} \dot{\vec v}_n
\end{align}
Using~\eqref{eq:ga:1} we observe
\begin{align*}
  \dot{\vec u}_{n+\alpha_\mathrm{m}} = \vec v_{n+\alpha_\mathrm{f}}
\end{align*}
and combining this with~\eqref{eq:ga:2}--\eqref{eq:ga:8} we conclude
\begin{align*}
 \vec v_{n+1} &= \frac{\alpha_\mathrm{m}}{\alpha_\mathrm{f} \gamma \Delta t}\left(\vec u_{n+1} - \vec u_n\right) + \frac{\gamma - \alpha_\mathrm{m}}{\gamma \alpha_\mathrm{f}} \dot{\vec u}_n + \frac{\alpha_\mathrm{f}-1}{\alpha_\mathrm{f}} \vec v_n,\\
 \dot{\vec v}_{n+1} &= \frac{\alpha_\mathrm{m}}{\alpha_\mathrm{f} \gamma^2 \Delta t^2}\left(\vec u_{n+1}- \vec u_n\right) - \frac{1}{\alpha_\mathrm{f} \gamma \Delta t} \vec v_n + \frac{\gamma -1}{\gamma} \dot{\vec v}_n\\
 &+\frac{\gamma - \alpha_\mathrm{m}}{\alpha_\mathrm{f} \gamma^2 \Delta t} \dot{\vec u}_n.
\end{align*}
Thus, a dependence of $\vec v_{n+1}$ and $\dot{\vec v}_{n+1}$ on $\vec u_{n+1}$ can be established.
Having this the unknown values $\vec u_{n+1}, \vec p_{n+1}$ can be computed with the Newton--Raphson method.
Based on~\cite{kadapa2017} we set the parameters depending only on
$\rho_\infty \in [0,1)$ by
\begin{align*}
  \alpha_\mathrm{f} &:= \frac{1}{1 + \rho_\infty},\\
  \alpha_\mathrm{m} &:= \frac{3 - \rho_\infty}{2(1 + \rho_\infty)},\\
  \gamma &:= \frac{1}{2}+\alpha_\mathrm{m}-\alpha_\mathrm{f}.
\end{align*}
In all our simulations we used a value of $\rho_\infty = 0.5$.
\reviewerTwo{
\subsection*{Remark on the implementation of the pressure-projection
  stabilized equal order pair}
Considering the bilinear form $s_h(p_h, q_h)$ defined in \eqref{eq:def_db_stabil} we can rewrite this with a simple calculation into
\begin{align*}
 s_h(p_h, q_h) := \sum_{l=1}^{n_\mathrm{el}}\left(\int\limits_{K_l} p_h q_h\dx - \frac{1}{\lvert\tau_l\rvert}\int\limits_{K_l} p_h \dx \int\limits_{K_l}q_h\dx\right).
\end{align*}
Denoting by $\{\phi_i\}_{i=1}^n$ the chosen ansatz functions the element contribution for an arbitrary element $K$ to the matrix $\tensor C_h$ is given by
\begin{align*}
  \int\limits_{K} \phi_i \phi_j \dx - \frac{1}{\lvert K \rvert} \int\limits_{K} \phi_i \dx \int\limits_{K} \phi_j \dx.
\end{align*}
This corresponds to an element mass matrix minus a rank-one correction.
}
\subsection*{Static condensation}
For completeness we provide a summary for the static condensation used for
the MINI element.
Consider a finite element $K \in \mathcal T_h$ with a local ordering of the
unknowns $\vec u$
\begin{multline*}
  \vec u = \left(u_x^1, u_y^1, u_z^1, \ldots, u_x^{\mathrm{ndofs}_N},
           u_y^{\mathrm{ndofs}_\mathrm{N}},
           u_z^{\mathrm{ndofs}_\mathrm{N}}, \right.\\
           \left. u_{x,\mathrm{B}}^1, u_{y,\mathrm{B}}^1, u_{z,\mathrm{B}}^1,
           \ldots,
           u_{x,\mathrm{B}}^{\mathrm{ndofs}_\mathrm{B}},
           u_{y,\mathrm{B}}^{\mathrm{ndofs}_\mathrm{B}},
           u_{z,\mathrm{B}}^{\mathrm{ndofs}_\mathrm{B}}\right)
\end{multline*}
and $\vec p$ as
\begin{align*}
  \vec p = \left(p^1,p^2,\ldots,p^{\mathrm{ndofs}_\mathrm{N}}\right).
\end{align*}
Here, $\mathrm{ndofs}_\mathrm{N}$ corresponds to the nodal degrees of freedom
per element and $\mathrm{ndofs}_\mathrm{B}$ to the bubble degrees of freedom
(one for tetrahedral elements and two for hexahedral elements).
Then the element contribution to the global saddle-point system can be written
as
\begin{align*}
 \begingroup
  \renewcommand*{\arraystretch}{1.5}
  \begin{pmatrix}
   \tensor K_\mathrm{NN} & \tensor K_\mathrm{NB} & \tensor B_\mathrm{N}^\top \\
   \tensor K_\mathrm{BN} & \tensor K_\mathrm{BB} & \tensor B_\mathrm{B}^\top \\
   \tensor B_\mathrm{N} & \tensor B_\mathrm{B} & \tensor C_\mathrm{N}
  \end{pmatrix}
  \begin{pmatrix}
   \Delta \vec u_\mathrm{N}\\
   \Delta \vec u_\mathrm{B}\\
   \Delta \vec p
  \end{pmatrix}
  =
  \begin{pmatrix}
   -\Xvec R_\mathrm{N}^\mathrm{upper}\\
   -\Xvec R_\mathrm{B}^\mathrm{upper}\\
   -\Xvec R_\mathrm{N}^\mathrm{lower}
  \end{pmatrix}
  \endgroup.
\end{align*}
The bubble part of the stiffness matrix,~$\tensor K_\mathrm{BB}$ is local to
the element and can be directly inverted. This gives the condensed system
\begin{align*}
\begingroup
 \renewcommand*{\arraystretch}{1.5}
 \begin{pmatrix}
  \tensor K_\mathrm{eff} & \tensor B_\mathrm{eff}^T \\
  \tensor B_\mathrm{eff} & \tensor C_\mathrm{eff}
 \end{pmatrix}
 \begin{pmatrix}
  \Delta \vec u_\mathrm{N}\\
  \Delta \vec p
 \end{pmatrix}
 =
 \begin{pmatrix}
  -\Xvec R_\mathrm{eff}^\mathrm{upper}\\
  -\Xvec R_\mathrm{eff}^\mathrm{lower}
 \end{pmatrix}
 \endgroup,
\end{align*}
where the effective matrices and vectors are given as
\begin{align*}
  \tensor K_\mathrm{eff}
    &:= \tensor K_\mathrm{NN}
      - \tensor K_\mathrm{NB}\tensor K_\mathrm{BB}^{-1}\tensor K_\mathrm{BN}\\
  \tensor B_\mathrm{eff}
    &:= \tensor B_\mathrm{N}
      - \tensor B_\mathrm{B}\tensor K_\mathrm{BB}^{-1}\tensor K_\mathrm{BN},\\
  \tensor C_\mathrm{eff}
    &:= \tensor C_\mathrm{N}
      - \tensor{B}_\mathrm{B}\tensor K_\mathrm{BB}^{-1}\tensor B_\mathrm{B}^\top,\\
  \Xvec R_\mathrm{eff}^\mathrm{upper}
    &:= \Xvec R_\mathrm{N}^\mathrm{upper}
      - \tensor K_\mathrm{NB} \tensor K_\mathrm{BB}^{-1}
        \Xvec R_\mathrm{B}^\mathrm{upper},\\
  \Xvec R_\mathrm{eff}^\mathrm{lower}
    &:= \Xvec R_\mathrm{N}^\mathrm{lower}
      - \tensor B_\mathrm{B} \tensor K_\mathrm{BB}^{-1}
        \Xvec R_\mathrm{B}^\mathrm{upper}.
\end{align*}
The effective matrices and vectors can then be assembled in a standard way into the global system.
The bubble update contributions can be calculated once $\Delta \vec u_N$ and $\Delta \vec p_N$ are know as
\begin{align*}
 \Delta \vec u_B = -\tensor K_\mathrm{BB}^{-1} \left(\Xvec R_\mathrm{B}^\mathrm{upper} + \tensor K_\mathrm{BN} \Delta \vec u_\mathrm{N} + \tensor B_\mathrm{B}^\top \Delta \vec p_\mathrm{N}\right).
\end{align*}

\subsection*{Tensor calculus}
We use the following results from tensor calculus,
for more details we refer to,
e.g.,~\cite{holzapfel2000nonlinear,wriggers2008nonlinear}.
\begin{align*}
  \frac{\partial \overline{\tensor C}}{\partial \tensor C} &= J^{-\frac{2}{3}} \mathbb P = J^{-\frac{2}{3}} \left(\mathbb I - \frac{1}{3}\tensor C^{-1} \otimes \tensor C\right),\\
\frac{\partial \tensor C^{-1}}{\partial \tensor C} &= - \tensor{C}^{-1} \odot \tensor{C}^{-1},\\
{(\tensor{A}\odot\tensor{A})}_{ijkl}&:= \frac{1}{2}\left(A_{ik}A_{jl}+A_{il}A_{jk}\right).
\end{align*}
For symmetric $\tensor A$ it holds
\begin{align*}
  \mathbb P : \tensor A = \mathrm{Dev}(\tensor A) = \tensor A - \frac{1}{3}(\tensor A : \tensor C) \tensor C^{-1}.
\end{align*}
%
The isochoric part of the second Piola--Kirchhoff stress tensor as well as the
isochoric part of the fourth order elasticity tensor are given as
\begin{align}
  \label{eq:def_Sisc}\tensor{S}_\mathrm{isc}
    &:= 2 \frac{\partial \overline{\Psi}(\overline{\tensor{C}})}{\partial \tensor{C}}
    = J^{-\frac{2}{3}}\mathrm{Dev}(\overline{\tensor{S}}),\\
  \nonumber\overline{\tensor{S}}
    &:= 2 \frac{\partial \overline{\Psi}(\overline{\tensor{C}})}{\partial \overline{\tensor{C}}},\\
  \label{eq:def_c_isc}\mathbb{C}_\mathrm{isc}
    & := 4 \frac{\overline{\Psi}(\overline{\tensor C})}
                {\partial \tensor C \partial \tensor C} \\
  \nonumber
    & = J^{-\frac{4}{3}} \mathbb P \overline{\mathbb C} \mathbb P^\top
      + J^{-\frac{2}{3}} \frac{2}{3}\mathrm{tr}(\tensor C \overline{\tensor S})
      \widetilde{\mathbb P}\\
   \nonumber & - \frac{4}{3}\tensor{S}_\mathrm{isc} \symotimes \tensor{C}^{-1},\\
  \nonumber\overline{\mathbb{C}}
    &:= 4\frac{\partial \overline{\Psi}
      (\overline{\tensor{C}})}{\partial\overline{\tensor C}
      \partial\overline{\tensor{C}}},\\
  \nonumber\widetilde{\mathbb{P}}
    &:= \tensor{C}^{-1} \odot \tensor{C}^{-1}
      - \frac{1}{3} \tensor{C}^{-1} \otimes \tensor{C}^{-1},\\
  \nonumber\tensor{A}\symotimes \tensor{B}
    &:= \frac{1}{2}\left(\tensor{A}\otimes \tensor{B}
    + \tensor{B} \otimes \tensor{A} \right).
\end{align}
\clearpage%
\printbibliography%
\end{document}
